\documentclass[a4paper, onecolumn, 11pt, unpublished]{vendor/quantumarticle/quantumarticle}
\pdfoutput=1

\usepackage{breakurl}
\usepackage{stmaryrd}
\usepackage{amsthm,amstext,amssymb,amsmath}
\usepackage{bm}
\usepackage{epsfig}
\usepackage{slashed}
\usepackage{fancyhdr}
\usepackage{graphicx}
\usepackage{pdfpages}
\usepackage{float}
\usepackage{placeins}
\usepackage{simplewick}
\usepackage{graphics,psfrag}
\usepackage{pstricks}
\usepackage{latexsym}
\usepackage{changepage}
\usepackage{caption}
\usepackage{paralist}
\usepackage{styles/tikzit}
\usetikzlibrary{positioning,arrows.meta}

\tikzstyle{boundary vertex}=[inner sep=0mm, minimum size=1mm, shape=circle, draw=black, fill=black]
\tikzstyle{grey_dot}=[fill={rgb,255: red,191; green,191; blue,191}, draw={rgb,255: red,191; green,191; blue,191}, shape=circle, minimum size=1mm, inner sep=0mm]
\tikzstyle{blue_dot}=[fill={rgb,255: red,202; green,251; blue,255}, draw=black, shape=circle, minimum size=1.5mm, inner sep=0mm]
\tikzstyle{white_dot}=[fill=white, draw=black, shape=circle, minimum size=1.5mm, inner sep=0mm]
\tikzstyle{red_dot}=[fill=red, draw=black, shape=circle, minimum size=1.5mm, inner sep=0mm]
\tikzstyle{box}=[fill=white, draw=black, shape=rectangle]
\tikzstyle{scalable_dot}=[fill=white, draw=black, shape=circle, minimum size=6.5mm, inner sep=0mm]
\tikzstyle{scalable_box}=[fill=white, draw=black, shape=rectangle, minimum height=6.5mm, minimum width=6.5mm]
\tikzstyle{scalable_meta_X}=[fill={rgb,255: red,255; green,170; blue,162}, draw=black, shape=diamond, minimum height=8mm, minimum width=8mm]
\tikzstyle{scalable_meta_Z}=[fill={rgb,255: red,112; green,203; blue,207}, draw=black, shape=diamond, minimum height=8mm, minimum width=8mm]
\tikzstyle{X_detector}=[fill={rgb,255: red,191; green,0; blue,64}, draw=black, shape=rectangle]
\tikzstyle{Z_fault}=[fill={rgb,255: red,112; green,203; blue,207}, draw=black, shape=circle, minimum size=1mm, inner sep=0.75mm]
\tikzstyle{X_fault}=[fill={rgb,255: red,210; green,203; blue,0}, draw=black, shape=circle, minimum size=1mm, inner sep=0.75mm]
\tikzstyle{Z_detector}=[fill={rgb,255: red,105; green,105; blue,255}, draw=black, shape=rectangle]
\tikzstyle{scalable_meta_mixed}=[fill={rgb,255: red,255; green,105; blue,255}, draw=black, shape=diamond, minimum height=8mm, minimum width=8mm]
\tikzstyle{green_fault}=[fill={rgb,255: red,4; green,220; blue,0}, draw=black, shape=circle, minimum size=1mm, inner sep=0.75mm]

\tikzstyle{arrow}=[->]
\tikzstyle{red_arrow}=[->, draw=red]
\tikzstyle{cyan_arrow}=[->, draw=cyan]
\tikzstyle{red_dash}=[-, dashed, draw=red]
\tikzstyle{grey dash}=[-, fill=none, draw={rgb,255: red,191; green,191; blue,191}, dashed]
\tikzstyle{dashed arrow}=[->, dashed]
\tikzstyle{blue_edge}=[-, draw=cyan]
\tikzstyle{red_edge}=[-, draw=red]
\tikzstyle{green_edge}=[-, fill=none, draw=green]
\tikzstyle{green_dash}=[-, draw=green, dashed]
\tikzstyle{X_edge}=[-, dashed, draw={rgb,255: red,191; green,0; blue,64}]
\tikzstyle{Z_edge}=[dashed, draw={rgb,255: red,105; green,105; blue,255}, -]

\tikzstyle{gate}=[shape=rectangle, text height=1.5ex, text depth=0.25ex, yshift=0.5mm, fill=white, draw=black, minimum height=5mm, yshift=-0.5mm, minimum width=5mm, font={\small}, tikzit category=circuit]
\tikzstyle{big gate}=[shape=rectangle, text height=1.5ex, text depth=0.25ex, yshift=0.5mm, fill=white, draw=black, minimum height=18mm, yshift=-0.5mm, minimum width=5mm, font={\small}, tikzit category=circuit]
\tikzstyle{Z dot}=[inner sep=0mm, minimum size=2mm, shape=circle, draw=black, fill={rgb,255: red,221; green,255; blue,221}, tikzit category=zx]
\tikzstyle{Z phase dot}=[minimum size=5mm, font={\footnotesize\boldmath}, shape=rectangle, rounded corners=2mm, inner sep=0.2mm, outer sep=-2mm, scale=0.8, tikzit shape=circle, draw=black, fill={rgb,255: red,221; green,255; blue,221}, tikzit draw=blue, tikzit category=zx]
\tikzstyle{X dot}=[Z dot, shape=circle, draw=black, fill={rgb,255: red,255; green,136; blue,136}, tikzit category=zx]
\tikzstyle{X phase dot}=[Z phase dot, tikzit shape=circle, tikzit draw=blue, fill={rgb,255: red,255; green,136; blue,136}, font={\footnotesize\boldmath}, tikzit category=zx]
\tikzstyle{hadamard}=[fill=yellow, draw=black, shape=rectangle, inner sep=0.6mm, minimum height=1.5mm, minimum width=1.5mm, tikzit category=zx]
\tikzstyle{paulibox}=[fill={rgb,255: red,221; green,221; blue,255}, draw=black, shape=rectangle, inner sep=0.6mm, minimum height=5mm, minimum width=5mm, font={\footnotesize}, text height=1.5ex, text depth=0.25ex, tikzit category=zx]
\tikzstyle{vertex}=[inner sep=0mm, minimum size=1mm, shape=circle, draw=black, fill=black, tikzit category=misc]
\tikzstyle{vertex set}=[inner sep=0mm, minimum size=1mm, shape=circle, draw=black, fill=white, font={\footnotesize\boldmath}, tikzit category=misc]
\tikzstyle{small black dot}=[fill=black, draw=black, shape=circle, inner sep=0pt, minimum width=1.2mm, tikzit category=circuit]
\tikzstyle{cnot ctrl}=[fill=black, draw=black, shape=circle, inner sep=0pt, minimum width=1.2mm, tikzit category=circuit]
\tikzstyle{cnot targ}=[fill=white, draw=white, shape=circle, tikzit category=circuit, label={center:$\oplus$}, inner sep=0pt, minimum width=2.1mm, tikzit fill={rgb,255: red,102; green,204; blue,255}, tikzit draw=black]
\tikzstyle{ket}=[fill=white, draw=black, shape=regular polygon, regular polygon sides=3, regular polygon rotate=-30, scale=0.7, inner sep=1pt, tikzit category=circuit, tikzit shape=rectangle, tikzit fill=green]
\tikzstyle{bra}=[fill=white, draw=black, shape=regular polygon, regular polygon sides=3, regular polygon rotate=30, scale=0.7, inner sep=1pt, tikzit category=circuit, tikzit shape=rectangle, tikzit fill=red]
\tikzstyle{scalar}=[shape=rectangle, text height=1.5ex, text depth=0.25ex, yshift=0.5mm, fill=white, draw=black, minimum height=5mm, yshift=-0.5mm, minimum width=5mm, font={\small}]
\tikzstyle{clabel}=[fill=white, draw=none, shape=rectangle, tikzit fill={rgb,255: red,56; green,255; blue,242}, font={\footnotesize}, inner sep=1pt, tikzit category=labels]
\tikzstyle{empty diagram}=[draw={gray!40!white}, dashed, shape=rectangle, minimum width=1cm, minimum height=1cm, tikzit category=misc]
\tikzstyle{bigger gate}=[fill=white, draw=black, shape=rectangle, minimum height=28mm, minimum width=5mm, font={\small}, tikzit category=circuit]
\tikzstyle{wider gate}=[fill=white, draw=black, shape=rectangle, minimum height=28mm, minimum width=9mm, font={\small}]
\tikzstyle{biggest gate}=[fill=white, draw=black, shape=rectangle, minimum height=36mm, minimum width=9mm, font={\small}]

\tikzstyle{hadamard edge}=[-, dashed, dash pattern=on 2pt off 0.5pt, thick, draw={rgb,255: red,68; green,136; blue,255}]
\tikzstyle{box edge}=[-, dashed, dash pattern=on 2pt off 0.5pt, thick, draw={rgb,255: red,203; green,192; blue,225}]
\tikzstyle{brace edge}=[-, tikzit draw=blue, decorate, decoration={brace,amplitude=1mm,raise=-1mm}]
\tikzstyle{diredge}=[->]
\tikzstyle{double edge}=[-, double, shorten <=-1mm, shorten >=-1mm, double distance=2pt]
\tikzstyle{gray edge}=[-, {gray!60!white}]
\tikzstyle{pointer edge}=[->, very thick, gray]
\tikzstyle{boldedge}=[-, line width=1.6pt, shorten <=-0.17mm, shorten >=-0.17mm]

\tikzstyle{boundary vertex}=[inner sep=0mm, minimum size=1mm, shape=circle, draw=black, fill=black]
\tikzstyle{grey_dot}=[fill={rgb,255: red,191; green,191; blue,191}, draw={rgb,255: red,191; green,191; blue,191}, shape=circle, minimum size=1mm, inner sep=0mm]
\tikzstyle{blue_dot}=[fill={rgb,255: red,202; green,251; blue,255}, draw=black, shape=circle, minimum size=1mm, inner sep=0mm]
\tikzstyle{white_dot}=[fill=white, draw=black, shape=circle, minimum size=1.5mm, inner sep=0mm]
\tikzstyle{gate}=[shape=rectangle, text height=1.5ex, text depth=0.25ex, yshift=0.5mm, fill=white, draw=black, minimum height=5mm, yshift=-0.5mm, minimum width=5mm, font={\small}, tikzit category=circuit]
\tikzstyle{big gate}=[shape=rectangle, text height=1.5ex, text depth=0.25ex, yshift=0.5mm, fill=white, draw=black, minimum height=18mm, yshift=-0.5mm, minimum width=5mm, font={\small}, tikzit category=circuit]
\tikzstyle{Z dot}=[inner sep=0mm, minimum size=2mm, shape=circle, draw=black, fill={rgb,255: red,221; green,255; blue,221}, tikzit category=zx]
\tikzstyle{Z phase dot}=[minimum size=5mm, font={\footnotesize\boldmath}, shape=rectangle, rounded corners=2mm, inner sep=0.2mm, outer sep=-2mm, scale=0.8, tikzit shape=circle, draw=black, fill={rgb,255: red,221; green,255; blue,221}, tikzit draw=blue, tikzit category=zx]
\tikzstyle{X dot}=[Z dot, shape=circle, draw=black, fill={rgb,255: red,255; green,136; blue,136}, tikzit category=zx]
\tikzstyle{X phase dot}=[Z phase dot, tikzit shape=circle, tikzit draw=blue, fill={rgb,255: red,255; green,136; blue,136}, font={\footnotesize\boldmath}, tikzit category=zx]
\tikzstyle{hadamard}=[fill=yellow, draw=black, shape=rectangle, inner sep=0.6mm, minimum height=1.5mm, minimum width=1.5mm, tikzit category=zx]
\tikzstyle{paulibox}=[fill={rgb,255: red,221; green,221; blue,255}, draw=black, shape=rectangle, inner sep=0.6mm, minimum height=5mm, minimum width=5mm, font={\footnotesize}, text height=1.5ex, text depth=0.25ex, tikzit category=zx]
\tikzstyle{vertex}=[inner sep=0mm, minimum size=1mm, shape=circle, draw=black, fill=black, tikzit category=misc]
\tikzstyle{vertex set}=[inner sep=0mm, minimum size=1mm, shape=circle, draw=black, fill=white, font={\footnotesize\boldmath}, tikzit category=misc]
\tikzstyle{small black dot}=[fill=black, draw=black, shape=circle, inner sep=0pt, minimum width=1.2mm, tikzit category=circuit]
\tikzstyle{cnot ctrl}=[fill=black, draw=black, shape=circle, inner sep=0pt, minimum width=1.2mm, tikzit category=circuit]
\tikzstyle{cnot targ}=[fill=white, draw=white, shape=circle, tikzit category=circuit, label={center:$\oplus$}, inner sep=0pt, minimum width=2.1mm, tikzit fill={rgb,255: red,102; green,204; blue,255}, tikzit draw=black]
\tikzstyle{ket}=[fill=white, draw=black, shape=regular polygon, regular polygon sides=3, regular polygon rotate=-30, scale=0.7, inner sep=1pt, tikzit category=circuit, tikzit shape=rectangle, tikzit fill=green]
\tikzstyle{bra}=[fill=white, draw=black, shape=regular polygon, regular polygon sides=3, regular polygon rotate=30, scale=0.7, inner sep=1pt, tikzit category=circuit, tikzit shape=rectangle, tikzit fill=red]
\tikzstyle{scalar}=[shape=rectangle, text height=1.5ex, text depth=0.25ex, yshift=0.5mm, fill=white, draw=black, minimum height=5mm, yshift=-0.5mm, minimum width=5mm, font={\small}]
\tikzstyle{clabel}=[fill=white, draw=none, shape=rectangle, tikzit fill={rgb,255: red,56; green,255; blue,242}, font={\footnotesize}, inner sep=1pt, tikzit category=labels]
\tikzstyle{empty diagram}=[draw={gray!40!white}, dashed, shape=rectangle, minimum width=1cm, minimum height=1cm, tikzit category=misc]
\tikzstyle{rmat}=[draw, signal, fill={gray!30}, signal to=east, signal from=west, inner sep=1.5pt, minimum height=9pt]
\tikzstyle{lmat}=[draw, signal, fill={gray!30}, signal to=west, signal from=east, inner sep=1.5pt, minimum height=9pt]
\tikzstyle{ggreen}=[fill=green, draw=black, shape=circle, tikzit category=SZX, tikzit fill=green, tikzit draw=black, line width=1pt, inner sep=2pt]
\tikzstyle{gred}=[fill=red, draw=black, shape=circle, rounded corners=2mm, tikzit category=SZX, inner sep=2pt, tikzit fill=red, line width=1pt]
\tikzstyle{divide}=[regular polygon, regular polygon sides=3, shape border rotate=90, draw=black, fill={gray!30}, inner sep=1.6pt, tikzit category=scal, rounded corners=0.8mm]
\tikzstyle{gather}=[fill={gray!30}, draw=black, tikzit category=scal, rounded corners=0.8mm, regular polygon, regular polygon sides=3, shape border rotate=-90, inner sep=1.6pt]
\tikzstyle{A}=[fill=white, shape=circle, tikzit category=scal, inner sep=1pt]
\tikzstyle{umat}=[draw, signal, fill={gray!30}, signal to=north, signal from=south, inner sep=1.5pt, minimum width=9pt]
\tikzstyle{dmat}=[draw, signal, fill={gray!30}, signal to=south, signal from=north, inner sep=1.5pt, minimum width=9pt]

\tikzstyle{arrow}=[->]
\tikzstyle{red_arrow}=[->, draw=red]
\tikzstyle{cyan_arrow}=[->, draw=cyan]
\tikzstyle{red_dash}=[-, dashed, draw=red]
\tikzstyle{grey dash}=[-, fill=none, draw={rgb,255: red,191; green,191; blue,191}, dashed]
\tikzstyle{dashed arrow}=[->, dashed]
\tikzstyle{blue_edge}=[-, draw={rgb,255: red,46; green,126; blue,255}]
\tikzstyle{red_edge}=[-, draw=red]
\tikzstyle{hadamard edge}=[-, dashed, dash pattern=on 2pt off 0.5pt, thick, draw={rgb,255: red,68; green,136; blue,255}]
\tikzstyle{box edge}=[-, dashed, dash pattern=on 2pt off 0.5pt, thick, draw={rgb,255: red,203; green,192; blue,225}]
\tikzstyle{brace edge}=[-, tikzit draw=blue, decorate, decoration={brace,amplitude=1mm,raise=-1mm}]
\tikzstyle{diredge}=[->]
\tikzstyle{double edge}=[-, double, shorten <=-1mm, shorten >=-1mm, double distance=2pt]
\tikzstyle{gray edge}=[-, {gray!60!white}]
\tikzstyle{pointer edge}=[->, very thick, gray]
\tikzstyle{boldedge}=[-, line width=1.6pt, shorten <=-0.17mm, shorten >=-0.17mm]
\tikzstyle{very thick}=[-, line width=1pt, tikzit draw=red]

\usepackage{graphicx}
\usepackage{array}
\usepackage{hyperref}
\usepackage{comment}
\usepackage{tikz-cd}
\usepackage{mathrsfs}
\usepackage{braket}
\usepackage{soul}

\usepackage{subcaption}

\usepackage{hhline}
\usepackage{algorithm}
\usepackage{algpseudocode}

\usepackage{scalefnt}
\usepackage[all]{xypic}
\usepackage{slashed}
\usepackage{extarrows}
\usepackage{lettrine}
\usepackage{booktabs}
\usepackage{paralist}
\usepackage{wrapfig}
\usepackage{xcolor}
\usepackage{enumerate}
\usepackage{dsfont}
\usepackage{xparse}
\usepackage[normalem]{ulem}
\newtheorem{lemma}{Lemma}[section]
\newtheorem{proposition}[lemma]{Proposition}
\newtheorem{corollary}[lemma]{Corollary}
\newtheorem{example}[lemma]{Example}
\newtheorem*{workedexample}{Example}

\newtheorem{definition}[lemma]{Definition}

\usepackage{thmtools}
\usepackage{thm-restate}

\DeclareMathAlphabet{\mathcal}{OMS}{cmsy}{m}{n}
\renewcommand{\imath}{\mathrm{i}}

\newcommand{\F}{\mathbb{F}}
\newcommand{\Z}{\mathbb{Z}}

\newcommand{\nkd}[3]{\ensuremath{\left \llbracket #1,\, #2,\, #3\right \rrbracket}}

\newcommand{\supp}{{\rm supp}}

\setstcolor{red}

\usepackage[numbers, sort, compress]{natbib}
\definecolor{ExpeditedTeal}{HTML}{006D77}
\definecolor{ExpeditedMagenta}{HTML}{B5179E}

\hypersetup{
    colorlinks=true,
    linkcolor=ExpeditedTeal,
    citecolor=ExpeditedMagenta,
    urlcolor=ExpeditedTeal,
    pdftitle={Building codes with transversal CCZ using projective geometry and SAT solvers},
}
\newcommand{\SF}{S_{\mathrm{FG48}}}

\newcommand{\Qcode}{\mathcal{Q}_{48}}
\newcommand{\PG}{\operatorname{PG}}
\newcommand{\RM}{\operatorname{RM}}
\newcommand{\rowspan}{\operatorname{rowspan}}

\begin{document}

\title{Building codes with transversal CCZ using projective geometry and SAT solvers}
\author{Bohan Lu}
\email{bohan.lu@duke.edu}
\affiliation{Duke Quantum Center, Duke University, Durham, NC 27701, USA}
\affiliation{Department of Electrical and Computer Engineering, Duke University, Durham, NC 27708, USA}
\author{Kenneth R. Brown}
\affiliation{Duke Quantum Center, Duke University, Durham, NC 27701, USA}
\affiliation{Department of Electrical and Computer Engineering, Duke University, Durham, NC 27708, USA}
\affiliation{Department of Physics, Duke University, Durham, NC 27708, USA}
\affiliation{Department of Chemistry, Duke University, Durham, NC 27708, USA}

\begin{abstract}
We construct CSS codes with three logical qubits whose logical CCZ gate is implemented by transversal physical $T/T^{\dagger}$ gates and whose distance is at least $3$. These codes start with $8$-divisible $X$-stabilizers defined by projective geometry, which guarantees $d_Z\geq 3$. A SAT solver then finds three $X$-logical operators compatible with the transversal CCZ. We build codes at thirteen block lengths from $n=48$ to $n=496$, with $d_Z=3$ or $4$, and illustrate the method with the $\nkd{48}{3}{3}$ code $\Qcode$. Additionally, we prove that no CSS code with three logical qubits, $Z$-distance at least $3$, and a quasi-transversal CCZ gate (transversal $T$ with a diagonal Clifford correction) exists below $n=39$. The $\nkd{47}{3}{3}$ code of Jacinto et al.~\cite{jacintoCompactFactories2026} attains the minimum length in this class, as established by the classification of Wills, Jain, and Singh~\cite{willsClassificationGeneralised2026}.
\end{abstract}

\maketitle
\hypersetup{pdftitle={Building codes with transversal CCZ using projective geometry and SAT solvers}}

\section{Introduction}\label{sec:exp-intro}
Fault-tolerant quantum computation requires non-Clifford logical gates. Magic-state distillation remains the standard route to such gates, but it introduces a substantial space-time overhead~\cite{bravyiMagicStateDistillation2012,gidneyMagicStateCultivation2024,bombinFaultTolerantPostselectionLowOverhead2024}.
Transversal-gate constructions address the overhead by designing codes whose physical phase rotations implement a desired logical diagonal gate directly~\cite{campbellUnifiedFrameworkMagic2017,haahCodesProtocolsDistilling2017,bargGeometricStructureTransversal2024}. Examples include color-code and topological constructions~\cite{bombinExactTopologicalQuantum2015,brownColorCodeLogical2025} and CSS-T codes~\cite{campsCSST2024,kohPhantomCodes2025}. Hu, Liang, and Calderbank characterize diagonal gates on CSS codes~\cite{huDesigningQuantumChannels2023}. Rengaswamy et al. give necessary conditions for fixed physical $T/T^\dagger$ patterns on stabilizer codes~\cite[Thm.~6]{rengaswamyOnOptimalityCSSCodes2020}, and Campbell and Howard give sufficient overlap conditions for CCZ~\cite{campbellUnifiedFrameworkMagic2017}. Related CSS constructions use the diagonal Clifford hierarchy or divisible coset weights~\cite{huClimbingDiagonalClifford2021,huDivisibleCodesQuantum2022}. Concurrent work develops codes over binary extension fields~\cite{gongMagicStateBinaryExtension2026}. For a single logical qubit and the $T$ gate, the shortest known such codes are the divisible and triorthogonal families of~\cite{jainTransversalCliffordTgate2025}. A recurring requirement in these constructions is that the code generated by the $X$-type stabilizers be \emph{triply even}, meaning every codeword weight is a multiple of $8$~\cite{betsumiyaMunemasa2012}.
Such a code is also called $8$-divisible, and the two terms are used interchangeably.

The present work focuses on finding Calderbank-Shor-Steane (CSS) codes with three logical qubits that implement a non-Clifford logical CCZ gate by physical $T$ and $T^\dagger$ rotations applied directly to the data block of a distance-$d$ code.
A CSS code is specified by nested classical binary codes
\[
\mathcal{C}_2=\rowspan(S)\subset\mathcal{C}_1=\rowspan(G),
\qquad
G=\begin{bmatrix}K\\S\end{bmatrix}.
\]
Here the $s$ rows of $S$ generate the $X$-type stabilizers, the $k=3$ rows of $K$ represent the
$X$-logical generators, and the $n$ columns index the physical qubits.
The quantity $s=\operatorname{rank}S=\dim\mathcal C_2$ is the ($X$) stabilizer rank.
For a binary vector or Boolean function $f$, $\lvert f\rvert$ denotes its Hamming
weight, and for a set $A$, $\lvert A\rvert$ denotes its cardinality.
The $X$-logicals are the vectors of $\mathcal{C}_1\setminus\mathcal{C}_2$ and the
$Z$-logicals those of $\mathcal{C}_2^\perp\setminus\mathcal{C}_1^\perp$, so the code distances are
$d_X=\min\{\,\lvert v\rvert:v\in\mathcal{C}_1\setminus\mathcal{C}_2\,\}$ and
$d_Z=\min\{\,\lvert v\rvert:v\in\mathcal{C}_2^\perp\setminus\mathcal{C}_1^\perp\,\}$.

Throughout the paper, $\mathcal{C}_2$ denotes the classical code and $S$ denotes only its chosen generator matrix.

The logical computational-basis state with label $\vec x\in\F_2^3$ is the uniform
superposition over the coset $K^{\!\top}\vec x+\mathcal{C}_2$ of $\mathcal{C}_2$ in
$\mathcal{C}_1$,
\begin{equation}\label{eq:exp-logicalstate}
\ket{\vec x_L}=\frac{1}{2^{s/2}}\sum_{\vec y\in\F_2^s}\ket{K^{\!\top}\vec x+S^{\!\top}\vec y}
=\frac{1}{2^{s/2}}\sum_{\vec y\in\F_2^s}\ket{G^{\!\top}\vec z},
\qquad\vec z=(\vec x,\vec y)\in\F_2^{3+s},
\end{equation}
following the notation of Campbell and Howard~\cite{campbellUnifiedFrameworkMagic2017}. Applying
$T^{\otimes n}$ to $\ket{\vec x_L}$ multiplies the term indexed by $\vec y$ by
$\omega^{\lvert G^{\!\top}\vec z\rvert}$ with $\omega=e^{2\pi i/8}$, so a diagonal gate's logical action
is governed by $\lvert G^{\!\top}\vec z\rvert\bmod 8$. The
\emph{Bravyi--Haah parity
identity}~\cite{bravyiMagicStateDistillation2012,campbellUnifiedFrameworkMagic2017,nezamiClassificationSmallTriorthogonal2022}
decomposes this weight mod~$8$ into single-row weights, pairwise overlaps and
triple overlaps: for any full-rank $G$ with rows $g^{\,1},\ldots,g^{\,r}$ and any
$\vec z\in\F_2^{r}$, with $i$, $i'$, $i''$ ranging over all $r$ rows of $G$,
\begin{equation}\label{eq:exp-BH}
\begin{aligned}
\lvert G^{\!\top}\vec z\rvert
&=\sum_i \lvert g^{\,i}\rvert\,z_i
-2\!\sum_{i<i'}\lvert g^{\,i}\wedge g^{\,i'}\rvert\,z_iz_{i'}\\
&\quad+4\!\sum_{i<i'<i''}\lvert g^{\,i}\wedge g^{\,i'}\wedge g^{\,i''}\rvert\,z_iz_{i'}z_{i''}
\pmod 8.
\end{aligned}
\end{equation}
where $\lvert g^{\,i}\wedge g^{\,i'}\rvert=\sum_j g^i_j g^{i'}_j$ is the \emph{overlap} (shared support)
of two rows, and similarly for triple overlaps
\footnote{Lemma~\ref{lem:exp-xorlift} states the exact identity over $\Z$. Terms of degree
$\ge 4$ have coefficients divisible by $8$, so~\eqref{eq:exp-BH} is the exact truncation
modulo~$8$.}.

For a target diagonal gate
$U_F=\sum_{\vec x}\omega^{F(\vec x)}\ket{\vec x}\!\bra{\vec x}$, we consider the case where a diagonal
Clifford correction $C$, applied after the transversal $T$ rotations, i.e. $C\,T^{\otimes n}$,  realizes logical
$U_F$ on the codespace. For $U_F=$ CCZ on $k=3$ logical qubits the target is
$F_{\mathrm{CCZ}}(\vec x)=4\,x_1x_2x_3$, and the constraint reads
\begin{equation}\label{eq:exp-FQT-CCZ}
\lvert K^{\!\top}\vec x+S^{\!\top}\vec y\rvert+2h(K^{\!\top}\vec x+S^{\!\top}\vec y)\equiv 4\,x_1x_2x_3\pmod 8
\quad\forall\,\vec x,\vec y,
\end{equation}
where the function $h$ collects the phases contributed by $C$: a linear term for each $S$-gate
 and a quadratic term for each $CZ$-gate. $h$ is explicitly defined in Eq.~\eqref{eq:exp-correction}.
Two regimes are considered:
\begin{itemize}
\item \emph{quasi-transversal CCZ}~\cite{campbellUnifiedFrameworkMagic2017}: $C$ is any diagonal
Clifford, so the correction may still contain entangling $CZ$ gates on pairs of qubits and gates
acting on the logical qubits.
\item
\emph{single-qubit transversal CCZ} (Definition~\ref{def:exp-CI}) is implemented by a physical
$T/T^\dagger$ pattern. Equivalently, apply $T^{\otimes n}$ followed by $S^\dagger$ on selected
qubits, using $T S^\dagger=T^\dagger$.
\end{itemize}

Equation~\eqref{eq:exp-FQT-CCZ} fixes the undetermined coefficients of~\eqref{eq:exp-BH},
one monomial of degree $\le3$ at a time, yielding nine conditions on the row overlaps of $G$, i.e. the
Campbell--Howard (CH) conditions~\cite{campbellUnifiedFrameworkMagic2017}. The conditions are necessary as
well as sufficient (Theorem~\ref{thm:exp-NS}).
Table~\ref{tab:exp-conditions} lists all nine conditions.
\begin{table}[tb]
\centering
\renewcommand{\arraystretch}{1.2}
\begin{tabular}{c l}
\toprule
\textbf{Label} & \textbf{Condition}\\
\midrule
\textbf{K1} & $\lvert K_i\rvert\equiv 0\pmod 2$\\
\textbf{K2} & $\lvert K_i\wedge K_{i'}\rvert\equiv 0\pmod 2$\\
\textbf{K3} & $\lvert K_1\wedge K_2\wedge K_3\rvert\equiv 1\pmod 2$\\
\midrule
\textbf{KS1} & $\lvert K_i\wedge S_l\rvert\equiv 0\pmod 2$\\
\textbf{KS2} & $\lvert K_i\wedge K_{i'}\wedge S_l\rvert\equiv 0\pmod 2$\\
\textbf{KS3} & $\lvert K_i\wedge S_l\wedge S_{l'}\rvert\equiv 0\pmod 2$\\
\midrule
\textbf{S1} & $\lvert S_l\rvert\equiv 0\pmod 2$\\
\textbf{S2} & $\lvert S_l\wedge S_{l'}\rvert\equiv 0\pmod 2$\\
\textbf{S3} & $\lvert S_l\wedge S_{l'}\wedge S_{l''}\rvert\equiv 0\pmod 2$\\
\bottomrule
\end{tabular}
\caption{The nine Campbell--Howard conditions~\cite{campbellUnifiedFrameworkMagic2017}, grouped like the rows of $G=[K;S]$: conditions on logical rows
(K1--K3), conditions mixing logical and stabilizer rows (KS1--KS3), and conditions on stabilizer rows (S1--S3). Indices range over $1\le i<i'\le3$
and $1\le l<l'<l''\le s$. K3 is the single odd parity.}
\label{tab:exp-conditions}
\end{table}

Classical coding theory constructs codes from finite geometry by
indexing code coordinates with points over a finite field and deriving weight properties from
intersections with subspaces~\cite{assmusKeyPolynomial1998,alyQLDPCFiniteGeometries2008}. The
construction uses the binary projective geometry $\PG(s-1,2)$, the $2^s-1$ distinct nonzero
vectors of $\F_2^s$.

An $8$-divisible $\mathcal C_2$ satisfies the stabilizer-only CH conditions, and choosing $S$
with distinct nonzero columns guarantees $d_Z\ge3$ for every choice
of logical rows since every weight-1 $Z$ error has a distinct syndrome.

The logical $K$-row search is encoded as a Boolean satisfiability
(SAT) problem, which has been widely adopted in solving constraint satisfaction problems in quantum error correction~\cite{jacintoCompactFactories2026,dalfaveroSatisfyingQuantumCodes2026}. Because $\mathcal C_2$ is fixed and already guarantees $d_Z\ge3$, the SAT encoding has no
distance constraints, and S1--S3 impose no constraints on the logical labels. The remaining
conditions give $3n$ Boolean label variables and
$7+6s+3\binom{s}{2}=O(s^2)$ parity constraints: respectively,
K1, K2, K3, KS1, KS2, and KS3.

Here, geometry fixes $\mathcal C_2$ and guarantees $d_Z\ge3$ before the SAT search for logical
rows, which is followed by the physical-sign test for transversal CCZ.
Jacinto et al.~\cite{jacintoCompactFactories2026} use SAT to find compact magic-state
distillation protocols and to prove lower bounds on their qubit count.

For FG48, the resulting $G$ matrix generates the single-qubit transversal code $\Qcode$, a new
$\nkd{48}{3}{3}$ CSS code with $\mathcal C_2$ equal to FG48, whose logical CCZ is implemented
by physical $T$ and $T^\dagger$ rotations. We encode the geometric
construction with the
column indicator of $S$ and the logical-row search with the column indicator of $G$.
Nezami and Haah introduced this indicator-polynomial representation to classify
small triorthogonal codes~\cite{nezamiClassificationSmallTriorthogonal2022}.

Furthermore, we prove that no CSS code with $k=3$ and $d_Z\ge3$ satisfying the nine CH conditions
exists at length $n\le38$, with any number of $X$-stabilizer generators
(Theorem~\ref{thm:exp-floor38}). The CH conditions involving only stabilizer generators limit the
valid block lengths to the weight spectrum of a Reed--Muller code, and each admissible length is
then decided exactly. The $\nkd{47}{3}{3}$ code of Jacinto et
al.~\cite{jacintoCompactFactories2026} is the smallest code of this class, as certified by the classification of Wills, Jain, and Singh~\cite{willsClassificationGeneralised2026}, and $\Qcode$ is the
smallest known single-qubit transversal one.

Section~\ref{sec:exp-fg48} gives the geometric construction, Section~\ref{sec:exp-criterion} the
single-qubit transversal criteria, Section~\ref{sec:exp-minimality} the length lower bound, and
Section~\ref{sec:exp-conclusion} summarizes our results and lists open problems,
including extensions to other non-Clifford gates and more logical qubits. The appendices contain the search details and deferred
proofs.

\section{Subspace-complement construction of stabilizers}\label{sec:exp-fg48}

In this section, we find geometric constructions that satisfy the stabilizer-only CH conditions. An $8$-divisible $\mathcal C_2$ satisfies the stabilizer-only conditions
\textbf{S1}--\textbf{S3} of Table~\ref{tab:exp-conditions}.
$8$-divisibility is a stronger condition:
\begin{equation}\label{eq:exp-divladder}
\begin{aligned}
\lvert S_l\rvert&\equiv0 \pmod 8,&
\lvert S_l\wedge S_{l'}\rvert&\equiv0 \pmod 4,\\
\lvert S_l\wedge S_{l'}\wedge S_{l''}\rvert&\equiv0 \pmod 2.
\end{aligned}
\end{equation}
Conditions S1--S3
only impose the conditions mod $2$. Codes satisfying these conditions without being
$8$-divisible are
known~\cite{nezamiClassificationSmallTriorthogonal2022,jainTransversalCliffordTgate2025}.

By imposing
$8$-divisibility, we can borrow from
classical constructions.
Divisible codes are a well-developed subject~\cite{kurzDivisible2021}, and
two-weight codes from subspace complements are among the simplest ways to produce one. The
stabilizer support is then fixed by geometry rather than found by search, while the logical rows and
physical signs remain to be chosen.

$8$-divisibility is the first of two requirements this section places on
$\mathcal C_2$. The second is the distance requirement $d_Z\ge3$, which we now reduce to a
condition on the columns of $S$.

Write each column of $G=[K;S]$ as a pair $(\alpha,\beta)\in\F_2^3\times\F_2^s$,
where $\alpha$ is its logical label (the column of $K$) and $\beta$ is its \emph{syndrome} (the
column of $S$, i.e.\ the $X$-stabilizer syndrome of a $Z$ error on that qubit). The $Z$-logicals
are $\mathcal C_2^\perp\setminus\mathcal C_1^\perp$ (Section~\ref{sec:exp-intro}), so the
$Z$-distance is
\begin{equation}\label{eq:exp-dZ}
d_Z=\min\{\,\lvert v\rvert:v\in\mathcal C_2^\perp\setminus\mathcal C_1^\perp\,\}
   =\min\{\,\lvert v\rvert:Sv=0,\ Kv\neq0\,\},
\end{equation}
the minimum weight of a vector with zero syndrome and nonzero logical label.

\begin{lemma}[Multiplicity-free support $=d_Z\ge3$]\label{lem:exp-fiber}
The following are equivalent: \textup{(a)} $d_Z\ge3$; \textup{(b)} no column has $\beta=0$ with
$\alpha\ne0$, and no two columns share the same nonzero $\beta$ with different labels $\alpha$. A support satisfying \textup{(b)}
is called \emph{multiplicity-free}.
\end{lemma}

\begin{proof}
A $Z$ error of weight at most two is determined by its support columns. For such $v$, the syndrome
$Sv$ and label $Kv$ are the sums of at most two columns, and $d_Z\le2$ means some such string
satisfies $Sv=0$ with $Kv\ne0$. The only candidates are
\[
v=e_j:\ (Sv,Kv)=(\beta_j,\alpha_j);
\qquad
v=e_i+e_j:\ (Sv,Kv)=(\beta_i+\beta_j,\ \alpha_i+\alpha_j).
\]
So a weight-one $Z$-logical is a column with zero syndrome but nonzero label, and a weight-two
$Z$-logical is a pair of columns sharing a syndrome but not a label. Condition \textup{(b)} forbids
exactly these two configurations, hence is equivalent to $d_Z\ge3$, which is \textup{(a)}.
\end{proof}

The code $\mathcal C_2$ is called \emph{projective} when the columns of
its generator matrix $S$ are distinct and nonzero, equivalently when
$d(\mathcal C_2^\perp)\ge3$.
The columns $\beta_1,\ldots,\beta_n\in\F_2^s$ then form a projective
set~\cite{landjev2000,kurzDivisible2021} $\mathcal P\subseteq\F_2^s\setminus\{0\}$ of
$n\le2^s-1$ points, where we write $\PG(s{-}1,2)$ for the set of $2^s-1$ nonzero vectors of
$\F_2^s$.
Projectivity excludes both failures in
Lemma~\ref{lem:exp-fiber}, so it guarantees $d_Z\ge3$ for every choice of $K$.
For the existence and length-lower-bound searches, it is enough to consider projective
$\mathcal C_2$.
Every $G$ with $d_Z\ge3$ satisfying the nine CH conditions can be punctured,
by deleting a zero column or a repeated pair of columns,
to one that contains a projective $\mathcal C_2$. The punctured code satisfies the same nine CH conditions,
$d_Z\ge3$, and possesses the same number of logicals $k$
(see Lemma~\ref{lem:exp-projred}, Appendix~\ref{app:exp-projred}).
The search therefore fixes a projective $\mathcal C_2$ and varies only the logical rows $K$.

FG48 and its subspace-complement family are constructed from projective
geometry~\cite[\S2]{kurzDivisible2021}, where codeword weights are read from hyperplane sections.
A nonzero vector $a\in\F_2^s$ selects the codeword
$c_a = a_1 S_1+\cdots+a_s S_s\in\F_2^n$, whose $j$-th coordinate is the dot product $a\cdot\beta_j$.
The same $a$ defines the \emph{hyperplane} $\mathcal H_a=\{\beta\in\F_2^s\setminus\{0\}: a\cdot\beta=0\}$, and the
points of $\mathcal P$ lying on it form the \emph{hyperplane section} $\mathcal P\cap \mathcal H_a$. A hyperplane of
$\PG(s-1,2)$ has $2^{s-1}-1$ points ($31$ when $s=6$). Coordinate $j$ is nonzero exactly when $\beta_j$ lies outside $\mathcal H_a$, so
\begin{equation}\label{eq:exp-weight-section}
\lvert c_a\rvert=\sum_{\beta\in\mathcal P} a\cdot\beta=n-\lvert \mathcal P\cap \mathcal H_a\rvert .
\end{equation}
The code is \emph{two-weight} when only two distinct nonzero codeword weights
occur~\cite{calderbankKantor1986,brouwerTwoWeight2021}. By~\eqref{eq:exp-weight-section} this is the
same as the hyperplane sections $\mathcal P\cap \mathcal H_a$ taking only two distinct sizes~\cite[Thm.~2.18]{kurzDivisible2021}. One way to construct such two-weight point sets is by using complements of subspaces, defined below:

\begin{definition}[Subspace complements~{\cite[\S19.7.1]{brouwerTwoWeight2021}}]\label{def:exp-flatcomplement}
For a $p$-dimensional vector subspace $P\subseteq\F_2^s$ (so $P\cong\F_2^p$), the \emph{subspace-complement} point set is
\[
\mathcal P=\PG(s-1,2)\setminus(P\setminus\{0\}),
\]
the projective points of $\F_2^s$ not in $P$
(see Calderbank--Kantor~\cite{calderbankKantor1986} for
this family's place in the two-weight literature).
\end{definition}

Set $(s,p)=(6,4)$ in
Definition~\ref{def:exp-flatcomplement}. Removing the $15$ nonzero points of a four-dimensional
subspace, a \emph{solid} $\PG(3,2)$, from the $63$ points of $\PG(5,2)$ leaves $48$ points. The
classical code with these $48$ points as columns is denoted
\textbf{FG48}\footnote{The deleted subspace has dimension exactly
$p=4$, forced by $n=2^6-2^p=48$, and a generic $15$-point deletion would not be $8$-divisible. FG48
is the binary member of Calderbank--Kantor's subspace-complement family,
Example~SU1 (complement of a $p$-dimensional subspace of $\F_2^s$,
here $(s,p)=(6,4)$, cited in the proof of Proposition~\ref{prop:exp-uniqueness}). Two-weight codes are a broad class with several unrelated families.}.
\begin{proposition}[{Uniqueness of the projective $8$-divisible $[48,6]$
code~\cite{calderbankKantor1986,brouwerTwoWeight2021,kurzDivisible2021}}]\label{prop:exp-uniqueness}
Up to a permutation of coordinates, the only projective $8$-divisible $[48,6]$ code is FG48, the
two-weight complement of a solid, with column set $\PG(5,2)\setminus\PG(3,2)$, the code
$\mathcal C_2$ of $\Qcode$.
\end{proposition}
\begin{proof}
The $48$ columns of $\mathcal C_2$ occupy $48$ of the $63$ points of $\PG(5,2)$, missing $15$.
Each hyperplane $\mathcal H_a$ has $31$ points, so it meets the column set in at least
$31-15=16$ points and at most $31$. By~\eqref{eq:exp-weight-section} the nonzero codeword
$c_a$ has weight $48$ minus this section size, and $8$-divisibility forces the section size to
be a multiple of $8$. In the range $[16,31]$ the only multiples of $8$ are $16$ and $24$,
giving codeword weights $32$ and $24$.

To determine how many hyperplanes give each weight, we count incidences in two ways.
There are $63$ hyperplanes in $\PG(5,2)$ (one per nonzero $a\in\F_2^6$), and each of the $48$
points lies on $31$ of them. For $m\in\{16,24\}$, let $N_m$ be the number of hyperplanes whose
intersection with the $48$-point column set has size $m$. Then
\[
N_{16}+N_{24}=63,\qquad 16N_{16}+24N_{24}=48\cdot31,
\]
so $N_{16}=3$ and $N_{24}=60$. Thus, our code has 3 codewords of weight 32 and 60 codewords of weight 24.

Every projective $8$-divisible $[48,6]$ code has this enumerator. The solid complement of
Definition~\ref{def:exp-flatcomplement} attains it
(Calderbank--Kantor Example~SU1~\cite[\S7]{calderbankKantor1986}), and among
projective codes the enumerator has no other realization, by the classification of projective
$8$-divisible codes of length at most
$48$~\cite{betsumiyaMunemasa2012,heinleinProjectiveDivisible2017,kurzDivisible2021}.
\end{proof}

To write an explicit generator matrix $S$ for FG48, we choose coordinates $(u,w)\in\F_2^4\mathbin\oplus\F_2^2$
with the removed subspace given by $w=0$. The surviving points have $w\ne0$ and split into three
cosets, one for each nonzero value of $w$.
Ordering the columns of $S$ coset by coset, and within each coset by the $16$ values of $u$,
gives~\eqref{eq:exp-blockform}. The columns of $A$ are the $16$ vectors of $\F_2^4$, written as
the binary expansions of $0,1,\ldots,15$ with the least significant bit in the top row, while the
bottom two rows record the three coset labels $(1,0),(0,1),(1,1)$:
\begin{equation}\label{eq:exp-blockform}
\setlength{\arraycolsep}{3pt}
\SF=
\begin{bmatrix}
A&A&A\\
\mathbf 1&\mathbf 0&\mathbf 1\\
\mathbf 0&\mathbf 1&\mathbf 1
\end{bmatrix},
\qquad
A=
\begin{bmatrix}
0&1&0&1&0&1&0&1&0&1&0&1&0&1&0&1\\
0&0&1&1&0&0&1&1&0&0&1&1&0&0&1&1\\
0&0&0&0&1&1&1&1&0&0&0&0&1&1&1&1\\
0&0&0&0&0&0&0&0&1&1&1&1&1&1&1&1
\end{bmatrix}.
\end{equation}

The six rows of $\SF$ are independent
coordinate functions on the surviving points and form a basis of $\mathcal C_2$. The two nonzero
weights arise from whether the hyperplane $\mathcal H_a$ of the codeword $c_a$
in~\eqref{eq:exp-weight-section} contains the removed subspace $P$.
\begin{corollary}[Block form of a subspace complement code]\label{cor:exp-blockform}
Let $P\cong\F_2^p$ be a $p$-dimensional subspace of $\F_2^s$.
Set $q=s-p$ and let $\pi\colon\F_2^s\to\F_2^s/P\cong\F_2^q$ send each $\beta$
to its coset $\beta+P$.
Choose coordinates $(u,w)$ on $\F_2^s$, with $u\in\F_2^p$ along $P$ and
$w=\pi(\beta)\in\F_2^q$.
The generator matrix of the code with column set
$\PG(s{-}1,2)\setminus\PG(p{-}1,2)$ then admits the block form
\[
S=\begin{bmatrix}A&A&\cdots&A\\[2pt]
w^{(1)}\mathbf 1&w^{(2)}\mathbf 1&\cdots&w^{(b)}\mathbf 1\end{bmatrix},
\]
where $b=2^q-1$ blocks are indexed by the nonzero vectors $w^{(1)},\ldots,w^{(b)}\in\F_2^q\setminus\{0\}$, and
the top $p$ rows repeat the block matrix $A$ listing all $2^p$ vectors of~$\F_2^p$ in each block.
Here $w^{(i)}\mathbf 1$ is the $q\times2^p$ block whose every column equals
$w^{(i)}$, so all $2^p$ columns of block $i$ carry the same coset label.
The block length is $n=b\cdot 2^p = 2^s-2^p$.
\end{corollary}

Proposition~\ref{prop:exp-uniqueness} determines
$\mathcal C_2=\mathrm{FG48}$ once $8$-divisibility is imposed.
A classical support determines neither the logical rows nor the physical gate pattern needed to
implement the target logical gate.
Section~\ref{sec:exp-criterion} constructs the logical rows and physical signs for FG48.

\section{CCZ conditions and a 48-qubit example}\label{sec:exp-criterion}

\subsection{CH conditions and column indicators}\label{sec:exp-full-g}
Section~\ref{sec:exp-fg48} treated the columns of $\SF$ as a projective point set and computed
codeword weights from sums over those columns. The same column picture applies to the
full generator matrix $G$. The CH conditions become
parity sums over columns in $G$, which can be captured by column indicator polynomials. Membership
of these indicator polynomials in certain Reed--Muller code cosets will also be characterized,
following the methodology of~\cite{nezamiClassificationSmallTriorthogonal2022}. The notation
$\RM(r,m)$ denotes the binary Reed--Muller code of order $r$ and length $2^m$, namely the
evaluations on $\F_2^m$ of polynomials of degree at most $r$.

\begin{definition}[Column indicator]\label{def:exp-indicator}

Let $M$ be a binary matrix with $r$ rows and pairwise distinct columns, and let
$X\subseteq\F_2^r$ be its set of columns. The \emph{column indicator} of $M$ is the Boolean
function
\[
\chi_M\colon\F_2^r\longrightarrow\F_2,\qquad
\chi_M(v)=\begin{cases}1&\text{if }v\in X,\\0&\text{otherwise,}\end{cases}
\]
For a Boolean polynomial or function $f$, its Hamming weight $\lvert f\rvert$ is the number of inputs at which
$f$ equals $1$, equivalently the Hamming weight of its evaluation vector. Hence
the same notation applies when $f$ represents an evaluation vector in a
Reed--Muller code or coset. In particular,
$\lvert \chi_M\rvert=\lvert X\rvert$, which equals the number of columns because the columns are
distinct.
Note every Boolean function $g$ on $\F_2^r$ has a unique multilinear expansion
\[
g=\sum_{T\subseteq\{1,\dots,r\}}c_T\prod_{l\in T}v_l,
\qquad c_T\in\F_2 ,
\]
and its degree is the largest
$\lvert T\rvert$ with $c_T=1$.  If $\chi_M$ has degree $\deg\chi_M$, then it belongs to the Reed--Muller code $\RM(\deg\chi_M,\,r)$.
\end{definition}
For the stabilizer matrix $M=\SF$, the ambient dimension is $r=s$.
Write $\chi_S$ for its indicator, whose support is the projective point set
$\mathcal P\subseteq\F_2^s\setminus\{0\}$ of Section~\ref{sec:exp-fg48}.
Consequently, $\chi_S(0)=0$ and $\lvert\chi_S\rvert=n$.
For the full generator matrix $M=G=[K;\SF]$, the ambient dimension is $r=k+s$,
and we write $\chi_G$ for its indicator.
Since the stabilizer columns are distinct, so are the columns of $G$, and the
indicator $\chi_G$ records $G$ faithfully.

Equation~\eqref{eq:exp-weight-section} computes a codeword weight by summing a linear
function over $\mathcal P$. For the full generator matrix $G$, each CH condition instead sums a
monomial of degree at most three over the columns and retains only the parity. For
Boolean functions $f,g$ on $\F_2^{r}$, define~\cite{campbellUnifiedFrameworkMagic2017}
\begin{equation*}
\langle f,g\rangle=
\left(\sum_{v\in\F_2^{r}}f(v)g(v)\right)\bmod 2.
\end{equation*}
In Boolean algebra and in SAT encodings, a sum mod $2$ i.e the parity, is the XOR, denoted $\bigoplus$.
For the linear function $\ell_a(z)=a\cdot z$, the inner product
$\langle\chi_S,\ell_a\rangle$ counts the parity of the points of $\mathcal P$ outside
the hyperplane $\mathcal H_a$. Hence the degree-one case recovers the parity of
the hyperplane count in~\eqref{eq:exp-weight-section}.

For a nonempty set $T$ of at most three rows, let
$T_K\subseteq\{1,2,3\}$ and $T_S\subseteq\{1,\ldots,s\}$ index its logical and stabilizer rows.
Define
\[
m_{T_K}(x)=\prod_{i\in T_K}x_i,\qquad
m_{T_S}(y)=\prod_{j\in T_S}y_j,\qquad
m_T(x,y)=m_{T_K}(x)m_{T_S}(y).
\]
On a column of $G$, $m_T=1$ exactly when every selected row has entry $1$.
Consequently, $\langle\chi_G,m_T\rangle$ counts their overlap modulo two.
Note that the factorization $m_T(x,y)=m_{T_K}(x)m_{T_S}(y)$
comes only from partitioning the selected rows into the $K$ and $S$ rows and no factorization of
$\chi_G$ is assumed.

Each CH condition is such an inner product for a set $T$ of at most three rows.
Define the \emph{label map} $\kappa\colon\mathcal P\to\F_2^3$ by assigning
to each stabilizer column the three entries of $K$ above it.
Lemma~\ref{lem:exp-fiber} ensures that each column has one such label.
In terms of $\kappa$, the CH conditions read
\begin{equation}\label{eq:exp-bridge}
\langle\chi_G,m_T\rangle
=\bigoplus_{\beta\in\mathcal P}m_{T_K}\bigl(\kappa(\beta)\bigr)m_{T_S}(\beta)=
\begin{cases}
1 & \text{if }T=\{1,2,3\}\text{, the K3 condition},\\
0 & \text{otherwise}.
\end{cases}
\end{equation}
To sum over all nonzero $\beta$ rather than only $\mathcal P$,
set $\kappa(\beta)=(0,0,0)$ for $\beta\notin\mathcal P$.
Multiplying by $\chi_S(\beta)$ makes every term outside $\mathcal P$ zero,
so Equation~\eqref{eq:exp-bridge} becomes
\begin{equation}\label{eq:exp-clause}
\bigoplus_{\beta\ne0} \chi_S(\beta)\,m_{T_S}(\beta)\,
m_{T_K}(\kappa(\beta))\;=\;
\begin{cases}
1,&T=\{1,2,3\},\\
0,&\text{otherwise}.
\end{cases}
\end{equation}
For the three stabilizer-only conditions S1--S3, $T_K=\emptyset$, so
$m_{T_K}=1$ and the label map drops out.

Appendix~\ref{app:exp-search} describes how the logical rows are found for a
fixed support.
For a fixed support and two logical rows satisfying the pair tests~\eqref{eq:exp-admissible-pair},
the remaining equations become linear in $\kappa_3$.
Algorithm~\ref{alg:exp-search} uses this fact to test each pair of logical rows
by solving a linear system for the third.

Lemma~\ref{lem:exp-coset} below expresses every solution of the CH equations as a
member of a Reed--Muller coset.

\begin{restatable}{lemma}{ExpCoset}\label{lem:exp-coset}
Write $r=k+s$. Let $\delta_0=\prod_{i=1}^{r}(1+z_i)$ be the indicator
of the origin of $\F_2^r$. Fix $\chi_{G,0}$ with
$\langle\chi_{G,0},x_1x_2x_3\rangle=1$, orthogonal to all monomials of degree
below three and all other degree-three monomials.
The CH conditions force the $G$ indicator
to belong to a coset of $\RM(r-4,r)$,
\begin{equation}\label{eq:exp-fullcoset}
\chi_G\;\in\;\chi_{G,0}+\varepsilon\delta_0+\RM(r-4,r),\qquad \varepsilon=n\bmod2,
\end{equation}
which does not depend on the choice of $\chi_{G,0}$.
\end{restatable}

The proof of Lemma~\ref{lem:exp-coset} is given in
Appendix~\ref{app:exp-coset-proof}.

For three logical rows, the function $\chi_{G,0}=y_1\cdots y_s$ on $\F_2^{3+s}$ is one special solution. The function equals $1$ exactly at the $8$ points $(x,\mathbf1)$ with
$x\in\F_2^3$, where $m_{T_S}(\mathbf1)=1$. Hence, for every $T$ with $\lvert T\rvert\le3$,
\[
\langle\chi_{G,0},m_T\rangle=\Bigl(\sum_{x\in\F_2^3}m_{T_K}(x)\Bigr)\bmod2
=2^{3-\lvert T_K\rvert}\bmod2,
\]
which equals $1$ only for $T_K=\{1,2,3\}$.
At even length, the $\delta_0$ term vanishes, so~\eqref{eq:exp-fullcoset} gives
$\chi_G\in y_1\cdots y_s+\RM(s-1,s+3)$.
Hence $\deg\chi_G=s$, with $y_1\cdots y_s$ as its only degree-$s$ term.

Represent a column of $G$ by $(x,u,w)\in\F_2^{3+s}$. The last $s$ coordinates
$(u,w)$ are the block coordinates of Corollary~\ref{cor:exp-blockform}, and the first
three coordinates are the logical label $x=\kappa(u,w)$. Thus, for a column with label $\alpha$
and syndrome $\beta=(u,w)$, the coordinate functions satisfy
$x_i(\alpha)=(\alpha)_i$ and $y_l(\beta)=(\beta)_l$.

The support factor $\chi_S(u,w)$ selects the occupied stabilizer columns. For each selected
column, the factor indexed by $i$ enforces $x_i=\kappa_i(u,w)$, so the product is the column
indicator of $G$:
\begin{equation}\label{eq:exp-graphind}
\chi_G(x,u,w)=\chi_S(u,w)\prod_{i=1}^{3}\bigl(1+x_i+\kappa_i(u,w)\bigr).
\end{equation}

\begin{workedexample}[The indicator of $\Qcode$]

Corollary~\ref{cor:exp-blockform} identifies the stabilizer columns with the points
$(u,w)$ for which $w\ne0$, with every $u\in\F_2^p$ occurring for each such $w$. Hence
$\chi_S$ is independent of $u$. Writing $w=(w_1,\ldots,w_q)$,
$\mathbf 1_{w=0}=(1+w_1)(1+w_2)\cdots(1+w_q)$:
\begin{equation}\label{eq:exp-chiS}
\chi_S(w)=\mathbf 1_{w\ne0}=1-\prod_{l=1}^q(1+w_l).
\end{equation}
For $q=2$ (FG48) this is $w_1+w_2+w_1w_2$, of degree $2$.

For each nonzero block $b$, the label map of $\Qcode$ is quadratic in the coordinates
$u$. Table~\ref{tab:exp-labelmap} lists $\kappa$ on each of the three blocks, one row
per block.
The rows are the output of Algorithm~\ref{alg:exp-search}.
\end{workedexample}
\begin{table}[tb]
\centering
\begin{tabular}{@{}c l l l@{}}
\toprule
$w^{(b)}$ & $\kappa_1$ & $\kappa_2$ & $\kappa_3$\\
\midrule
$(1,0)$ & $u_1+u_3$ & $1+u_1+u_2+u_3+u_4$ & $1+u_1+u_1u_2+u_1u_4+u_3u_4$\\
$(0,1)$ & $1+u_3$ & $1+u_1+u_3+u_1u_3+u_1u_4$ & $u_1+u_1u_2+u_2u_3+u_2u_4+u_3u_4$\\
$(1,1)$ & $u_4$ & $u_3+u_4+u_1u_3+u_1u_4$ & $u_1+u_3+u_1u_4+u_2u_3+u_2u_4$\\
\bottomrule
\end{tabular}
\caption{The label map $\kappa$ of $\Qcode$ in the block coordinates $(u,w)$, one row per nonzero
coset. With $\chi_S$ of~\eqref{eq:exp-chiS} these are the two factors of the
indicator~\eqref{eq:exp-graphind}, and their product reproduces $\chi_G$ at every point of
$\F_2^{9}$.}
\label{tab:exp-labelmap}
\end{table}

Equation~\eqref{eq:exp-bridge} separates the roles of the nine CH conditions.
Conditions S1--S3 depend only on the projective point set $\mathcal P$. Conditions K1--K3
constrain the parities of the numbers of columns carrying each logical label. The label map
$\kappa$ appears in KS1--KS3 and identifies where those labels occur.

For $\Qcode$, Equation~\eqref{eq:exp-chiS} gives a stabilizer indicator of degree
$2$, while the even-length consequence of~\eqref{eq:exp-fullcoset} gives a full indicator
of degree $6$. The degree of the
indicator of $G$ still does not determine the physical gate pattern. The next subsection
determines that pattern.

\subsection{CCZ criteria and the 48-qubit sign pattern}\label{sec:exp-native}

For the physical gate
$C\,T^{\otimes n}$ and logical CCZ considered here, comparing coefficients in
\eqref{eq:exp-FQT-CCZ} yields the nine conditions of Table~\ref{tab:exp-conditions}.

\begin{restatable}{theorem}{ExpNS}\label{thm:exp-NS}
Let $G=[K;S]$ be the generator matrix of a CSS $\nkd{n}{3}{d}$ code.
\begin{enumerate}[(i)]
\item \emph{Sufficiency}~\cite{campbellUnifiedFrameworkMagic2017}. If $G$ satisfies the CH conditions,
then $C\,T^{\otimes n}$ realizes logical CCZ for some diagonal Clifford correction $C$, a product of
$S$ and $CZ$ gates.
\item \emph{Necessity}~\cite[Thm.~7, Rem.~8, and Thm.~9]{huDesigningQuantumChannels2023}. Conversely, if $C\,T^{\otimes n}$ realizes logical CCZ for some
diagonal Clifford $C$, then $G$ satisfies all nine conditions.
\end{enumerate}
\end{restatable}%
\begin{proof}[Proof sketch of necessity]
Substitute the parity identity~\eqref{eq:exp-BH} into the
demand~\eqref{eq:exp-FQT-CCZ} and compare the coefficient of each monomial of degree $\le3$ in
$\vec z=(\vec x,\vec y)$. Set $\vec f=G^{\!\top}\vec z\in\F_2^n$. The diagonal correction contributes
$2h(\vec f)$, where
\begin{equation}\label{eq:exp-correction}
h(\vec f)=\sum_j\lambda_j f_j+\sum_{j<j'}\lambda_{jj'}f_jf_{j'},
\qquad\lambda_j\in\{0,1,2,3\},\quad\lambda_{jj'}\in\{0,2\}.
\end{equation}
Here $\lambda_j$ are the $S$-gate phases, $\lambda_{jj'}$ the $CZ$-gate phases, and
all arithmetic is mod~$8$.
The correction can absorb any multiple of $2$ at degree one and any multiple of $4$ at degree
two, and nothing at degree
three~\cite[App.~A]{nezamiClassificationSmallTriorthogonal2022}. The proof of
Theorem~\ref{thm:exp-NS} in Appendix~\ref{app:exp-ch-proof}
proves this from the parity identity~\eqref{eq:exp-BH}. So the logical gate is CCZ exactly when every residual coefficient
vanishes and the cubic coefficient of $x_1x_2x_3$ is $4$.
Imposing that on one monomial of degree at most three at a time, and discarding the coefficients
the correction absorbs, leaves exactly the nine parities of
Table~\ref{tab:exp-conditions}~\cite{campbellUnifiedFrameworkMagic2017}.
\end{proof}

A product of single-qubit gates cannot spread an error between qubits.
We consider a physical $T/T^\dagger$ pattern specified by $\Gamma\in\{\pm1\}^n$,
with $+1$ for $T$ and $-1$ for $T^\dagger$. Since $T S^\dagger=T^\dagger$,
the pattern is $T^{\otimes n}$ followed by $S^\dagger$ on the qubits with $\Gamma_j=-1$.
For $\Gamma=\vec1$, no correction is applied.

Applying a specific $T/T^\dagger$ pattern to a codeword $\vec v\in\F_2^n$ defines a
\emph{signed weight}. Interpret each coordinate $v_j$ as the integer $0$ or $1$. For a sign
vector $\Gamma\in\{\pm1\}^n$,
\begin{equation}\label{eq:exp-sw}
\mathrm{sw}_\Gamma(\vec v)=\sum_j\Gamma_j v_j,
\end{equation}
the Hamming weight with each coordinate counted with its sign. Signed weight replaces the Hamming weight
$\lvert G^{\!\top}\vec z\rvert$ of~\eqref{eq:exp-BH} whenever the pattern mixes $T$ and $T^\dagger$.

\begin{definition}[$T/T^\dagger$ implementation]\label{def:exp-CI}
$G=[K;S]$ admits a \emph{$T/T^\dagger$ implementation} of CCZ if there is $\Gamma\in\{\pm1\}^n$ with
\begin{equation}\label{eq:exp-CI}
\mathrm{sw}_\Gamma\!\bigl(K^{\!\top}\vec x+S^{\!\top}\vec y\bigr)\equiv 4\,x_1x_2x_3\pmod 8
\qquad\forall\,\vec x\in\F_2^3,\ \vec y\in\F_2^s,
\end{equation}
with $\mathrm{sw}_\Gamma$ the signed weight of~\eqref{eq:exp-sw}. Such an implementation is called \emph{single-qubit transversal}.
\end{definition}

Checking Definition~\ref{def:exp-CI} need not require evaluating every codeword.
Replacing the Hamming weight in Equation~\eqref{eq:exp-BH} by the signed
weight~\eqref{eq:exp-sw} gives the following equivalent row conditions for Equation~\eqref{eq:exp-CI}.
\begin{proposition}
\label{prop:exp-signed-overlaps}
The $T/T^\dagger$ pattern specified by $\Gamma$ is single-qubit transversal if and only if, for every row $g$,
every pair of distinct rows $g,g'$, and every triple of distinct rows $g,g',g''$,
\begin{equation}\label{eq:exp-signed-overlap-conditions}
\begin{aligned}
\mathrm{sw}_\Gamma(g)&\equiv0\pmod8,\\
\mathrm{sw}_\Gamma(g\wedge g')&\equiv0\pmod4,\\
\mathrm{sw}_\Gamma(g\wedge g'\wedge g'')&\equiv
\begin{cases}
1\pmod2,&\{g,g',g''\}=\{K_1,K_2,K_3\},\\
0\pmod2,&\text{otherwise}.
\end{cases}
\end{aligned}
\end{equation}
Moreover, the conditions above that contain a stabilizer row make the signed weight independent
of the stabilizers: if they hold, then
\begin{equation}\label{eq:exp-signed-stabilizer-cancellation}
\mathrm{sw}_\Gamma(K^{\!\top}\vec x+S^{\!\top}\vec y)
\equiv\mathrm{sw}_\Gamma(K^{\!\top}\vec x)\pmod8
\end{equation}
for every $\vec x$ and $\vec y$.
\end{proposition}
Once the conditions involving a stabilizer row hold,
Equation~\eqref{eq:exp-signed-stabilizer-cancellation} reduces the check of
Equation~\eqref{eq:exp-CI} to the eight logical inputs $\vec x\in\F_2^3$ at $\vec y=0$.

The proof is given in Appendix~\ref{app:exp-signed-overlaps-proof}.

Write $\Gamma_j=1-2\gamma_j$ with $\gamma_j\in\{0,1\}$. The first and second conditions in
Equation~\eqref{eq:exp-signed-overlap-conditions} then become
\begin{equation}\label{eq:exp-signed-local-system}
\gamma\cdot g\equiv\frac{\lvert g\rvert}{2}\pmod4,
\qquad
\gamma\cdot(g\wedge g')\equiv
\frac{\lvert g\wedge g'\rvert}{2}\pmod2,
\end{equation}
for every row $g$ and pair of distinct rows $g,g'$. Since $\Gamma_j\equiv1\pmod2$,
\[
\mathrm{sw}_\Gamma(g\wedge g'\wedge g'')
\equiv\lvert g\wedge g'\wedge g''\rvert\pmod2.
\]
The CH conditions therefore imply the third line of Equation~\eqref{eq:exp-signed-overlap-conditions}
for every choice of $\Gamma$.

Substituting $\Gamma_j=1-2\gamma_j$ directly
into~\eqref{eq:exp-CI} gives a system mod $4$ whose binary solutions are exactly the
single-qubit transversal patterns, as shown in the following Corollary.

\begin{corollary}\label{cor:exp-Z4}
Let $G=[K;S]$ satisfy the nine CH conditions of Table~\ref{tab:exp-conditions}.
Then $G$ admits a $T/T^\dagger$ implementation of CCZ
(Definition~\ref{def:exp-CI}) if and only if the system
\begin{equation}\label{eq:exp-Z4half}
\gamma\cdot c\;\equiv\;\tfrac{1}{2}\lvert c\rvert-2x_1x_2x_3\pmod 4,
\qquad c=K^{\!\top}\vec x+S^{\!\top}\vec y,
\end{equation}
has a solution with $\gamma_j\in\{0,1\}$ for all~$j$.
\end{corollary}

\begin{proof}
Write $\Gamma_j=1-2\gamma_j$ with $\gamma_j\in\{0,1\}$.
Then $\mathrm{sw}_\Gamma(c)=\lvert c\rvert-2(\gamma\cdot c)$,
and condition~\eqref{eq:exp-CI} reads
$\lvert c\rvert-2(\gamma\cdot c)\equiv4x_1x_2x_3\pmod 8$.
Conditions K1 and S1 make every generator row, and hence every codeword $c$,
even. Both sides of Equation~\eqref{eq:exp-CI} are therefore
even, and division by $2$ gives~\eqref{eq:exp-Z4half}.
The left side $\sum_j c_j\gamma_j$ is linear over $\Z_4$ in $\gamma$, and the right
side is a constant for each codeword.
\end{proof}
Step~5 of Algorithm~\ref{alg:exp-search} solves~\eqref{eq:exp-Z4half} for a binary
$\gamma$ and accepts a code only when a solution exists (Appendix~\ref{app:exp-search}).
The $\nkd{47}{3}{3}$ code of~\cite{jacintoCompactFactories2026} satisfies the CH conditions but
fails this test:
\begin{example}\label{prop:exp-nonnative47}
The $\nkd{47}{3}{3}$ code of~\cite{jacintoCompactFactories2026} satisfies the nine CH conditions
($d_X=16$, $d_Z=3$, $s=6$, projective support). For it, Equation~\eqref{eq:exp-Z4half} is a
$\Z_4$-linear system in $47$ unknowns, one equation per codeword, and elimination over $\Z_4$
shows it inconsistent. So the code admits no $T/T^\dagger$ implementation of transversal CCZ, and
any diagonal Clifford correction for it needs entangling $CZ$ gates.
\end{example}

For FG48, the following identity reduces the signed-weight calculation to
Hamming weights on the three coset blocks of Corollary~\ref{cor:exp-blockform}.

\begin{restatable}{proposition}{ExpBlockCriterion}
\label{prop:exp-blockcriterion}
Let $c=K^{\!\top}\vec x+S^{\!\top}\vec y$, and let $c_b$ and $\gamma_b$ be the
restrictions of $c$ and the flip vector $\gamma$ to block~$b$ of
Corollary~\ref{cor:exp-blockform}. If $\Gamma_j=(-1)^{\gamma_j}$, then
\begin{equation}\label{eq:exp-blockcriterion}
\mathrm{sw}_\Gamma(c)
=\sum_b\bigl(\lvert \gamma_b+c_b\rvert-\lvert \gamma_b\rvert\bigr).
\end{equation}
Thus the pattern specified by $\Gamma$ is single-qubit transversal if and only if, for every $\vec x,\vec y$,
\begin{equation}\label{eq:exp-blockcriterion-mod16}
\sum_b\lvert \gamma_b+c_b\rvert\equiv
\sum_b\lvert \gamma_b\rvert+4x_1x_2x_3\pmod 8.
\end{equation}
If all conditions in Proposition~\ref{prop:exp-signed-overlaps} that contain a stabilizer row hold,
the left hand side is independent of $\vec y$ modulo $8$, and it suffices to check the following
equation for the eight $\vec x\in\F_2^3$:
\begin{equation}\label{eq:exp-block-logical-criterion}
\sum_b\lvert \gamma_b+(K^{\!\top}\vec x)_b\rvert
\equiv\sum_b\lvert \gamma_b\rvert+4x_1x_2x_3\pmod 8.
\end{equation}

Equations~\eqref{eq:exp-blockcriterion}--\eqref{eq:exp-block-logical-criterion} apply to the
Boolean block restrictions $c_b$ and $\gamma_b$ with no degree restriction.
\end{restatable}

To apply Proposition~\ref{prop:exp-blockcriterion} to FG48, take $S=\SF$
from~\eqref{eq:exp-blockform} and the logical rows $K$ from Table~\ref{tab:exp-labelmap}.
The resulting $G=[K;S]$ has $\operatorname{rank}G=9$,
$\operatorname{rank}S=6$, $d_X=16$, and $d_Z=3$.

With the stabilizer rows $S$ and logical rows $K$ fixed, we choose the
$T/T^\dagger$ exponents $\Gamma_j=(-1)^{\gamma_j}$ by solving
Equation~\eqref{eq:exp-Z4half} for the flip vector $\gamma\in\F_2^{48}$.
The polynomial $\gamma(u,w)$ below specifies this vector.

\begin{restatable}{proposition}{ExpQFortyEightPattern}
\label{prop:exp-q48-tpattern}\label{prop:exp-tpattern}
Index the columns of $\SF$ by their coordinates $(u,w)$ of~\eqref{eq:exp-blockform},
writing $u=(u_1,u_2,u_3,u_4)$ for the four rows of $A$ and $w$ for the coset label.
The $T/T^{\dagger}$ pattern of $\Qcode$ is obtained by evaluating the following polynomial.
Write $\mathrm{ev}_{\mathcal P}$ for the
evaluation map $p\mapsto(p(\beta))_{\beta\in\mathcal P}$.
\begin{equation}\label{eq:exp-tpattern}
\gamma(u,w)=u_3+w_1Q_1(u)+w_2Q_2(u),\qquad
\begin{aligned}
Q_1&=u_1+u_2+u_3+u_1u_3+u_2u_3+u_2u_4,\\
Q_2&=u_3+u_4+u_1u_2+u_1u_4+u_3u_4,
\end{aligned}
\end{equation}
The flip vector is $\gamma=\mathrm{ev}_{\mathcal P}(\gamma(u,w))\in\F_2^{48}$.
The flip vector has weight $22$. The resulting $26\,T+22\,T^\dagger$ pattern is a single-qubit transversal
implementation of logical CCZ for $\Qcode$.
\end{restatable}

The proof is given in Appendix~\ref{app:exp-tpattern-proof}.

\begin{figure}[H]
\centering
{\tikzset{every picture/.style={}}%
\begin{tikzpicture}[
  cellline/.style={draw=black!60,line width=0.4pt},
  axis/.style={font=\scriptsize,text=black!70},
  panelhead/.style={font=\small},
  panelnote/.style={font=\small},
  rowlabel/.style={font=\scriptsize},
]
\colorlet{tfill}{red!14}
\colorlet{dagfill}{teal!32}
\node[panelhead] at (2.12,0.58) {$w=(1,0)$};
\draw[cellline,fill=tfill] (0.00,-1.06) rectangle (1.06,0.00);
\node at (0.53,-0.53) {$T$};
\draw[cellline,fill=tfill] (1.06,-1.06) rectangle (2.12,0.00);
\node at (1.59,-0.53) {$T$};
\draw[cellline,fill=tfill] (2.12,-1.06) rectangle (3.18,0.00);
\node at (2.65,-0.53) {$T$};
\draw[cellline,fill=tfill] (3.18,-1.06) rectangle (4.24,0.00);
\node at (3.71,-0.53) {$T$};
\draw[cellline,fill=dagfill] (0.00,-2.12) rectangle (1.06,-1.06);
\node at (0.53,-1.59) {$T^\dagger$};
\draw[cellline,fill=tfill] (1.06,-2.12) rectangle (2.12,-1.06);
\node at (1.59,-1.59) {$T$};
\draw[cellline,fill=dagfill] (2.12,-2.12) rectangle (3.18,-1.06);
\node at (2.65,-1.59) {$T^\dagger$};
\draw[cellline,fill=tfill] (3.18,-2.12) rectangle (4.24,-1.06);
\node at (3.71,-1.59) {$T$};
\draw[cellline,fill=tfill] (0.00,-3.18) rectangle (1.06,-2.12);
\node at (0.53,-2.65) {$T$};
\draw[cellline,fill=dagfill] (1.06,-3.18) rectangle (2.12,-2.12);
\node at (1.59,-2.65) {$T^\dagger$};
\draw[cellline,fill=dagfill] (2.12,-3.18) rectangle (3.18,-2.12);
\node at (2.65,-2.65) {$T^\dagger$};
\draw[cellline,fill=tfill] (3.18,-3.18) rectangle (4.24,-2.12);
\node at (3.71,-2.65) {$T$};
\draw[cellline,fill=dagfill] (0.00,-4.24) rectangle (1.06,-3.18);
\node at (0.53,-3.71) {$T^\dagger$};
\draw[cellline,fill=dagfill] (1.06,-4.24) rectangle (2.12,-3.18);
\node at (1.59,-3.71) {$T^\dagger$};
\draw[cellline,fill=tfill] (2.12,-4.24) rectangle (3.18,-3.18);
\node at (2.65,-3.71) {$T$};
\draw[cellline,fill=tfill] (3.18,-4.24) rectangle (4.24,-3.18);
\node at (3.71,-3.71) {$T$};
\node[axis] at (-0.32,-0.53) {00};
\node[axis] at (-0.32,-1.59) {01};
\node[axis] at (-0.32,-2.65) {11};
\node[axis] at (-0.32,-3.71) {10};
\node[axis] at (0.53,-4.54) {00};
\node[axis] at (1.59,-4.54) {01};
\node[axis] at (2.65,-4.54) {11};
\node[axis] at (3.71,-4.54) {10};
\node[panelhead] at (7.16,0.58) {$w=(0,1)$};
\draw[cellline,fill=tfill] (5.04,-1.06) rectangle (6.10,0.00);
\node at (5.57,-0.53) {$T$};
\draw[cellline,fill=dagfill] (6.10,-1.06) rectangle (7.16,0.00);
\node at (6.63,-0.53) {$T^\dagger$};
\draw[cellline,fill=tfill] (7.16,-1.06) rectangle (8.22,0.00);
\node at (7.69,-0.53) {$T$};
\draw[cellline,fill=tfill] (8.22,-1.06) rectangle (9.28,0.00);
\node at (8.75,-0.53) {$T$};
\draw[cellline,fill=tfill] (5.04,-2.12) rectangle (6.10,-1.06);
\node at (5.57,-1.59) {$T$};
\draw[cellline,fill=dagfill] (6.10,-2.12) rectangle (7.16,-1.06);
\node at (6.63,-1.59) {$T^\dagger$};
\draw[cellline,fill=tfill] (7.16,-2.12) rectangle (8.22,-1.06);
\node at (7.69,-1.59) {$T$};
\draw[cellline,fill=tfill] (8.22,-2.12) rectangle (9.28,-1.06);
\node at (8.75,-1.59) {$T$};
\draw[cellline,fill=dagfill] (5.04,-3.18) rectangle (6.10,-2.12);
\node at (5.57,-2.65) {$T^\dagger$};
\draw[cellline,fill=dagfill] (6.10,-3.18) rectangle (7.16,-2.12);
\node at (6.63,-2.65) {$T^\dagger$};
\draw[cellline,fill=tfill] (7.16,-3.18) rectangle (8.22,-2.12);
\node at (7.69,-2.65) {$T$};
\draw[cellline,fill=dagfill] (8.22,-3.18) rectangle (9.28,-2.12);
\node at (8.75,-2.65) {$T^\dagger$};
\draw[cellline,fill=tfill] (5.04,-4.24) rectangle (6.10,-3.18);
\node at (5.57,-3.71) {$T$};
\draw[cellline,fill=tfill] (6.10,-4.24) rectangle (7.16,-3.18);
\node at (6.63,-3.71) {$T$};
\draw[cellline,fill=dagfill] (7.16,-4.24) rectangle (8.22,-3.18);
\node at (7.69,-3.71) {$T^\dagger$};
\draw[cellline,fill=tfill] (8.22,-4.24) rectangle (9.28,-3.18);
\node at (8.75,-3.71) {$T$};
\node[axis] at (4.72,-0.53) {00};
\node[axis] at (4.72,-1.59) {01};
\node[axis] at (4.72,-2.65) {11};
\node[axis] at (4.72,-3.71) {10};
\node[axis] at (5.57,-4.54) {00};
\node[axis] at (6.63,-4.54) {01};
\node[axis] at (7.69,-4.54) {11};
\node[axis] at (8.75,-4.54) {10};
\node[panelhead] at (12.20,0.58) {$w=(1,1)$};
\draw[cellline,fill=tfill] (10.08,-1.06) rectangle (11.14,0.00);
\node at (10.61,-0.53) {$T$};
\draw[cellline,fill=dagfill] (11.14,-1.06) rectangle (12.20,0.00);
\node at (11.67,-0.53) {$T^\dagger$};
\draw[cellline,fill=dagfill] (12.20,-1.06) rectangle (13.26,0.00);
\node at (12.73,-0.53) {$T^\dagger$};
\draw[cellline,fill=dagfill] (13.26,-1.06) rectangle (14.32,0.00);
\node at (13.79,-0.53) {$T^\dagger$};
\draw[cellline,fill=dagfill] (10.08,-2.12) rectangle (11.14,-1.06);
\node at (10.61,-1.59) {$T^\dagger$};
\draw[cellline,fill=dagfill] (11.14,-2.12) rectangle (12.20,-1.06);
\node at (11.67,-1.59) {$T^\dagger$};
\draw[cellline,fill=tfill] (12.20,-2.12) rectangle (13.26,-1.06);
\node at (12.73,-1.59) {$T$};
\draw[cellline,fill=dagfill] (13.26,-2.12) rectangle (14.32,-1.06);
\node at (13.79,-1.59) {$T^\dagger$};
\draw[cellline,fill=dagfill] (10.08,-3.18) rectangle (11.14,-2.12);
\node at (10.61,-2.65) {$T^\dagger$};
\draw[cellline,fill=tfill] (11.14,-3.18) rectangle (12.20,-2.12);
\node at (11.67,-2.65) {$T$};
\draw[cellline,fill=tfill] (12.20,-3.18) rectangle (13.26,-2.12);
\node at (12.73,-2.65) {$T$};
\draw[cellline,fill=tfill] (13.26,-3.18) rectangle (14.32,-2.12);
\node at (13.79,-2.65) {$T$};
\draw[cellline,fill=dagfill] (10.08,-4.24) rectangle (11.14,-3.18);
\node at (10.61,-3.71) {$T^\dagger$};
\draw[cellline,fill=dagfill] (11.14,-4.24) rectangle (12.20,-3.18);
\node at (11.67,-3.71) {$T^\dagger$};
\draw[cellline,fill=tfill] (12.20,-4.24) rectangle (13.26,-3.18);
\node at (12.73,-3.71) {$T$};
\draw[cellline,fill=dagfill] (13.26,-4.24) rectangle (14.32,-3.18);
\node at (13.79,-3.71) {$T^\dagger$};
\node[axis] at (9.76,-0.53) {00};
\node[axis] at (9.76,-1.59) {01};
\node[axis] at (9.76,-2.65) {11};
\node[axis] at (9.76,-3.71) {10};
\node[axis] at (10.61,-4.54) {00};
\node[axis] at (11.67,-4.54) {01};
\node[axis] at (12.73,-4.54) {11};
\node[axis] at (13.79,-4.54) {10};
\end{tikzpicture}}
\caption{The $T/T^{\dagger}$ pattern of $\Qcode$ in closed form (Proposition~\ref{prop:exp-q48-tpattern}). The three panels show the flip vector $\gamma(u,w)=u_3+w_1Q_1(u)+w_2Q_2(u)$ on the three cosets of~\eqref{eq:exp-blockform}. Each panel lists the $16$ points $u=(u_1,u_2,u_3,u_4)$ of one coset, rows indexed by $(u_1,u_2)$ and columns by $(u_3,u_4)$, with $T$ at $\gamma=0$ and $T^\dagger$ at $\gamma=1$. The flip counts are $6$, $6$, $10$, resulting in a total of $26\,T$ and $22\,T^\dagger$.}
\label{fig:exp-tpattern-grid}

\end{figure}

The code $\Qcode$ is not an isolated example. Building $S$ from a larger subspace
complement, as in Corollary~\ref{cor:exp-blockform}, also admits CSS codes with a $T/T^\dagger$
implementation of logical CCZ. Table~\ref{tab:exp-family} lists, for each $(p,q)$ with $s=p+q\le9$,
the codes of this kind that we found. Each of these codes pairs the
subspace-complement matrix $S$ of Corollary~\ref{cor:exp-blockform} with three logical rows $K$
that satisfy the CH conditions. For $\Qcode$, the SAT encoding of
Section~\ref{sec:exp-intro} found the logical rows. For the longer codes, the logical rows come
from structured label maps $\kappa$ (explicit polynomials in the block coordinates $(u,w)$, and
maps invariant under a cyclic symmetry of the support), from random label maps of low degree, and
from lifts of codes on smaller supports. A candidate was kept when it satisfied the CH conditions
and Equation~\eqref{eq:exp-Z4half} had a binary solution. The exhaustive pair enumeration of
Algorithm~\ref{alg:exp-search} was not used for these codes.

The stabilizer code FG48 is unique by
Proposition~\ref{prop:exp-uniqueness}, but the logical rows are not. FG48 has exactly two
inequivalent sets of compatible $K$ rows with a single-qubit transversal CCZ, and $\Qcode$ is one of them. For the other supports
the numbers of inequivalent codes are lower bounds. Supports with $p\le3$ do not admit any valid code. The supports with
$(p,q)=(4,1)$ and $(5,1)$ have $n\le38$ and are excluded by Theorem~\ref{thm:exp-floor38}.

\begin{table}[H]
\centering
\small
\renewcommand{\arraystretch}{1.2}
\begin{tabular}{@{}c|ccccc@{}}
\toprule
$p\backslash q$ & $1$ & $2$ & $3$ & $4$ & $5$\\
\midrule
$4$ & none & \begin{tabular}[c]{@{}c@{}}$\nkd{48}{3}{3}$\\ $2$ codes$^{\star}$\end{tabular}
 & \begin{tabular}[c]{@{}c@{}}$\nkd{112}{3}{3}$\\ $12$ codes\end{tabular}
 & \begin{tabular}[c]{@{}c@{}}$\nkd{240}{3}{3}$\\ $17$ codes\end{tabular}
 & \begin{tabular}[c]{@{}c@{}}$\nkd{496}{3}{3}$\\ $18$ codes\end{tabular}\\[3pt]
$5$ & none & \begin{tabular}[c]{@{}c@{}}$\nkd{96}{3}{3}$\\ $11$ codes\end{tabular}
 & \begin{tabular}[c]{@{}c@{}}$\nkd{224}{3}{3}$\\ $8$ codes\end{tabular}
 & \begin{tabular}[c]{@{}c@{}}$\nkd{480}{3}{3}$\\ $14$ codes\end{tabular} & \\[3pt]
$6$ & \begin{tabular}[c]{@{}c@{}}$\nkd{64}{3}{4}$\\ $5$ codes\end{tabular}
 & \begin{tabular}[c]{@{}c@{}}$\nkd{192}{3}{3}$\\ $14$ codes\end{tabular}
 & \begin{tabular}[c]{@{}c@{}}$\nkd{448}{3}{3}$\\ $13$ codes\end{tabular} & & \\[3pt]
$7$ & \begin{tabular}[c]{@{}c@{}}$\nkd{128}{3}{4}$\\ $3$ codes\end{tabular}
 & \begin{tabular}[c]{@{}c@{}}$\nkd{384}{3}{3}$\\ $18$ codes\end{tabular} & & & \\[3pt]
$8$ & \begin{tabular}[c]{@{}c@{}}$\nkd{256}{3}{4}$\\ $3$ codes\end{tabular} & & & & \\
\bottomrule
\end{tabular}
\caption{CSS codes with a $T/T^\dagger$ implementation of logical CCZ whose stabilizer
code is the subspace complement $\PG(p+q-1,2)\setminus\PG(p-1,2)$ of Corollary~\ref{cor:exp-blockform},
for $p+q\le9$. Each cell gives the parameters $\nkd{n}{3}{d}$, where $d=d_Z<d_X$, and the number of
inequivalent codes found. Two codes count as inequivalent when their distances or the weight
enumerators of their codes $\mathcal C_1$ differ. The count marked $\star$ is complete, and the
others are lower bounds. ``None'' marks the supports excluded by Theorem~\ref{thm:exp-floor38}.}
\label{tab:exp-family}
\end{table}

In Section~\ref{sec:exp-minimality}, we use the nine CH conditions to derive a lower bound on block-length.

\section{A length lower bound}\label{sec:exp-minimality}

For the logical $T$ gate, Koutsioumpas, Banfield, and Kay prove the analogous statement that the
$15$-qubit Reed--Muller code is the smallest distance-$3$ code with a transversal non-Clifford
gate~\cite{koutsioumpasSmallestCodeTransversal2022}.  Here, we show that no CSS code with $k=3$ and $d_Z\ge3$ satisfies the nine CH conditions at any length $n\le38$ and
any stabilizer rank (Table~\ref{tab:exp-lengths}).

\begin{table}[!htbp]
\centering
\begin{tabular}{@{}l l l@{}}
\toprule
lengths & nine CH conditions & evidence\\
\midrule
$n\le38$ & met by no code & Theorem~\ref{thm:exp-floor38}\\
$39\le n\le46$ & met by no code & Wills, Jain, and Singh~\cite{willsClassificationGeneralised2026}\\
$n=47$ & met, no $T/T^\dagger$ implementation & Example~\ref{prop:exp-nonnative47}\\
$n=48$ & met, $T/T^\dagger$ implementation & $\Qcode$\\
\bottomrule
\end{tabular}
\caption{Lengths settled by our lower bound and the concurrent classification, with the two nearest comparison points.}
\label{tab:exp-lengths}
\end{table}

To exclude smaller codes, we first restrict the possible stabilizer supports using
S1--S3, then ask whether any choice of logical rows satisfies K1--K3 and KS1--KS3.
Distinct nonzero stabilizer columns guarantee $d_Z\ge3$ by Lemma~\ref{lem:exp-fiber}.
For $n\le38$, Lemma~\ref{lem:exp-levelred} reduces the problem to $s\le8$.
Appendix~\ref{app:exp-decide-one} covers the supports at these ranks and shows that none admits
logical rows satisfying all six remaining CH conditions.

\subsection{Admissible block lengths}\label{sec:exp-spectrum}

We first discuss constraints on the block length $n$ by conditions S1--S3.

\begin{restatable}[Admissible block lengths]{theorem}{ExpSpectrum}\label{thm:exp-spectrum}
Conditions S1--S3 hold if and only if $\chi_S=c$ or $\chi_S=c+\delta_0$ for some
$c\in\RM(s-4,s)$, where $\delta_0$ is the indicator of the origin, that is,
$\chi_S\in\RM(s-4,s)\cup\bigl(\delta_0+\RM(s-4,s)\bigr)$. The lengths allowed by S1--S3 together with
$\operatorname{rank}S=s$ are exactly $n=\lvert c\rvert-c(0)$ for $c\in\RM(s-4,s)$ satisfying
\[
\operatorname{span}\{\beta\in\F_2^s\setminus\{0\}:c(\beta)=1\}=\F_2^s.
\]
\end{restatable}
The proof is given in Appendix~\ref{app:exp-support-lengths}.

\begin{table}[htbp]
\centering
\small
\begin{tabular}{@{}c l l l l@{}}
\toprule
$s$ & Reed--Muller code & relevant $|c|\le39$ & candidate $15\le n\le38$ & source\\
\midrule
$5$ & $\RM(1,5)$ & $16,32$ & $15,16,31$ & \cite[Ch.~13]{macwilliams1977theory}\\
$6$ & $\RM(2,6)$ & $16,24,28,32,36$ & $15,16,23,24,27,28,31,32,35,36$ & \cite{sloaneBerlekamp1970}\\
$7$ & $\RM(3,7)$ & $16,24,28,32,36$ & $15,16,23,24,27,28,31,32,35,36$ & \cite{kasamiTokura1970,kasamiTokuraAzumi1976,carletBooleanFunctions2021}\\
\bottomrule
\end{tabular}
\caption{Relevant Reed--Muller weights and the candidate block lengths obtained from
Theorem~\ref{thm:exp-spectrum}. Every omitted length at $s=5,6,7$ in the range $15\le n\le38$
is excluded by S1--S3. A listed length has stabilizer rank $s$ only when its support spans
$\F_2^s$.}
\label{tab:exp-spectrum}
\end{table}

The $s=5$ row follows from the weights of first-order Reed--Muller codewords.
The closed weight-enumerator formula for second-order Reed--Muller codes gives the $s=6$
row~\cite{sloaneBerlekamp1970}. For $s=7$, three sources cover three ranges of
weights. Kasami--Tokura~\cite{kasamiTokura1970} determine the weights of $\RM(3,7)$ below
$2d=32$, which are $16$, $24$, and $28$. Ax's theorem~\cite[Sec.~4.1.5]{carletBooleanFunctions2021}
makes every weight of $\RM(3,7)$ divisible by $4$, so $32$ and $36$ are the only possible weights
from $32$ to $39$. Kasami--Tokura--Azumi~\cite{kasamiTokuraAzumi1976} classify the codewords of
weight below $2.5d=40$, which supplies the normal forms at these weights used in
Appendix~\ref{app:exp-decide-one}.

The candidate column of Table~\ref{tab:exp-spectrum} lists, for each weight $|c|$, the
lengths $n=|c|$ and $n=|c|-1$ given by Theorem~\ref{thm:exp-spectrum}. Lemma~\ref{lem:exp-flatobstruction}
excludes $n=15,16$ because K3 is unsatisfiable on the supports defined below.

\begin{restatable}{lemma}{ExpFlat}\label{lem:exp-flatobstruction}
Call a set an \emph{affine flat of dimension $m$} if it equals $v+U$ for a vector
$v\in\F_2^s$ and an $m$-dimensional linear subspace $U\subseteq\F_2^s$ (so the set has $2^m$
elements). If the support $\mathcal P$ is an affine flat of dimension $m\le4$, or such a flat with
the origin deleted, no label map satisfies the CH conditions, at any stabilizer rank $s$. In
particular, $n\in\{15,16\}$ are excluded uniformly in $s$.
\end{restatable}
The proof is given in Appendix~\ref{app:exp-flat-supports}.

\subsection{Excluding lengths below 39}\label{sec:exp-floor-sub}

\begin{restatable}[Stabilizer-rank lower bound]{proposition}{ExpLB}\label{prop:exp-lb}
For $k=3$ and $d_Z\ge3$, no multiplicity-free code (Lemma~\ref{lem:exp-fiber}) satisfying the CH
conditions exists with $s\le5$ (equivalently $r=k+s\le8$). Hence every such code has $s\ge6$.
\end{restatable}
The proof is given in Appendix~\ref{app:exp-s5}.

\begin{restatable}[Stabilizer-rank reduction]{lemma}{ExpLevelred}\label{lem:exp-levelred}
Let $G=[K;S]$ satisfy the nine CH conditions with $k=3$, $d_Z\ge3$, and
$\mathcal C_2$ projective (Section~\ref{sec:exp-fg48}) at stabilizer rank $s$,
meaning that $S$ has $s$ independent rows. If
$2^s-1>n+\lfloor\binom{n}{2}/2\rfloor$, then $G$ reduces to
an instance with strictly smaller stabilizer rank
of length at most $n$ with $k=3$, all nine CH conditions,
$d_Z\ge3$, and $\mathcal C_2$ projective.
\end{restatable}
The proof is given in Appendix~\ref{app:exp-levelred-proof}.

For $n\le38$, Proposition~\ref{prop:exp-lb} excludes $s\le5$, and
Lemma~\ref{lem:exp-levelred} reduces $s\ge9$, leaving $s=6,7,8$ for the finite
support and logical-label analysis in Appendix~\ref{app:exp-decide-one}.

The low-weight results of~\cite{kasamiTokura1970,kasamiTokuraAzumi1976}
provide normal forms for the Reed--Muller codeword $c$ below the bound
$\lvert c\rvert<2.5d=40$, where $d=16$ is the minimum distance of $\RM(s-4,s)$.
To apply these normal forms to every candidate with $n\le38$, we must check
that the corresponding codeword $c$ lies below this weight bound.
Theorem~\ref{thm:exp-spectrum} gives $\lvert c\rvert=n$ or $n+1$, so either choice
has weight at most $39$ and is covered by the classification.
Appendix~\ref{app:exp-decide-one} uses the applicable normal forms to cover the possible supports,
then rules out logical rows satisfying K1--K3 and KS1--KS3 on each support.

\begin{restatable}[Length lower bound]{theorem}{ExpFloor}\label{thm:exp-floor38}
No CSS code with $k=3$ and $d_Z\ge3$ satisfying the nine CH conditions
(Theorem~\ref{thm:exp-NS}) exists with $n\le38$, at any stabilizer rank. \end{restatable}
The proof is given in Appendix~\ref{app:exp-decide-one}.

Figure~\ref{fig:exp-floor-parameters} summarizes the exclusions by block
length and stabilizer rank.

\begin{figure}[tbp]
\centering
\begingroup
\tikzset{every picture/.style={}}
\begin{tikzpicture}[x=.46cm,y=.82cm,font=\small]
  \fill[black!8] (14.5,-.42) rectangle (38.5,.42);
  \fill[blue!12] (14.5,3.58) rectangle (38.5,4.42);
  \fill[blue!12] (14.5,2.58) rectangle (30.5,3.42);
  \fill[orange!10] (30.5,2.58) rectangle (38.5,3.42);
  \foreach \y in {1,2} {
    \foreach \a/\b in {16.5/22.5,24.5/26.5,28.5/30.5,32.5/34.5,36.5/38.5}
      \fill[black!8] (\a,\y-.32) rectangle (\b,\y+.32);
    \draw[black!20] (14.5,\y) -- (38.5,\y);
    \foreach \n in {24,28,32,36}
      \fill (\n,\y) circle[radius=1.8pt];
    \foreach \n in {23,27,31,35}
      \draw[fill=white,line width=.65pt] (\n,\y) circle[radius=1.8pt];
    \foreach \n in {15,16} {
      \draw[red!75!black,line width=1pt] (\n-.18,\y-.09) -- (\n+.18,\y+.09);
      \draw[red!75!black,line width=1pt] (\n-.18,\y+.09) -- (\n+.18,\y-.09);
    }
  }
  \foreach \n in {32,34,36,38}
    \fill[orange!75!black] (\n,3) circle[radius=1.8pt];
  \foreach \n in {31,33,35,37}
    \draw[orange!75!black,fill=white,line width=.65pt] (\n,3) circle[radius=1.8pt];
  \node at (26.5,0) {Excluded: Proposition~\ref{prop:exp-lb}};
  \node at (26.5,4) {Reduce rank: Lemma~\ref{lem:exp-levelred} (Step~1)};
  \node at (22.5,3) {Reduce rank: Lemma~\ref{lem:exp-levelred} (Step~3)};
  \draw[-{Stealth[length=2mm]}] (14.5,-.65) -- (40,-.65) node[right] {$n$};
  \foreach \n in {15,20,25,30,35,38} {
    \draw (\n,-.65) -- (\n,-.76);
    \node[below] at (\n,-.76) {$\n$};
  }
  \node[anchor=east] at (14,4.8) {$s$};
  \foreach \y/\r in {0/{\le5},1/6,2/7,3/8,4/{\ge9}}
    \node[anchor=east] at (14,\y) {$\r$};
\end{tikzpicture}
\par\smallskip
{\footnotesize
\begin{tabular}{@{}>{\raggedright\arraybackslash}p{3.1cm}>{\raggedright\arraybackslash}p{10.4cm}@{}}
Gray band at $s\le5$ & Excluded by Proposition~\ref{prop:exp-lb}.\\
Blue band at $s\ge9$ & Reduced by Lemma~\ref{lem:exp-levelred} after
Equation~\eqref{eq:exp-floor-ranks} (Step~1).\\
Blue band at $s=8$ & Reduced by Lemma~\ref{lem:exp-levelred} after
Equation~\eqref{eq:exp-rank8-small} (Step~3).\\
$s=6,7$ gray intervals & Excluded by Theorem~\ref{thm:exp-spectrum} and
Table~\ref{tab:exp-spectrum}.\\
Filled / open circles & $n=|c|$ / $n=|c|-1$ by Theorem~\ref{thm:exp-spectrum},
before rank, geometry, and label checks.\\
Crosses ($n=15,16$) & Excluded by
Lemma~\ref{lem:exp-flatobstruction} (Step~2).\\
Remaining $s=6,7$ circles & Table~\ref{tab:exp-floorcells} records their dispositions
in its $s=6,7$ rows (Step~2).\\
Orange region and circles & The $s=8$ rows of Table~\ref{tab:exp-floorcells} cover the
weights retained by Equation~\eqref{eq:exp-rank8-small} and Steps~3--7.\\
\end{tabular}}
\endgroup
\caption{Arguments excluding $15\le n\le38$ at each stabilizer rank.
Step numbers refer to
Appendix~\ref{app:exp-decide-one}.}
\label{fig:exp-floor-parameters}
\end{figure}

Our argument proves the lower bound through $n=38$ independently of the concurrent
classification of Wills, Jain, and Singh~\cite{willsClassificationGeneralised2026}.
Their classification excludes the remaining lengths $39\le n\le46$ for quasi-transversal CCZ
with $k=3$ and $d_Z\ge3$.

\section{Discussion and open problems}\label{sec:exp-conclusion}
In this work, we set out to construct a small CSS code whose non-Clifford logical
CCZ is realized by physical $T/T^\dagger$ gates.
Applying CCZ directly to the encoded block provides an alternative to magic-state injection.
Projectivity of $\mathcal C_2$ guarantees $d_Z\ge3$, and $8$-divisibility enforces
the stabilizer-only CH conditions. The logical rows must satisfy the remaining CH conditions, while
Equation~\eqref{eq:exp-signed-overlap-conditions} constrains the physical signs.

The same codes also give magic-state distillation
protocols~\cite{bravyiMagicStateDistillation2012,campbellUnifiedFrameworkMagic2017}. Applying the
physical rotations to an encoded $|{+}\rangle^{\otimes3}$ with $n$ noisy $T$ gates, then measuring the
stabilizers and decoding, yields a CCZ state. Because $d_Z\ge3$, an error from one or two faulty
$T$ gates is either detected or acts trivially. Distillation permits a diagonal Clifford
correction, so the nine CH conditions suffice there. Jacinto et al.\ use the $\nkd{47}{3}{3}$ code
for a distance-$3$, $47$-to-$1$ CCZ distillation
protocol~\cite{jacintoCompactFactories2026}.

At length $48$, Proposition~\ref{prop:exp-uniqueness} identifies FG48 as the unique
projective $8$-divisible $[48,6]$ code up to a permutation of coordinates.
The code $\Qcode$ completes this support to a $\nkd{48}{3}{3}$ code with a single-qubit transversal logical CCZ.
Proposition~\ref{prop:exp-q48-tpattern} verifies the $26T+22T^\dagger$ pattern from its explicit
polynomial formula.

To determine how small such codes can be, Theorem~\ref{thm:exp-floor38} excludes
every CSS code with three logical qubits, $d_Z\ge3$, and the nine CH conditions below length $39$.
The bound applies beyond subspace-complement supports. The concurrent classification of Wills, Jain, and Singh~\cite{willsClassificationGeneralised2026}
excludes lengths $39\le n\le46$ in the nine-CH-condition class.
Together with the code of Jacinto et al., the classification establishes the minimum length
$47$ when a diagonal Clifford correction is allowed.

Several problems remain open. The subspace-complement supports in
Corollary~\ref{cor:exp-blockform} with $n\ge48$ that admit a single-qubit transversal CCZ
implementation are not yet characterized. A further direction is to adapt the construction of $\Qcode$ and the arguments
for Theorem~\ref{thm:exp-floor38} to other non-Clifford gates and to more than three logical
qubits.

\section*{Code}
The companion repository is \url{https://github.com/jerrylvx/gate2code-public}.

\section*{Acknowledgements}
This work was supported by the National Science Foundation (NSF) Software-Tailored Architecture for Quantum codesign (STAQ) Project (Phy-2325080) and the NSF Quantum Leap Challenge Institute on Robust Quantum Simulation (MCP-2120757 and MCP-2553574).

The authors want to thank Robert Calderbank for pointing us to the classical coding theory literature on two-weight and finite geometry codes and providing feedback on an earlier version of the draft. The authors also want to thank June Vanlerberghe, Arianna Meinking, and Eren Guttentag for feedback during preparation of the draft.

The authors used AI
assistants for literature review, exposition, computer simulations and proofreading. The authors independently verified all mathematical claims, constructions, and
reported numbers. The authors take full
responsibility for the content.

\appendix
\section{Shared algebra and coding theory facts}\label{app:exp-facts}

\subsection{Boolean functions and Reed--Muller codes}\label{app:exp-facts-boolean}

This subsection collects the facts about Boolean functions on $\F_2^m$ and
Reed--Muller codes that Appendices~\ref{app:exp-search}--\ref{app:exp-floor} use.

\begin{lemma}[XOR expansion]\label{lem:exp-xorlift}
For bits $a_1,\dots,a_m\in\{0,1\}$,
\begin{equation}\label{eq:exp-xorlift}
a_1\oplus\cdots\oplus a_m
\;=\;\sum_{\emptyset\ne T\subseteq\{1,\dots,m\}}(-2)^{|T|-1}\prod_{i\in T}a_i .
\end{equation}
\end{lemma}
\begin{proof}
Induct on $m$, starting from the two-bit case $a\oplus b=a+b-2ab$, with the induction hypothesis
\begin{equation*}
a_1\oplus\cdots\oplus a_{m-1}
\;=\;\sum_{\emptyset\ne T\subseteq\{1,\dots,m-1\}}(-2)^{|T|-1}\prod_{i\in T}a_i .
\end{equation*}
Let $a=a_m,\ b = a_1\oplus\cdots\oplus a_{m-1}$ and apply the two-bit case.
\end{proof}

\begin{corollary}[XOR and Hamming weight]\label{cor:exp-xorweight}
For bits $a,b\in\{0,1\}$, $2ab=a+b-(a\oplus b)$.
Consequently, if $\xi$ and $\eta$ are binary strings of the same length, $\xi\oplus\eta$
denotes coordinatewise XOR, $\lvert \xi\rvert$ denotes Hamming weight, and $\wedge$
denotes entrywise product, then
\begin{equation}\label{eq:exp-weightxor}
\lvert \xi\oplus\eta\rvert=\lvert \xi\rvert+\lvert \eta\rvert-2\lvert\xi\wedge\eta\rvert.
\end{equation}
For binary strings $\xi,\xi',\eta,\eta'$ of the same length,
\begin{equation*}
\lvert(\xi\oplus\xi')\wedge\eta\wedge\eta'\rvert\equiv
\lvert\xi\wedge\eta\wedge\eta'\rvert+\lvert\xi'\wedge\eta\wedge\eta'\rvert\pmod2.
\end{equation*}
Setting $\eta'=\mathbf1$, the all-ones string, gives the same additivity for
$\lvert\xi\wedge\eta\rvert$, and setting $\eta=\eta'=\mathbf1$ gives it for $\lvert\xi\rvert$.
\end{corollary}
\begin{proof}
Applying $2ab=a+b-(a\oplus b)$
coordinatewise gives~\eqref{eq:exp-weightxor}, and distributivity of $\wedge$ over $\oplus$
gives the overlap-parity claim.
\end{proof}

\begin{lemma}[Coefficient extraction and support parity]\label{lem:exp-coefficient-parity}
Let $\mathbb K$ be $\Z$ or $\F_2$, and let $F\colon\F_2^m\to\mathbb K$ have the multilinear expansion
$F=\sum_{R\subseteq\{1,\ldots,m\}}c_Rm_R$ with $c_R\in\mathbb K$, where
$m_R(v)=\prod_{i\in R}v_i$. If $\mathbf1_{T'}$ is one on $T'$ and zero elsewhere, then
for every $T\subseteq\{1,\ldots,m\}$
\begin{equation}\label{eq:exp-topcoeff}
c_T=\sum_{T'\subseteq T}(-1)^{|T|-|T'|}F(\mathbf1_{T'}).
\end{equation}
In particular, for $\mathbb K=\F_2$ and $T=\{1,\ldots,m\}$, every sign equals $1$ and
$\mathbf1_{T'}$ runs over all of $\F_2^m$, so
\begin{equation}\label{eq:exp-topcoeff-parity}
c_{\{1,\ldots,m\}}=\sum_{v\in\F_2^m}F(v).
\end{equation}
\end{lemma}
\begin{proof}
Substitution of the multilinear expansion into the first sum makes the coefficient of $c_R$
equal to
$\sum_{R\subseteq T'\subseteq T}(-1)^{|T|-|T'|}=(1-1)^{|T|-|R|}$, which is one for $R=T$ and zero
otherwise. The case $\mathbb K=\F_2$ and $T=\{1,\ldots,m\}$
gives~\eqref{eq:exp-topcoeff-parity}.
\end{proof}

\begin{lemma}[Origin indicator]\label{lem:exp-origin-indicator}
\phantomsection\label{fact:exp-origin-degree}
On $\F_2^m$, let $\delta_0(v)=\prod_{i=1}^m(1+v_i)$ be the indicator of the origin.
The indicator has degree $m$, and
\begin{equation}\label{eq:exp-deltainner}
\langle\delta_0,\mathbf 1\rangle=1,
\qquad
\langle\delta_0,m_T\rangle=m_T(0)=0
\quad\text{for every nonempty }T\subseteq\{1,\ldots,m\}.
\end{equation}
\end{lemma}
\begin{proof}
The product defining $\delta_0$ contains the degree-$m$ monomial $v_1\cdots v_m$ and
equals one only at the origin. Hence its inner product with a function $f$ is $f(0)$. Equation~\eqref{eq:exp-deltainner}
follows for $f=1$ and for every nonconstant monomial.
\end{proof}

\begin{lemma}[Affine changes of Boolean functions]\label{lem:exp-affine-invariance}
An \emph{invertible affine map} on $\F_2^m$ is a map
$L\colon\F_2^m\to\F_2^m$ of the form
\begin{equation}\label{eq:exp-affine-map}
L(x)=Ax+b,
\end{equation}
where $A\in\operatorname{GL}(m,2)$ and $b\in\F_2^m$. For a set
$E\subseteq\F_2^m$, write $L(E)=\{L(x):x\in E\}$ for its direct image.
Boolean functions $c,c'\colon\F_2^m\to\F_2$ are \emph{affine-equivalent}
when $c'=c\circ L^{-1}$ for such a map $L$. The induced transformation satisfies
\begin{equation*}
\supp(c')=L(\supp(c)),\qquad |c'|=|c|,\qquad \deg c'=\deg c.
\end{equation*}
\end{lemma}
\begin{proof}
For $y\in\F_2^m$, $c'(y)=1$ if and only if $c(L^{-1}(y))=1$.
Hence $\supp(c')=\{L(x):c(x)=1\}=L(\supp(c))$, and bijectivity gives the weight
identity. Substitution of affine linear forms gives
$\deg(c\circ L^{-1})\le\deg c$, and substituting into $c=c'\circ L$ gives the
opposite inequality.
\end{proof}

\phantomsection\label{app:exp-quadratic-normal-forms}
\paragraph{Quadratic normal forms.}
Dickson's classification~\cite[Ch.~15, Thms.~4--5]{macwilliams1977theory} gives the
following representatives of quadratic Boolean functions under affine equivalence:
\begin{equation}\label{eq:exp-dickson-general}
c(L(y))=\sum_{i=1}^{t}y_{2i-1}y_{2i}+\ell(y)+\varepsilon.
\end{equation}
Here $0\le2t\le m$, $\varepsilon\in\F_2$, and $\ell$ is either zero or
$y_{2t+1}$, with the latter choice available only when $2t<m$. A nonconstant
degree-one function has the form $c(y)=a\cdot y+\varepsilon$, with $a\ne0$, and
\begin{equation}\label{eq:exp-affine-hyperplane-support}
\supp(c)=\{y\in\F_2^m:a\cdot y=1+\varepsilon\}.
\end{equation}
The support in~\eqref{eq:exp-affine-hyperplane-support} is an affine hyperplane. When the
support contains the origin, its projective
support is obtained by deleting the origin.

\paragraph{Reed--Muller duality.}
For binary Reed--Muller codes~\cite[Ch.~13, Thm.~4, p.~375]{macwilliams1977theory},
\begin{equation}\label{eq:exp-rmdual}
\RM(r,m)^\perp=\RM(m-r-1,m),\qquad\RM(r,m)=\{0\}\ \text{for }r<0.
\end{equation}

\paragraph{Minimum weight.}
\phantomsection\label{fact:exp-rm-minweight}
The minimum weight of $\RM(r,m)$ is $2^{m-r}$, and its minimum-weight codewords are
exactly the indicators of the $(m-r)$-dimensional affine
flats~\cite[Cor.~5.3.3]{assmusKeyDesigns1992}.

\subsection{Operations on codes}\label{app:exp-facts-generators}

Row operations may preserve both code spaces, whereas puncturing can change their dimensions
and distance. The results below specify the effects of these operations on the code parameters
and the CH conditions.

In Appendices~\ref{app:exp-search}--\ref{app:exp-floor}, we write a generator matrix as
$G=[K;S]\in\F_2^{(3+\widetilde s)\times n}$, with logical rows $K_1,K_2,K_3$ and stabilizer rows
$S_1,\ldots,S_{\widetilde s}$. The rows of $G$ span the code $\mathcal C_1$, and the rows of $S$
span the code $\mathcal C_2\subseteq\mathcal C_1$, as for a generator
matrix~\cite[{\S}1.2]{huffmanFundamentals2003}, except that the rows need not be independent.
The code parameters are
\begin{equation*}
s=\dim\mathcal C_2\le\widetilde s,\qquad
k=\dim\mathcal C_1-\dim\mathcal C_2.
\end{equation*}
By~\eqref{eq:exp-dZ}, the $Z$-distance $d_Z$ is the least
weight of a vector in $\mathcal C_2^\perp\setminus\mathcal C_1^\perp$. The parameters $s$, $k$,
and $d_Z$ therefore depend only on the codes $\mathcal C_2\subseteq\mathcal C_1$, and not on the
rows chosen to span them.

Column $j$ of $G$ is a full column
$(\alpha,\beta)$ as in Section~\ref{sec:exp-fg48}, with logical label $\alpha\in\F_2^3$ and
syndrome $\beta\in\F_2^{\widetilde s}$. When $\mathcal C_2$ is projective, the label map $\kappa$
of Section~\ref{sec:exp-full-g} writes each full column as $(\kappa(\beta),\beta)$ with
$\beta\in\mathcal P$, and the logical row $K_i$ is the vector
$(\kappa_i(\beta))_{\beta\in\mathcal P}$. Appendix~\ref{app:exp-search} writes this vector as
$\kappa_i$.

Appendix~\ref{app:exp-floor} works at stabilizer rank $s$, meaning that $S$ has $s$
independent rows, so that the syndromes lie in $\F_2^s$. By
Lemma~\ref{lem:exp-generator-operations}(a) and~(b) below, replacing $S$ by a row basis changes
neither the nine CH conditions nor $n$, $s$, $k$, and $d_Z$.

\begin{corollary}\label{cor:exp-ch-codes}
Let $G=[K;S]$ be the generator matrix for the codes $\mathcal C_2\subseteq\mathcal C_1$. Then $G$ satisfies the nine
CH conditions if and only if
\begin{enumerate}
\item[(i)] $|a\wedge b|$ is even for all $a,b\in\mathcal C_1$,
\item[(ii)] $|a\wedge b\wedge u|$ is even for all $a,b\in\mathcal C_1$ and
$u\in\mathcal C_2$, and
\item[(iii)] $|K_1\wedge K_2\wedge K_3|$ is odd.
\end{enumerate}
In that case every codeword of $\mathcal C_1$ has even weight, by~(i) with $a=b$, and
the following hold.
\begin{enumerate}
\item[(a)] For $x\in\F_2^3$, $y\in\F_2^{\widetilde s}$, and
$\{i,i',i''\}=\{1,2,3\}$,
\begin{equation*}
w=\sum_{h=1}^{3}x_hK_h+\sum_{l=1}^{\widetilde s}y_lS_l
\quad\Longrightarrow\quad
|w\wedge K_{i'}\wedge K_{i''}|\equiv x_i\pmod2.
\end{equation*}
\item[(b)] If $n$ is even, then the matrix $[K;S;\mathbf1^{\mathsf T}]$, with the
all-ones row appended as a stabilizer row, satisfies the nine CH conditions.
\end{enumerate}
\end{corollary}
\begin{proof}
By Corollary~\ref{cor:exp-xorweight}, the parities of $|a\wedge b|$ and $|a\wedge b\wedge u|$ are additive in each argument, so (i)
and~(ii) hold for all codewords exactly when they hold for rows. On rows, $R\wedge R=R$ reduces a
repeated row, (i) becomes K1, K2, KS1, S1, and S2, (ii) becomes KS2, KS3, and S3, and (iii) is
K3.

For part~(a), additivity expresses $|w\wedge K_{i'}\wedge K_{i''}|$ as
$\sum_hx_h|K_h\wedge K_{i'}\wedge K_{i''}|$ plus the overlap of
$\sum_ly_lS_l\in\mathcal C_2$ with $K_{i'}\wedge K_{i''}$. The terms with $h=i',i''$ are even
by~(i), the stabilizer term is even by~(ii), and~(iii) makes the term with $h=i$ equal to $x_i$
mod $2$.

For part~(b), the codes become $\mathcal C_1\cup(\mathbf1+\mathcal C_1)$ and
$\mathcal C_2\cup(\mathbf1+\mathcal C_2)$, and additivity leaves only overlaps containing
$\mathbf1$. Each equals $n$ or an overlap $|a\wedge b|$ with $a,b\in\mathcal C_1$, which is even
by hypothesis or by~(i). Condition~(iii) is unchanged.
\end{proof}

\begin{definition}[Puncturing]\label{def:exp-puncturing}
Let $\mathcal C\subseteq\F_2^n$ be a linear code, and let
$J\subseteq\{1,\ldots,n\}$ be a set of coordinate indices. The \emph{puncture} of
$\mathcal C$ on $J$ is the linear code $\mathcal C'\subseteq\F_2^{n-\lvert J\rvert}$ obtained by
deleting the coordinates in $J$ from every codeword of $\mathcal C$
~\cite[{\S}1.5.1 and Thm.~1.5.1]{huffmanFundamentals2003}.
\end{definition}

\begin{lemma}[Row operations and puncturing]\label{lem:exp-generator-operations}
\begin{enumerate}
\item[(a)] \emph{Left multiplication.} For a matrix $Q$ with $3+\widetilde s$ columns,
\begin{equation*}
\begin{aligned}
\operatorname{rowspan}QG&\subseteq\operatorname{rowspan}G,\\
\operatorname{rank}QG&=\operatorname{rank}Q
&&\text{if the rows of $G$ are independent.}
\end{aligned}
\end{equation*}
The row-space inclusion is an equality when $Q$ is invertible. Adding
one row to another, deleting a row, and replacing $S$ by $AS$ for an invertible
$\widetilde s\times\widetilde s$ matrix $A$ are left multiplications.

Suppose that $G$ satisfies
the nine CH conditions, and let $G'=[K';S']$ have logical rows $K_i'\in K_i+\mathcal C_2$ and
stabilizer rows in $\mathcal C_2$. Then $G'$ satisfies the nine CH conditions. If moreover the
rows of $S'$ span $\mathcal C_2$, then $G'$ and $G$ generate the same codes, and hence the same $n$,
$s$, $k$, and $d_Z$.
\item[(b)] \emph{Row basis.} Let the rows of $S$ indexed by an ordered set $I$ form a
basis of $\mathcal C_2$, and let $S_I$ be the $s\times n$ matrix they form. Restricting a syndrome
$\beta$ to its entries in $I$ gives the corresponding column $\beta_I\in\F_2^s$ of $S_I$.
The restriction preserves distinct and nonzero syndromes, and the columns $\beta_I$
span $\F_2^s$. Thus $S_I$ has the same actual rank as $S$ and is projective with $n$ columns
whenever $S$ is projective.
\item[(c)] \emph{Puncturing.} Let $\mathcal C'$ be the code obtained by puncturing the
set of coordinates $J$ from a code $\mathcal C\subseteq\F_2^n$. Deleting the columns indexed by
$J$ from a matrix whose rows span $\mathcal C$ gives a matrix whose rows span $\mathcal C'$, and
\begin{equation}\label{eq:exp-puncture-rank}
\dim\mathcal C-\dim\mathcal C'
=\dim\{\,c\in\mathcal C:\supp(c)\subseteq J\,\}
\le \lvert J\rvert.
\end{equation}
Write $\iota_J\colon\F_2^{n-|J|}\to\F_2^n$ for insertion of zeros
at each coordinate in $J$. Then $\iota_J(\mathcal C'^\perp)\subseteq\mathcal C^\perp$ and
$|\iota_Jv'|=|v'|$. Consequently, a generator matrix punctured on $J$ has $Z$-distance at least
$d_Z$ of the original.

Moreover, if $\sigma\in\F_2^n$ equals $e_j$ or $e_j+e_{j'}$ with $j\ne j'$
and $G\sigma=0$, then $G$ punctured on $\supp(\sigma)$ satisfies the nine CH conditions exactly
when $G$ does.
\end{enumerate}
\end{lemma}

\begin{proof}
For part~(a), row $i$ of $QG$ is $\sum_jQ_{ij}g^{\,j}$, where $g^{\,j}$ are the
rows of $G$, so $\operatorname{rowspan}QG\subseteq\operatorname{rowspan}G$, and $G=Q^{-1}(QG)$ gives equality when $Q$ is invertible. If the rows of $G$ are
independent, then the columns of $G$ span $\F_2^{3+\widetilde s}$, so the columns of $QG$ span the
column space of $Q$.

For CH preservation, the codes of $G'$ satisfy $\mathcal C_2'\subseteq\mathcal C_2$ and
$\mathcal C_1'\subseteq\mathcal C_1$, so conditions (i) and~(ii) of
Corollary~\ref{cor:exp-ch-codes} for $G'$ are instances of the same conditions for $G$. Writing
$K_i'=K_i+u_i$ with $u_i\in\mathcal C_2$, additivity writes $|K_1'\wedge K_2'\wedge K_3'|$ as
$|K_1\wedge K_2\wedge K_3|$ plus overlaps with a factor $u_i$ and two factors in $\mathcal C_1$,
which are even by~(ii). When the rows of $S'$ span $\mathcal C_2$, the codes of $G'$ equal
$\mathcal C_2\subseteq\mathcal C_1$.

For part~(b), every omitted row is a linear combination of the selected rows,
so $\beta_I$ determines every entry of $\beta$. If $\beta_I=0$, every entry of $\beta$ is zero as well. Finally,
$\operatorname{rowspan}S_I=\mathcal C_2$ gives $\operatorname{rank}S_I=s$, so its columns span
$\F_2^s$.

For part~(c), coordinate deletion is a linear map from $\mathcal C$ onto
$\mathcal C'$ whose kernel consists of the codewords supported inside $J$, and rank--nullity
gives~\eqref{eq:exp-puncture-rank}.

For $v'\in\mathcal C'^\perp$ and $c\in\mathcal C$, let $c'\in\mathcal C'$ be the punctured codeword of
$c$ on $J$, the word obtained by deleting the coordinates in $J$ from $c$. Then the coordinates in $J$ contribute nothing to $c\cdot\iota_Jv'$, so
$c\cdot\iota_Jv'=c'\cdot v'=0$. For the
distance, let $v'$ satisfy $S'v'=0$ and $K'v'\ne0$ for the punctured generator matrix $[K';S']$,
and set $v=\iota_Jv'$. Then $Sv=0$, $Kv=K'v'\ne0$, and $|v|=|v'|$, so~\eqref{eq:exp-dZ} gives
the bound.

Finally, $G\sigma=0$ says that column $j$ is zero or that columns $j$ and $j'$ are
equal. The deleted columns therefore contribute $0$, or the same bit twice, to every overlap of
rows, so every CH parity is the same with and without them.
\end{proof}

\section{SAT search and exclusion}\label{app:exp-search}
Fix a support $\mathcal P$ satisfying S1--S3.
Order the $n$ columns of $\mathcal P$ as they appear in $S$. Each logical row is the vector
\begin{equation*}
\kappa_i=\bigl(\kappa_i(\beta)\bigr)_{\beta\in\mathcal P}
\in\F_2^{\mathcal P}\cong\F_2^n,
\qquad i\in\{1,2,3\}.
\end{equation*}
As in~\eqref{eq:exp-bridge}, all inner products below are over $\mathcal P$, as indicated by the
subscript. For each logical label $\kappa_i$, K1, KS1, and KS3 give, respectively,
\begin{equation}\label{eq:exp-single-label-tests}
\langle\kappa_i,\mathbf1\rangle_{\mathcal P}=0,\qquad
\langle\kappa_i,m_{\{l\}}\rangle_{\mathcal P}=0,\qquad
\langle\kappa_i,m_{\{l,l'\}}\rangle_{\mathcal P}=0.
\end{equation}
The indices are $i\in\{1,2,3\}$,
$1\le l\le s$, and $1\le l<l'\le s$ in the last constraint.
Here $m_{\{l,l'\}}(\beta)=\beta_l\beta_{l'}$ is the product of coordinates $l$ and $l'$ of the
column $\beta\in\mathcal P$. The functions $\mathbf1$, $m_{\{l\}}$, and
$m_{\{l,l'\}}$ span $\RM(2,s)$.
Their restrictions
to $\mathcal P$ therefore span $\operatorname{ev}_{\mathcal P}(\RM(2,s))$, where
$\operatorname{ev}_{\mathcal P}(p)=(p(\beta))_{\beta\in\mathcal P}$ is the evaluation map of
Proposition~\ref{prop:exp-q48-tpattern}.
Consequently, for each $i\in\{1,2,3\}$,
\begin{equation}\label{eq:exp-VA}
\kappa_i\text{ satisfies K1, KS1, and KS3}
\quad\Longleftrightarrow\quad
\kappa_i\in V_{\mathcal P}:=
\bigl(\operatorname{ev}_{\mathcal P}(\RM(2,s))\bigr)^{\perp}
\subseteq\F_2^n.
\end{equation}
Call $(\kappa_1,\kappa_2,\kappa_3)$ a \emph{CH triple} when the three logical rows
satisfy all CH conditions on $\mathcal P$.
Equation~\eqref{eq:exp-VA}
places both in $V_{\mathcal P}$, and K2 and KS2 impose
\begin{equation}\label{eq:exp-admissible-pair}
\langle\kappa_1\kappa_2,\mathbf1\rangle_{\mathcal P}=0,\qquad
\langle\kappa_1\kappa_2,m_{\{l\}}\rangle_{\mathcal P}=0.
\end{equation}
The second constraint in~\eqref{eq:exp-admissible-pair} ranges over $1\le l\le s$.
The first two labels of every CH triple therefore form an ordered pair
$(\kappa_1,\kappa_2)\in(V_{\mathcal P}\setminus\{0\})\times(V_{\mathcal P}\setminus\{0\})$ satisfying
these equations, so enumerating all such pairs covers every possible CH triple. For a fixed pair,
the remaining conditions on $\kappa_3\in V_{\mathcal P}$ are
\begin{equation}\label{eq:exp-floor-third-label}
\begin{aligned}
\langle\kappa_3,\kappa_i\rangle_{\mathcal P}&=0 &&\text{(K2)},\\
\langle\kappa_3,m_{\{l\}}\kappa_i\rangle_{\mathcal P}&=0 &&\text{(KS2)},\\
\langle\kappa_3,\kappa_1\kappa_2\rangle_{\mathcal P}&=1 &&\text{(K3)}.
\end{aligned}
\end{equation}
The indices satisfy $i\in\{1,2\}$ and $1\le l\le s$.
Membership in $V_{\mathcal P}$ imposes K1, KS1, and KS3 on $\kappa_3$,
while~\eqref{eq:exp-floor-third-label} imposes K2, KS2, and K3.
Let $d=\dim V_{\mathcal P}=n-\dim\operatorname{ev}_{\mathcal P}(\RM(2,s))$, let
$v_1,\ldots,v_d$ be a basis of $V_{\mathcal P}$, and write the vector
$\kappa_3=\sum_{a=1}^dt_av_a$ with $t_a\in\F_2$. By
Corollary~\ref{cor:exp-xorweight}, substitution into
\eqref{eq:exp-floor-third-label} gives the binary linear system
\begin{equation}\label{eq:exp-third-label-coefficients}
\begin{aligned}
\sum_{a=1}^d t_a\langle v_a,\kappa_i\rangle_{\mathcal P}&=0,
&&i\in\{1,2\},\\
\sum_{a=1}^d t_a\langle v_a,m_{\{l\}}\kappa_i\rangle_{\mathcal P}&=0,
&&i\in\{1,2\},\ 1\le l\le s,\\
\sum_{a=1}^d t_a\langle v_a,\kappa_1\kappa_2\rangle_{\mathcal P}&=1.&
\end{aligned}
\end{equation}
Since $\kappa_1$, $\kappa_2$, and the basis vectors $v_a$ are fixed, every inner product
in~\eqref{eq:exp-third-label-coefficients} is a known bit, so~\eqref{eq:exp-third-label-coefficients}
is a linear system over $\F_2$ in $t_1,\ldots,t_d$. Gaussian elimination then decides whether the pair extends to a CH triple.
Failure to extend to a CH triple for every possible pair $\kappa_1,\kappa_2$ excludes $\mathcal P$.
When Equation~\eqref{eq:exp-third-label-coefficients} admits a solution, its solutions form an affine space, a particular solution
plus the kernel of the coefficient matrix, and each solution gives a label $\kappa_3$ that extends
the pair to a CH triple. Deciding whether $\mathcal P$ carries a CH triple needs only one solution.

\begin{algorithm}[htbp]
\caption{Finding CH label maps by SAT or testing all label maps on a specified support, followed by binary flip-vector feasibility ($k=3$, $d_Z\ge3$)}\label{alg:exp-search}
\begin{algorithmic}[1]
\Statex \textit{Find a code by SAT:}
\State For each of the $2^s-1$ nonzero $\beta$ introduce the support bit $\chi_S(\beta)$ and the
three label bits $\kappa_i(\beta)$, $i\in\{1,2,3\}$, for $4(2^s-1)$ support and logical label variables,
before encoding auxiliaries
\State For each nonzero $\beta$ and each nonempty $T_K\subseteq\{1,2,3\}$, introduce an
auxiliary variable with clauses that force it to equal
$\chi_S(\beta)\prod_{i\in T_K}\kappa_i(\beta)$. Impose K1--K3, KS1--KS3, and S1--S3 as the XOR
equations~\eqref{eq:exp-clause}, whose terms are support bits and these auxiliary variables.
Conditions S1--S3 restrict the support to
$\RM(s-4,s)\cup\bigl(\delta_0+\RM(s-4,s)\bigr)$, where $\delta_0$ is the origin indicator
\State Solve the constraints globally with a SAT solver. Each satisfying assignment
supplies $\chi_S$ and $\kappa$, and hence a generator matrix such as $\Qcode$
\Statex \textit{Alternatively, test a candidate support $\mathcal P$}:
\State For a support $\mathcal P$ satisfying S1--S3, enumerate every
ordered pair of nonzero $\kappa_1,\kappa_2\in V_{\mathcal P}$
satisfying~\eqref{eq:exp-admissible-pair}. Solve the
system~\eqref{eq:exp-third-label-coefficients} by Gaussian elimination. Each solution
$t_1,\ldots,t_d$ gives a label $\kappa_3\in V_{\mathcal P}$ and a generator matrix. Failure to
extend to a CH triple for every possible pair $\kappa_1,\kappa_2$ excludes $\mathcal P$. Exhausting all supports
satisfying S1--S3 proves non-existence
\Statex \textit{For either route, test the physical signs:}
\State For each generator matrix obtained from SAT or a consistent
fixed-support system, solve~\eqref{eq:exp-Z4half} for a binary $\gamma$.
A solution gives a $T/T^\dagger$ implementation of Definition~\ref{def:exp-CI}
\end{algorithmic}
\end{algorithm}
\FloatBarrier

\section{Proofs for the column indicators and CCZ conditions}\label{app:exp-criterion-proof}

\subsection{Proof of Lemma~\ref{lem:exp-coset}}\label{app:exp-coset-proof}
\begin{proof}

The nonconstant monomials of degree at most three are the $m_T$ with $1\le|T|\le3$.
By~\eqref{eq:exp-bridge}, $\langle\chi_G,m_T\rangle$ equals $1$ for $T=\{1,2,3\}$ and $0$ for
every other such $T$. The representative $\chi_{G,0}$ has the same inner products by its
definition in Lemma~\ref{lem:exp-coset}. Hence $\chi_G-\chi_{G,0}$ is orthogonal to every
nonconstant monomial of degree at most three. Since
$\langle\chi_{G,0},\mathbf1\rangle=\lvert\chi_{G,0}\rvert\bmod2=0$
and $\chi_G$ has weight $n$, set
\[
\varepsilon=\langle\chi_G-\chi_{G,0},\mathbf1\rangle=n\bmod2.
\]
Equation~\eqref{eq:exp-deltainner} then gives
\[
\chi_G-\chi_{G,0}-\varepsilon\delta_0
\in\RM(3,r)^\perp=\RM(r-4,r),
\]
which proves~\eqref{eq:exp-fullcoset}.
Any other admissible $\chi'_{G,0}$ has the same inner products against all monomials of
degree at most three, including the constant, so
$\chi'_{G,0}-\chi_{G,0}\in\RM(r-4,r)$ by~\eqref{eq:exp-rmdual}.
The coset is therefore independent of the representative.
\end{proof}
\subsection{Proof of Theorem~\ref{thm:exp-NS}}\label{app:exp-ch-proof}
\begin{proof}
This proof covers the necessity direction~(ii).
Assume that $C\,T^{\otimes n}$ realizes logical CCZ and set
$\vec z=(\vec x,\vec y)$, $\vec f=G^{\!\top}\vec z$, and $\omega=e^{\pi i/4}$. Every diagonal
Clifford is, up to a global phase, expressed as a product of $Z$, $S$, and $CZ$ gates, so $C$
may be taken in that form. Physical $T$ gates contribute $\omega^{\lvert \vec f\rvert}$, while $C$
contributes $\omega^{2h(\vec f)}$. Equating the combined phase
$\omega^{\lvert \vec f\rvert+2h(\vec f)}$ to $\omega^{4x_1x_2x_3}$ gives
\begin{equation*}
\lvert \vec f\rvert+2h(\vec f)\equiv4x_1x_2x_3\pmod8,
\qquad
h(\vec f)=\sum_j\lambda_j f_j+\sum_{j<j'}\lambda_{jj'}f_jf_{j'},
\end{equation*}
where $\vec f=K^{\!\top}\vec x+S^{\!\top}\vec y$,
$\lambda_j\in\{0,1,2,3\}$, and $\lambda_{jj'}\in\{0,2\}$, as in
\eqref{eq:exp-FQT-CCZ} and~\eqref{eq:exp-correction}.
Logical CCZ up to a global phase would add a constant $c$ to the right-hand side.
The constant changes only the constant term of the expansion in $\vec z$, which the conditions below do
not use. At $\vec z=0$ the left-hand side is $0$, so $c\equiv0\pmod8$ in any case.

List the rows of $G$ as $g^{\,1},\ldots,g^{\,3+s}$, with $g^{\,i}=K_i$ for
$1\le i\le3$ and $g^{\,3+l}=S_l$ for $1\le l\le s$. For a nonempty index set
$T\subseteq\{1,\ldots,3+s\}$, use the monomial $m_T(\vec z)$ defined before Equation~\eqref{eq:exp-bridge}.

Apply Lemma~\ref{lem:exp-xorlift} to expand $f_j=\bigoplus_i g^i_jz_i$ and group its terms by $T$:
\begin{equation}\label{eq:exp-app-xor}
f_j=
\sum_{\emptyset\ne T\subseteq\{1,\ldots,3+s\}}
(-2)^{\lvert T\rvert-1}
\left(\prod_{i\in T}g^i_j\right)m_T(\vec z).
\end{equation}
To determine which coefficients the Clifford correction can change, use
Corollary~\ref{cor:exp-xorweight} to write the following integer identity
\[
2h(\vec f)=\sum_j2\lambda_jf_j+
\sum_{j<j'}\lambda_{jj'}\bigl(f_j+f_{j'}-(f_j\oplus f_{j'})\bigr).
\]
Each of $f_j$, $f_{j'}$, and $f_j\oplus f_{j'}$ is an XOR of selected variables $z_i$.
By~\eqref{eq:exp-app-xor}, the coefficient of $m_T$ in each XOR is
divisible by $2^{|T|-1}$.
Both $2\lambda_j$ and $\lambda_{jj'}$ are even, so every $m_T$ coefficient in
$2h(\vec f)$ is divisible by $2^{|T|}$ for every nonempty $T$.
In particular, degrees one, two, and three give
\begin{equation}\label{eq:exp-app-budget}
\text{coefficient of $m_T$ in $2h(\vec{f})$}\ \in\
\begin{cases}
\{0,2,4,6\},&\lvert T\rvert=1,\\
\{0,4\},&\lvert T\rvert=2,\\
\{0\},&\lvert T\rvert=3,
\end{cases}
\qquad(\bmod 8)
\end{equation}

Equation~\eqref{eq:exp-BH}, with $\vec f=G^{\!\top}\vec z$, gives the coefficients of
$z_i$, $z_iz_{i'}$, and $z_iz_{i'}z_{i''}$ as
$\lvert g^{\,i}\rvert$, $-2\lvert g^{\,i}\wedge g^{\,i'}\rvert$, and
$4\lvert g^{\,i}\wedge g^{\,i'}\wedge g^{\,i''}\rvert$, respectively.
Apply Equation~\eqref{eq:exp-topcoeff} to both sides of
Equation~\eqref{eq:exp-FQT-CCZ}.
The coefficient restrictions~\eqref{eq:exp-app-budget} and these coefficients then give
\begin{equation}\label{eq:exp-app-coefficients}
\begin{aligned}
\lvert g^{\,i}\rvert\bmod8&\in\{0,2,4,6\},\\
-2\lvert g^{\,i}\wedge g^{\,i'}\rvert\bmod8&\in\{0,4\},\\
4\lvert g^{\,i}\wedge g^{\,i'}\wedge g^{\,i''}\rvert
&\equiv
\begin{cases}
4\pmod8,&\{i,i',i''\}=\{1,2,3\},\\
0\pmod8,&\text{otherwise}.
\end{cases}
\end{aligned}
\end{equation}
Together, the three lines give the nine CH conditions in Table~\ref{tab:exp-conditions}.
\end{proof}

\subsection{Proof of Proposition~\ref{prop:exp-signed-overlaps}}\label{app:exp-signed-overlaps-proof}
\begin{proof}
Use $\vec z=(\vec x,\vec y)$, $\vec f=G^{\!\top}\vec z$, the rows
$g^{\,i}$, and the monomials $m_T(\vec z)$ from the preceding proof.
Multiplying~\eqref{eq:exp-app-xor} by $\Gamma_j$ and summing over $j$ gives the exact
integer expansion
\[
\mathrm{sw}_\Gamma(\vec f)=
\sum_{\emptyset\ne T}(-2)^{|T|-1}
\mathrm{sw}_\Gamma\!\left(\bigwedge_{i\in T}g^{\,i}\right)m_T(\vec z).
\]
By the coefficient identity~\eqref{eq:exp-topcoeff}, the equation
$\mathrm{sw}_\Gamma(\vec f)\equiv4x_1x_2x_3\pmod8$ for every $\vec z$
is equivalent to
\[
(-2)^{|T|-1}\mathrm{sw}_\Gamma\!\left(\bigwedge_{i\in T}g^{\,i}\right)
\equiv
\begin{cases}
4,&T=\{1,2,3\},\\
0,&T\ne\{1,2,3\}
\end{cases}
\pmod8
\qquad(T \ne \emptyset).
\]
Terms with $|T|\ge4$ vanish automatically. For $|T|=1,2,3$, the factor
$(-2)^{|T|-1}$ is $1$, $-2$, or $4$.
The three cases reduce to congruences mod $8$, $4$, and $2$, which are the three
lines of Equation~\eqref{eq:exp-signed-overlap-conditions}, proving both directions.
Finally, subtracting the expansion at $\vec y=0$ leaves only terms containing a
stabilizer row. The conditions that contain a stabilizer row make all such coefficients zero
mod 8, proving~\eqref{eq:exp-signed-stabilizer-cancellation}.

\end{proof}

\subsection{Proof of Proposition~\ref{prop:exp-blockcriterion}}\label{app:exp-blockcriterion-proof}

\begin{proof}[Proof of Proposition~\ref{prop:exp-blockcriterion}]

On each block, Corollary~\ref{cor:exp-xorweight} gives
$\lvert \gamma_b+c_b\rvert-\lvert \gamma_b\rvert=\lvert c_b\rvert-2\lvert \gamma_b\wedge c_b\rvert$. Summing the blocks gives
\begin{equation}\label{eq:exp-tpsplit}
\mathrm{sw}_\Gamma(c)=\lvert c\rvert-2(\gamma\cdot c)
=\sum_b\bigl(\lvert \gamma_b+c_b\rvert-\lvert \gamma_b\rvert\bigr).
\end{equation}
Substituting~\eqref{eq:exp-tpsplit} into~\eqref{eq:exp-CI} gives
Equation~\eqref{eq:exp-blockcriterion-mod16}. Under the conditions of
Proposition~\ref{prop:exp-signed-overlaps} involving a stabilizer row,
Equation~\eqref{eq:exp-signed-stabilizer-cancellation} makes the signed weight independent
of $\vec y$ modulo $8$. Setting $\vec y=0$ gives
Equation~\eqref{eq:exp-block-logical-criterion}, with eight logical inputs $\vec x\in\F_2^3$.

\end{proof}

\subsection{Proof of Proposition~\ref{prop:exp-q48-tpattern}}\label{app:exp-tpattern-proof}

{\renewcommand{\thelemma}{\ref{prop:exp-q48-tpattern}}\ExpQFortyEightPattern*}

\begin{proof}[Proof of Proposition~\ref{prop:exp-q48-tpattern}]

Equation~\eqref{eq:exp-Z4half} was used to find the pattern~\eqref{eq:exp-tpattern}.
The calculation below verifies the pattern independently of the search.

Proposition~\ref{prop:exp-signed-overlaps} reduces the check to $\vec y=0$ once its conditions
involving at least one stabilizer row hold.
The pattern~\eqref{eq:exp-tpattern} satisfies
\begin{equation*}
\begin{aligned}
\mathrm{sw}_\Gamma(S_l)&\equiv0\pmod8
&& (1\le l\le6),\\
\mathrm{sw}_\Gamma(K_i\wedge S_l)&\equiv0\pmod4
&& (1\le i\le3,\ 1\le l\le6),\\
\mathrm{sw}_\Gamma(S_l\wedge S_{l'})&\equiv0\pmod4
&& (1\le l<l'\le6).
\end{aligned}
\end{equation*}
The corresponding integer signed weights are recorded below. The row and column indices
are $l,l'\in\{1,\ldots,6\}$ and $i\in\{1,2,3\}$.
\begin{equation*}
\begin{array}{c|rrrrrr}
l&1&2&3&4&5&6\\ \hline
\mathrm{sw}_\Gamma(S_l)&0&0&8&0&0&0
\end{array}
\qquad
\begin{gathered}
\mathrm{sw}_\Gamma(K_i\wedge S_l)\\
\begin{array}{c|rrrrrr}
i\backslash l&1&2&3&4&5&6\\ \hline
1&0&0&4&-4&0&0\\
2&0&0&0&-4&0&-4\\
3&4&4&8&4&4&4
\end{array}
\end{gathered}
\end{equation*}
\begin{equation*}
\begin{gathered}
\mathrm{sw}_\Gamma(S_l\wedge S_{l'})\quad(l<l')\\
\begin{array}{c|rrrrrr}
l\backslash l'&1&2&3&4&5&6\\ \hline
1&-&0&4&0&0&0\\
2&&-&4&0&0&0\\
3&&&-&4&4&4\\
4&&&&-&0&0\\
5&&&&&-&-4\\
6&&&&&&-
\end{array}
\end{gathered}
\end{equation*}
The signed weights $\mathrm{sw}_\Gamma(S_l)$ of single stabilizer rows are divisible by
$8$, and the signed weights $\mathrm{sw}_\Gamma(K_i\wedge S_l)$ and
$\mathrm{sw}_\Gamma(S_l\wedge S_{l'})$ of pairs of rows are divisible by $4$.
Every triple of distinct rows containing a stabilizer row has even overlap by S3, KS2, or KS3.
Since $\Gamma_j\equiv1\pmod2$, its signed overlap is even as well.
Together with the displayed one-row and two-row equations, Equation~\eqref{eq:exp-signed-stabilizer-cancellation} implies
\[
\mathrm{sw}_\Gamma(K^{\!\top}\vec x+S^{\!\top}\vec y)
\equiv\mathrm{sw}_\Gamma(K^{\!\top}\vec x)\pmod8.
\]
Proposition~\ref{prop:exp-blockcriterion} now reduces Equation~\eqref{eq:exp-CI} to
$\mathrm{sw}_\Gamma(K^{\!\top}\vec x)\equiv4x_1x_2x_3\pmod8$ for the eight
logical inputs $\vec x\in\F_2^3$. Equation~\eqref{eq:exp-tpsplit} computes this signed
weight from the three block weights of $\gamma+K^{\!\top}\vec x$.

Write $b=1,2,3$ for the blocks with coset labels
$w^{(1)}=(1,0)$, $w^{(2)}=(0,1)$, and $w^{(3)}=(1,1)$. A column in block $b$ has
stabilizer coordinates $(u,w^{(b)})$ and logical labels
$\kappa(u,w^{(b)})$ from Table~\ref{tab:exp-labelmap}.
For $c=K^{\!\top}\vec x$, let $f_b$ be the restriction of $\gamma+c$ to block~$b$, that is,
$f_b=\gamma_b+c_b$ in the notation of Proposition~\ref{prop:exp-blockcriterion}:
\begin{equation}\label{eq:exp-tpblockrestriction}
f_b(u)=\gamma(u,w^{(b)})+
\sum_{i=1}^3x_i\kappa_i(u,w^{(b)}).
\end{equation}
For each logical input $\vec x$, the signed weight to be checked is
$\mathrm{sw}_\Gamma(K^{\!\top}\vec x)=\sum_b\lvert f_b\rvert-\lvert\gamma\rvert$.

Counting the $u\in\F_2^4$ with $f_b(u)=1$ for each of the eight logical
inputs $\vec x$ gives the block weights needed in~\eqref{eq:exp-tpsplit}:
\phantomsection\label{table:exp-block-weights}
\[
\begin{array}{c|rrrrrrrr}
\vec x&(0,0,0)&(0,0,1)&(0,1,0)&(0,1,1)&(1,0,0)&(1,0,1)&(1,1,0)&(1,1,1)\\ \hline
\lvert f_1\rvert&6&10&10&10&6&6&6&6\\
\lvert f_2\rvert&6&10&6&10&6&10&6&6\\
\lvert f_3\rvert&10&10&6&10&10&6&10&6\\ \hline
\sum_b\lvert f_b\rvert&22&30&22&30&22&22&22&18
\end{array}
\]
For example, at $\vec x=0$ the first block has
$f_1(u)=(u_1+u_2)(1+u_3)+u_2u_4$.
The cases $u_2=0$ and $u_2=1$ contribute two and four ones, respectively, giving
$\lvert f_1\rvert=6$.

At $\vec x=0$, $f_b=\gamma_b$, so the first column of the
\hyperref[table:exp-block-weights]{block-weight table} gives
$\lvert\gamma\rvert=6+6+10=22$.

Subtracting $\lvert\gamma\rvert=22$ from the sum row of the
\hyperref[table:exp-block-weights]{block-weight table} gives
$0$ or $8$ for every input except $(1,1,1)$, where it gives $-4$.
Equation~\eqref{eq:exp-tpsplit} therefore yields
\begin{equation}\label{eq:exp-tpsigncount}
\mathrm{sw}_\Gamma(c)=\sum_b\lvert f_b\rvert-22
\equiv4x_1x_2x_3\pmod8.
\end{equation}
The specified $T/T^\dagger$ pattern thus implements logical CCZ.

\end{proof}

\section{Support reductions and length lower bounds}\label{app:exp-floor}

\subsection{Reduction to projective stabilizer supports}\label{app:exp-projred}

\begin{lemma}[Reduction to projective $\mathcal C_2$]\label{lem:exp-projred}
Let $G=[K;S]$ satisfy the nine CH conditions with $k=3$ and $d_Z\ge3$.
If $\mathcal C_2$ is non-projective, then $S$ has a zero column $j$ or two equal columns $j\ne j'$.
In the zero-column case, form $G'$ by deleting column $j$.
In the repeated-column case, form $G'$ by deleting both equal columns $j$ and $j'$.
In either case, $G'$ satisfies:
\begin{enumerate}
\renewcommand{\labelenumi}{(\roman{enumi})}
\item $G'$ satisfies the nine CH conditions,
\item $k'=k$,
\item $d_Z(G')\ge d_Z(G)$.
\end{enumerate}
Iterating until $\mathcal C_2$ is projective terminates and preserves all three properties.
\end{lemma}

\begin{proof}
Non-projectivity of $\mathcal C_2$ gives a zero column of $S$ or a repeated pair. Let
$e_j\in\F_2^n$ denote the $j$th coordinate vector, and define $\sigma\in\F_2^n$ by
\begin{equation*}
\sigma=
\begin{cases}
e_j,&\text{if column $j$ of $S$ is zero},\\
e_j+e_{j'},&\text{if columns $j\ne j'$ of $S$ are equal}.
\end{cases}
\end{equation*}
Lemma~\ref{lem:exp-fiber}(b) and $d_Z\ge3$ imply that the corresponding full column of $G$
is zero in the first case and that the two corresponding full columns are equal in the second.
Therefore
\begin{equation}\label{eq:exp-proj-defect}
G\sigma=0,\qquad \lvert \sigma\rvert\le2.
\end{equation}
Puncture the coordinates in $\supp(\sigma)$ from both $S$ and $G=[K;S]$, and call the
results $S'$ and $G'$.

By Lemma~\ref{lem:exp-generator-operations}(c), $G'$ satisfies the nine CH conditions.

The logical dimension is preserved if the row-space dimensions of $G$ and $S$ fall by the same
amount. Equation~\eqref{eq:exp-proj-defect} says that every codeword in $\mathcal C_1$, and hence
every codeword in $\mathcal C_2$, has even parity on $\supp(\sigma)$. The only possible nonzero
codeword in either code supported inside $\supp(\sigma)$ is therefore $\sigma^{\!\top}$.
Equation~\eqref{eq:exp-puncture-rank} gives
\begin{equation}\label{eq:exp-proj-rank}
\begin{aligned}
\operatorname{rank}G-\operatorname{rank}G'
&=\begin{cases}1,&\sigma^{\!\top}\in\mathcal C_1,\\0,&\sigma^{\!\top}\notin\mathcal C_1,
\end{cases}
&\qquad
\operatorname{rank}S-\operatorname{rank}S'
&=\begin{cases}1,&\sigma^{\!\top}\in\mathcal C_2,\\0,&\sigma^{\!\top}\notin\mathcal C_2.
\end{cases}
\end{aligned}
\end{equation}

In the zero-column case, every codeword of $\mathcal C_1$ has zero in coordinate $j$, so
$\sigma^{\!\top}=e_j^{\!\top}\notin\mathcal C_1$ and neither rank falls.

For the repeated-column case, $\mathcal C_2\subseteq\mathcal C_1$ means the two lines of
\eqref{eq:exp-proj-rank} can differ only if
$\sigma^{\!\top}\in\mathcal C_1\setminus\mathcal C_2$. Assume
$\sigma^{\!\top}\in\mathcal C_1$ and write
\begin{equation}\label{eq:exp-proj-decomp}
\sigma^{\!\top}=\sum_{i=1}^3x_iK_i+\sum_{l=1}^s y_lS_l,
\qquad x\in\F_2^3,\quad y\in\F_2^s.
\end{equation}
For $\{i,i',i''\}=\{1,2,3\}$, the two deleted full columns are equal, so
$|\sigma^{\!\top}\wedge K_{i'}\wedge K_{i''}|$ is even.
Corollary~\ref{cor:exp-ch-codes}(a), applied to
\eqref{eq:exp-proj-decomp}, therefore gives $x_i=0$.
Thus $x_1=x_2=x_3=0$, and~\eqref{eq:exp-proj-decomp} puts
$\sigma^{\!\top}$ in $\mathcal C_2$.
So $\sigma^{\!\top}\in\mathcal C_1$ forces $\sigma^{\!\top}\in\mathcal C_2$. Together
with $\mathcal C_2\subseteq\mathcal C_1$, the vector $\sigma^{\!\top}$ lies in both codes or in
neither, so the two cases of~\eqref{eq:exp-proj-rank} always agree. The rank drops
in~\eqref{eq:exp-proj-rank} are equal, so $k'=\operatorname{rank}G'-\operatorname{rank}S'=k$.

Lemma~\ref{lem:exp-generator-operations}(c) gives $d_Z(G')\ge d_Z(G)$.

Each puncture reduces the number of columns by one or two. The process therefore stops after
finitely many steps, when the columns of $S$ are distinct and nonzero, which is precisely when
$\mathcal C_2$ is projective.
\end{proof}

\subsection{Admissible supports and affine classification}\label{app:exp-lengths-flat}
\subsubsection{Support polynomials and block lengths}\label{app:exp-support-lengths}

\begin{proof}[Proof of Theorem~\ref{thm:exp-spectrum}]
In this proof, we translate conditions S1--S3 to inner products associated with the
column indicator $\chi_S$, which allows us to compute the block length $n$ directly.

Set $T_K=\emptyset$ in Equation~\eqref{eq:exp-clause}. Then $m_{T_K}(\kappa(\beta))=1$ for every
$\beta$, and conditions S1--S3 become
\begin{equation*}
\bigoplus_{\beta\ne0}\chi_S(\beta)m_{T_S}(\beta)
=\langle\chi_S,m_{T_S}\rangle=0,
\qquad 1\le\lvert T_S\rvert\le3.
\end{equation*}
The cases $\lvert T_S\rvert=1,2,3$ are S1, S2, and S3, respectively.

To characterize the indicators orthogonal to these monomials, let
\begin{equation*}
W=\operatorname{span}\{m_T:1\le |T|\le3\},
\qquad
\RM(3,s)=W\cup(\mathbf1+W).
\end{equation*}
Conditions S1--S3 imply $\chi_S\in W^\perp$. Reed--Muller
duality~\eqref{eq:exp-rmdual} gives
\begin{equation*}
\RM(3,s)^\perp=\RM(s-4,s).
\end{equation*}
Moreover,
Equation~\eqref{eq:exp-deltainner} gives $\delta_0\in W^\perp$ and
$\langle\delta_0,1\rangle=1$.
Since $W\subseteq\RM(3,s)$, the code $\RM(s-4,s)=\RM(3,s)^\perp$ lies in $W^\perp$.
Every element of $\RM(s-4,s)$ is orthogonal to $1\in\RM(3,s)$, so
$\langle\delta_0,1\rangle=1$ gives $\delta_0\notin\RM(s-4,s)$. Finally,
\begin{equation*}
\dim W^\perp
=2^s-\bigl(\dim\RM(3,s)-1\bigr)
=\dim\RM(s-4,s)+1.
\end{equation*}
The inclusions and the dimension equality yield
\begin{equation*}
W^\perp=\RM(s-4,s)\cup\bigl(\delta_0+\RM(s-4,s)\bigr),
\qquad
\chi_S\in\RM(s-4,s)\cup\bigl(\delta_0+\RM(s-4,s)\bigr).
\end{equation*}

\medskip\noindent
Since $S$ has no zero column, its column indicator satisfies $\chi_S(0)=0$.
Writing $\chi_S=c+\varepsilon\delta_0$ with $c\in\RM(s-4,s)$ forces
$\varepsilon=c(0)$, and hence
\begin{equation}\label{eq:exp-puncture}
n=
\begin{cases}
\lvert c\rvert,&c(0)=0,\\
\lvert c\rvert-1,&c(0)=1.
\end{cases}
\end{equation}
Conversely, every $c\in\RM(s-4,s)$ gives
$\chi_S=c+c(0)\delta_0$, which vanishes at the origin and is orthogonal to every
nonconstant monomial of degree at most three.
Thus S1--S3 hold, and the two cases above give exactly the lengths allowed by
these conditions.

The columns of $S$ are exactly the nonzero $\beta$ with $c(\beta)=1$, so
\[
\operatorname{rank}S
=\dim\operatorname{span}\{\beta\in\F_2^s\setminus\{0\}:c(\beta)=1\}.
\]
Requiring $\operatorname{rank}S=s$ gives the last condition in
Theorem~\ref{thm:exp-spectrum}.

\end{proof}

\subsubsection{Affine equivalence and puncturing}\label{app:exp-affine-equivalence}

Theorem~\ref{thm:exp-spectrum} describes each candidate support $\mathcal P$ by a polynomial
$c\in\RM(s-4,s)$. For a map $\kappa\colon\supp(c)\to\F_2^3$, let $G$ be the matrix
with columns $(\kappa(\beta),\beta)$ for $\beta\in\supp(c)$.
\begin{lemma}[Affine changes of the support]\label{lem:exp-affine-support}
Let $L$ be an invertible affine map of $\F_2^s$, and let $G_L$ be the matrix with
columns $(\kappa(\beta),L(\beta))$ for $\beta\in\supp(c)$. Then $G$ satisfies the nine CH
conditions if and only if $G_L$ does.
\end{lemma}
Each column keeps its logical label $\kappa(\beta)$, and only its syndrome moves from
$\beta$ to $L(\beta)$.

\begin{proof}
By Lemma~\ref{lem:exp-affine-invariance}, the syndromes $L(\beta)$ of $G_L$ form
$\supp(c\circ L^{-1})$, and $c\circ L^{-1}\in\RM(s-4,s)$, so $G_L$ has the same form as $G$.
Write $L(\beta)=A(\beta+a)$ with $A\in\operatorname{GL}(s,2)$ and $a\in\F_2^s$, a translation by
$a$ followed by the linear map $A$. Proving the forward implication for each of the two maps
suffices, because $L^{-1}$ is again an invertible affine map and gives the converse.

An invertible linear change of the syndromes is a row operation on $S$. Replacing every
syndrome $\beta$ by $A\beta$ replaces $S$ by $AS$. Since $A$ is invertible,
Lemma~\ref{lem:exp-generator-operations}(a) shows that $[K;AS]$ satisfies CH exactly when $[K;S]$
does, with the same $n$, $s$, $k$, and $d_Z$.

A translation of the syndromes is a row operation on $[S;\mathbf1^{\mathsf T}]$.
Replacing every syndrome $\beta$ by $\beta+a$ adds $a_j$ to every entry of row $j$ of $S$, so $S$
becomes $S+a\mathbf1^{\mathsf T}$. Each new row is a combination of the rows of
$[S;\mathbf1^{\mathsf T}]$, and that matrix satisfies CH by Corollary~\ref{cor:exp-ch-codes}(b)
once the full support has even size. The size is even because the coefficient of
$\beta_1\cdots\beta_s$ in $c$ vanishes ($\deg c\le s-4$), so
Lemma~\ref{lem:exp-coefficient-parity} gives
\begin{equation}\label{eq:exp-affine-constant}
|c|\equiv0\pmod2.
\end{equation}
Lemma~\ref{lem:exp-generator-operations}(a) then shows that $[K;S+a\mathbf1^{\mathsf T}]$
satisfies CH.
\end{proof}

Moving a supported point to the origin creates a zero column, which is deleted. Every
invertible affine change $L(\beta)=A\beta+b$ can be written as $L(\beta)=A(\beta+a)$ with
$a=A^{-1}b$, and $a$ is the unique point that $L$ sends to the origin. After the linear part $A$
is removed as above, the candidate support is
\begin{equation}\label{eq:exp-floor-origin}
\mathcal P_a=(\supp(c)+a)\setminus\{0\}.
\end{equation}
The point $a$ is an \emph{origin choice}. If $c(a)=0$, then $\mathcal P_a$ has length
$|c|$. If $c(a)=1$, the column indexed by $a$ has syndrome $0$ after the translation. The
logical label at $a$ must be $\kappa(a)=0$ when $d_Z\ge3$ (Lemma~\ref{lem:exp-fiber}), and puncturing the
code on this zero column preserves CH and $k$ (Lemma~\ref{lem:exp-projred}), so
$\mathcal P_a$ has length $|c|-1$.
The proof of Theorem~\ref{thm:exp-floor38} does not test every origin choice. For each
polynomial $c$, the proof tests only the support $\mathcal P_{a_0}$ of length $|c|$, for one point
$a_0$ with $c(a_0)=0$. Corollary~\ref{cor:exp-one-support} shows that this single test also
excludes $\mathcal P_a$ for every other $a$, including the supports of length $|c|-1$ with
$c(a)=1$.

\begin{corollary}\label{cor:exp-one-support}
Let $c\in\RM(s-4,s)$, and let $a_0\in\F_2^s$ satisfy $c(a_0)=0$. If no label map
satisfies the CH conditions on $\mathcal P_{a_0}$, then for every $a\in\F_2^s$ no label map on
$\mathcal P_a$ satisfies the CH conditions with $d_Z\ge3$.
\end{corollary}

\begin{proof}
Since $c(a_0)=0$, the support $\mathcal P_{a_0}$ equals $\supp(c)+a_0$.
Given a label map on some $\mathcal P_a$
that satisfies the CH conditions with $d_Z\ge3$, reinsert the zero column if $c(a)=1$
(Lemma~\ref{lem:exp-generator-operations}(c)) and translate back by $a$. The result is a label
map on $\supp(c)$ that satisfies the CH conditions, and translating by $a_0$ gives one on
$\mathcal P_{a_0}$ by Lemma~\ref{lem:exp-affine-support}. The resulting label map on $\mathcal P_{a_0}$ contradicts the hypothesis.
\end{proof}

Appendix~\ref{app:exp-flat-supports} below focuses on the supports $\mathcal P$ that are affine flats, which
carry no label map satisfying the CH conditions, and bounds the rank of supports that lie inside
a flat.

\subsubsection{Affine-flat supports}\label{app:exp-flat-supports}

\begin{proof}[Proof of Lemma~\ref{lem:exp-flatobstruction}]
On a flat of dimension $m\le4$, conditions K1, KS1, and KS3 force each label
$\kappa_i$ to have degree at most $m-3$ in the coordinates of the flat, and $\kappa_i=0$ when
$m\le2$. The product $\kappa_1\kappa_2\kappa_3$ then has degree below $m$ and even weight, so K3
fails before any search for the labels. The proof makes this degree count precise.

Let $\mathcal P=v+U$, where $U\subseteq\F_2^s$ has dimension $m\le4$.
Choose affine coordinates on $\F_2^s$ in which $\mathcal P$ is the set of points
$(z,0)$ with $z\in\F_2^m$, so that every function on $\mathcal P$ becomes a function of
$z\in\F_2^m$, and write $\kappa_i(z)$ for the label at $(z,0)$. Degrees are unchanged
under this affine change by Lemma~\ref{lem:exp-affine-invariance}, so we have
\begin{equation}\label{eq:exp-flat-restriction}
\RM(2,s)|_{\mathcal P}
=\{q(z,0):q\in\RM(2,s)\}
=\RM(2,m).
\end{equation}
Equation~\eqref{eq:exp-VA},
Reed--Muller duality~\eqref{eq:exp-rmdual}, and
Equation~\eqref{eq:exp-flat-restriction} give
\begin{equation}\label{eq:exp-flat-label-degree}
\kappa_i\in\RM(2,m)^\perp=
\begin{cases}
\{0\},&m\le2,\\
\RM(m-3,m),&3\le m\le4.
\end{cases}
\end{equation}
For $m\le2$, every $\kappa_i$ is zero, contradicting K3. For $3\le m\le4$, the degree
of a product is at most the sum of the degrees, $\deg(fg)\le\deg f+\deg g$, and $3(m-3)\le m-1$
holds because $m\le4$. Hence
\begin{equation}\label{eq:exp-flat-product-degree}
\deg(\kappa_1\kappa_2\kappa_3)\le3(m-3)\le m-1.
\end{equation}
The coefficient of $z_1\cdots z_m$ in $\kappa_1\kappa_2\kappa_3$ therefore vanishes, and
Equation~\eqref{eq:exp-topcoeff-parity} of Lemma~\ref{lem:exp-coefficient-parity} gives
\begin{equation}\label{eq:exp-flat-product-even}
|\kappa_1\kappa_2\kappa_3|\equiv0\pmod2.
\end{equation}
Equation~\eqref{eq:exp-bridge} with $T=\{1,2,3\}$ and $T_S=\emptyset$ therefore gives
\begin{align}\label{eq:exp-flat-k3-contradiction}
|K_1\wedge K_2\wedge K_3|
&\equiv\sum_{\beta\in\mathcal P}\kappa_1(\beta)\kappa_2(\beta)\kappa_3(\beta)\notag\\
&\equiv|\kappa_1\kappa_2\kappa_3|\equiv0\pmod2,
\end{align}
which contradicts K3.

The lemma also covers a flat with the origin deleted. For such a support, inserting
the zero column at the origin gives, by Lemma~\ref{lem:exp-generator-operations}(c), a label
map on the full flat that satisfies the CH conditions, and the argument above excludes it.

For $n\in\{15,16\}$, Theorem~\ref{thm:exp-spectrum} gives $n=|c|$ or $n=|c|-1$ for some
$c\in\RM(s-4,s)$, and $|c|$ is even by~\eqref{eq:exp-affine-constant}. Hence
$|c|=16$, the minimum weight of $\RM(s-4,s)$. The minimum-weight codewords are the indicators of
$4$-dimensional affine flats (Appendix~\ref{app:exp-facts-boolean}), so the support is such a flat, and this lemma excludes it.
\end{proof}

\phantomsection\label{case:exp-hyperplane-origin}
A support inside an affine flat has actual rank at most the dimension of the flat plus
one. Steps~2, 3, and~5 of the proof of Theorem~\ref{thm:exp-floor38} use this bound to discard
supports that lie in a flat.

\subsubsection{Small-rank support exclusion}\label{app:exp-s5}

\begin{proof}[Proof of Proposition~\ref{prop:exp-lb}]
At stabilizer rank at most $5$, Theorem~\ref{thm:exp-spectrum} leaves so few
supports that all but one are flats, and a direct search settles the remaining case.

By Lemma~\ref{lem:exp-projred} and Theorem~\ref{thm:exp-spectrum}, a code with
stabilizer rank $s\le5$ has support $\mathcal P=\supp(c)\setminus\{0\}$ for some
$c\in\RM(s-4,s)$.

For $s<4$, the convention $\RM(s-4,s)=\{0\}$ in~\eqref{eq:exp-rmdual}
makes the support empty. For $s=4$, the only nonzero codeword in $\RM(0,4)$
is the constant function $\mathbf1$. The function's projective support is
$\F_2^4\setminus\{0\}$, which Lemma~\ref{lem:exp-flatobstruction} excludes.

For $s=5$, every codeword of $\RM(1,5)$ has the form $c(\beta)=a\cdot\beta+b$, and
Table~\ref{tab:exp-s5-supports} lists the resulting supports. The $62$ supports with $a\ne0$ are
affine hyperplanes of $\F_2^5$, which are affine flats of dimension $4$, or such hyperplanes with
the origin deleted. Lemma~\ref{lem:exp-flatobstruction} excludes them, and only
$\mathcal P=\F_2^5\setminus\{0\}$ remains.
\begin{table}[htbp]
\centering
{\small
\begin{tabular}{@{}llrrl@{}}
\hline
Coefficients & Support & Length & Count & Exclusion \\
\hline
$a\ne0,\ b=0$ & Affine hyperplane & $16$ & $31$ & Lemma~\ref{lem:exp-flatobstruction} \\
$a\ne0,\ b=1$ & Punctured hyperplane & $15$ & $31$ & Lemma~\ref{lem:exp-flatobstruction} \\
$a=0,\ b=1$ & $\F_2^5\setminus\{0\}$ & $31$ & $1$ & Algorithm~\ref{alg:exp-search} \\
\hline
\end{tabular}}
\caption{Supports from the $63$ nonzero codewords of $\RM(1,5)$.}
\label{tab:exp-s5-supports}
\end{table}

To run Algorithm~\ref{alg:exp-search} on this support we need $d=\dim V_{\mathcal P}$,
the number of unknowns in~\eqref{eq:exp-third-label-coefficients}. By~\eqref{eq:exp-VA},
$V_{\mathcal P}$ is the orthogonal complement of $\operatorname{ev}_{\mathcal P}(\RM(2,5))$ in
$\F_2^{31}$. The map $\operatorname{ev}_{\mathcal P}$ is injective on $\RM(2,5)$. A polynomial
of degree at most two that vanishes on every nonzero point of $\F_2^5$ is $0$ or $\delta_0$, and
$\delta_0$ has degree five by Lemma~\ref{lem:exp-origin-indicator}, so it is $0$. Hence
$\operatorname{ev}_{\mathcal P}(\RM(2,5))$ has dimension $\dim\RM(2,5)=16$, and
\begin{equation}\label{eq:exp-s5-label-dimension}
\dim V_{\mathcal P}=15.
\end{equation}
With $d=15$, Algorithm~\ref{alg:exp-search} exhausts all $49{,}741{,}825$ ordered pairs
$(\kappa_1,\kappa_2)$ satisfying~\eqref{eq:exp-admissible-pair}. For each pair,
Equation~\eqref{eq:exp-third-label-coefficients} has $2+2s+1=13$
equations at $s=5$ in the $15$ basis coefficients of $\kappa_3$. The equation admits no
solution for any pair, so no pair extends to a CH label map.
\end{proof}

\subsection{Stabilizer-rank reduction and hyperplane supports}\label{app:exp-rank-reductions}

Let $G=[K;S]$ satisfy the nine CH conditions with $k=3$ and $d_Z\ge3$ and $\mathcal C_2$ projective. Express $G$ in
terms of its set $\mathcal P$ of syndrome columns:
\begin{equation}\label{eq:exp-rank-reduction-columns}
G=\left(\binom{\kappa(\beta)}{\beta}\right)_{\beta\in\mathcal P}.
\end{equation}

To delete one independent stabilizer row, construct a projection
$\pi:\F_2^s\to\F_2^{s-1}$ with a one-dimensional kernel,
generated by a nonzero vector $v\in\F_2^s$:
\begin{equation}\label{eq:exp-levelred-kernel}
\ker\pi=\langle v\rangle=\{0,v\}.
\end{equation}
Choose $m$ with $v_m=1$, and define $\pi$ by
\begin{equation*}
\pi(\beta)=\bigl(\beta_l+v_l\beta_m\bigr)_{l\ne m}.
\end{equation*}
If $\pi(\beta)=0$, then $\beta_l=v_l\beta_m$ for $l\ne m$, hence
$\beta=\beta_mv$. For $\delta\in\F_2^3$, apply the row operations
\begin{equation*}
K_i\mapsto K_i+\delta_iS_m\quad(1\le i\le3),\qquad
S_l\mapsto S_l+v_lS_m\quad(l\ne m),
\end{equation*}
and delete $S_m$. The full-column operation is
\begin{equation}\label{eq:exp-levelred-full-map}
\binom{\alpha}{\beta}\longmapsto
\binom{\alpha+\beta_m\delta}{\pi(\beta)},\qquad
\alpha\in\F_2^3,\quad\beta\in\F_2^s.
\end{equation}
Equation~\eqref{eq:exp-levelred-kernel} and $v_m=1$ give
\begin{equation}\label{eq:exp-levelred-full-kernel}
\binom{\alpha+\beta_m\delta}{\pi(\beta)}=0
\quad\Longleftrightarrow\quad
\binom{\alpha}{\beta}\in\left\{0,\binom{\delta}{v}\right\}.
\end{equation}

Whether the output keeps $d_Z\ge3$ depends on the labels of the syndromes that differ
by $v$. For nonzero $v\in\F_2^s$, collect their label differences:
\begin{equation}\label{eq:exp-labeldiff}
\Delta_v=\bigl\{\,\kappa(\beta+v)+\kappa(\beta)
 : \{\beta,\beta+v\}\subseteq\mathcal P\,\bigr\}\subseteq\F_2^3.
\end{equation}
Several pairs may give the same value, and $\Delta_v=\emptyset$ when no supported pair
differs by $v$.

\begin{lemma}\label{lem:exp-projection}
Let $G=[K;S]$ satisfy the nine CH conditions with $k=3$, $d_Z\ge3$, and
$\mathcal C_2$ projective, at stabilizer rank $s$. Let $\pi$, $v$, $m$, and $\delta$ be as above,
and suppose that
\begin{enumerate}
\item[(i)] $\Delta_v\subseteq\{\delta\}$, and
\item[(ii)] $\kappa(v)=\delta$ if $v\in\mathcal P$.
\end{enumerate}
Let
\begin{equation*}
G'=[K';S']=
\left(\binom{\kappa(\beta)+\beta_m\delta}{\pi(\beta)}\right)_{\beta\in\mathcal P}.
\end{equation*}
Then $G'$ satisfies the nine CH conditions with $k'=3$, $d_Z'\ge3$, and
$\operatorname{rank}S'=s-1$. A column of $G'$ is zero, or two distinct columns of $G'$ are
equal, in the following cases:
\begin{align}
\binom{\kappa(\beta)+\beta_m\delta}{\pi(\beta)}=0
&\quad\Longleftrightarrow\quad\beta=v,\label{eq:exp-projection-zero}\\
\binom{\kappa(\beta)+\beta_m\delta}{\pi(\beta)}=
\binom{\kappa(\beta')+\beta'_m\delta}{\pi(\beta')}
&\quad\Longleftrightarrow\quad\beta'=\beta+v,\label{eq:exp-projection-equal}
\end{align}
where $\beta,\beta'\in\mathcal P$ and $\beta\ne\beta'$ in the second line.
Applying Lemma~\ref{lem:exp-projred} to $G'$ gives a generator matrix of stabilizer rank at most
$s-1$ and length at most $n$ with $k=3$, all nine CH conditions, $d_Z\ge3$, and
$\mathcal C_2$ projective.
\end{lemma}

\begin{proof}
Each new logical row is $K_i+\delta_iS_m$, and each retained stabilizer row is
$S_l+v_lS_m$. Lemma~\ref{lem:exp-generator-operations}(a) therefore gives the CH conditions
for $G'$.

Since $k=3$ and the $s$ rows of $S$ are independent, the $s+3$ rows of $G$ are
independent. The row additions preceding~\eqref{eq:exp-levelred-full-map} are invertible.
Deleting $S_m$ therefore leaves $s+2$ independent rows of $G'$ and $s-1$ independent
stabilizer rows, so
\begin{equation*}
\operatorname{rank}G'=s+2,\qquad
\operatorname{rank}S'=s-1,\qquad k'=(s+2)-(s-1)=3.
\end{equation*}

For distinct $\beta,\beta'\in\mathcal P$, Equation~\eqref{eq:exp-levelred-kernel}
gives $\pi(\beta)=\pi(\beta')$ exactly when $\beta'=\beta+v$. Condition~(i) then gives
$\kappa(\beta')+\kappa(\beta)=\delta$. The two input columns differ by $(\delta,v)$, so their
images under~\eqref{eq:exp-levelred-full-map} agree by~\eqref{eq:exp-levelred-full-kernel}.
Conversely, equal columns have equal syndromes, so $\beta'=\beta+v$, and
\eqref{eq:exp-projection-equal} follows.

The syndrome $\pi(\beta)$ of a column of $G'$ is zero exactly when
$\beta\in\ker\pi=\{0,v\}$. Since $0\notin\mathcal P$, this happens exactly when $\beta=v$ and
$v\in\mathcal P$. In that case condition~(ii) gives $\kappa(v)=\delta$, and $v_m=1$ makes the new
label $\kappa(v)+v_m\delta=0$, so the column is zero by~\eqref{eq:exp-levelred-full-kernel}.
Conversely, a zero column has zero syndrome, so $\beta=v$, and
\eqref{eq:exp-projection-zero} follows.
By~\eqref{eq:exp-projection-equal}, two columns of $G'$ with the same nonzero syndrome
are equal, so they carry the same logical label. By~\eqref{eq:exp-projection-zero}, a column of
$G'$ with zero syndrome has logical label zero. Lemma~\ref{lem:exp-fiber}(b) requires exactly
these two properties, so $d_Z'\ge3$.

Lemma~\ref{lem:exp-projred} removes zero columns and pairs of equal columns
while retaining CH, $k'=3$, and $d_Z'\ge3$.
\end{proof}

The hypotheses of Lemma~\ref{lem:exp-projection} follow from counting in
Appendix~\ref{app:exp-levelred-proof}, or from hyperplane containment in
Appendix~\ref{app:exp-hyperplane-reductions}, which we discuss below.

\subsubsection{General rank reduction}\label{app:exp-levelred-proof}

\begin{proof}[Proof of Lemma~\ref{lem:exp-levelred}]

Each unordered pair of distinct columns $\{\beta,\beta+v\}\subseteq\mathcal P$ determines the
unique nonzero difference $v$. A vector $v$ with $\lvert\Delta_v\rvert\ge2$ accounts for at
least two pairs, and no pair is counted for two different $v$. Hence
\begin{equation}\label{eq:exp-levelred-count}
\#\{\,v:\lvert\Delta_v\rvert\ge2\,\}
\le \Bigl\lfloor\tbinom n2\big/2\Bigr\rfloor.
\end{equation}
Exactly $n$ vectors $v$ lie in $\mathcal P$.
There are $2^s-1$ nonzero vectors $v$, so
\begin{equation*}
\#\{\,v:v\notin\mathcal P\ \text{and}\ \lvert\Delta_v\rvert\le1\,\}
\ge 2^s-1-n-\Bigl\lfloor\tbinom n2\big/2\Bigr\rfloor.
\end{equation*}

\phantomsection\label{proof:exp-levelred-construction}
The hypothesis of Lemma~\ref{lem:exp-levelred} makes this lower bound positive.
Fix a nonzero
$v\notin\mathcal P$ with $\lvert\Delta_v\rvert\le1$. Let $\delta$ be the unique element of
$\Delta_v$ when the set is nonempty, and set $\delta=0$ otherwise.
Both hypotheses of Lemma~\ref{lem:exp-projection} hold: $\Delta_v\subseteq\{\delta\}$
and $v\notin\mathcal P$. Applying the lemma proves the claimed reduction.

\end{proof}

\subsubsection{Hyperplane reductions}\label{app:exp-hyperplane-reductions}
We can also use the membership of a support $\mathcal P$ within a hyperplane $H$ as a
hypothesis for Lemma~\ref{lem:exp-projection}, independent of the counting inequality.
 Table~\ref{tab:exp-hyperplane-reductions} records the two cases
for the rank-$s$ generator~\eqref{eq:exp-rank-reduction-columns}.

\begin{table}[htbp]
\centering
\begin{tabular}{@{}p{0.20\linewidth}p{0.53\linewidth}p{0.21\linewidth}@{}}
\toprule
Support case & Operation & Length and rank\\
\midrule
$\mathcal P\subsetneq H$
& Lemma~\ref{lem:exp-projection} with $v\in H\setminus\mathcal P$, $\delta=0$
& $n'=n$,\newline $s'=s-1$\\[3pt]
$\mathcal P=H$
& Lemma~\ref{lem:exp-projection} with $v\in H$, $\delta=\kappa(v)$, puncture the
zero column& $n'=n-1$,\newline $s'=s-1$\\
\bottomrule
\end{tabular}
\caption{Hyperplane reductions for $\mathcal P\subseteq H\subset\F_2^s$ with
$\operatorname{span}\mathcal P=\F_2^s$ and $0\notin H$. Both cases retain $k=3$, the nine CH
conditions, $d_Z\ge3$, and projectivity.}
\label{tab:exp-hyperplane-reductions}
\end{table}

\begin{corollary}[Hyperplane reduction]\label{cor:exp-hyperplane-reduction}
Let $G=[K;S]$ satisfy the nine CH conditions with $k=3$, $d_Z\ge3$, and
$\mathcal C_2$ projective, at stabilizer rank $s$. If its syndrome set $\mathcal P$ lies in an
affine hyperplane $H$ of $\F_2^s$, then $G$ reduces to a generator matrix of stabilizer rank at most
$s-1$ and length at most $n$, with $k=3$, all nine CH conditions, $d_Z\ge3$, and $\mathcal C_2$ projective.
The length stays $n$ when $\mathcal P\subsetneq H$ and becomes $n-1$ when $\mathcal P=H$.
\end{corollary}

\begin{proof}
Since $\operatorname{span}\mathcal P=\F_2^s$, write
$H=\{\beta:a\cdot\beta=1\}$ with $a\ne0$. For $v\in H$ and $\beta\in H$,
$a\cdot(\beta+v)=0$, so $\beta+v\notin H$ and $\Delta_v=\emptyset$.
Thus Lemma~\ref{lem:exp-projection}(i) holds for any $\delta\in\F_2^3$, the coefficient vector
in the logical-row additions. The projection $\pi$ of~\eqref{eq:exp-levelred-kernel} is
injective on $H$: equal images of distinct points would require their difference to be $v$.

\phantomsection\label{case:exp-hyperplane-proper}
If $\mathcal P\subsetneq H$, choose $v\in H\setminus\mathcal P$ and $\delta=0$.
Both hypotheses of Lemma~\ref{lem:exp-projection} hold, and the column operation is
\begin{equation}\label{eq:exp-hyperplane-proper-map}
\binom{\kappa(\beta)}{\beta}\longmapsto
\binom{\kappa(\beta)}{\pi(\beta)}.
\end{equation}
Injectivity on $H$ and $\mathcal P\cap\ker\pi=\emptyset$ give distinct nonzero
syndromes. No puncturing is needed, so $n'=n$, $s'=s-1$, and $\mathcal C_2'$ is projective.

\phantomsection\label{case:exp-hyperplane-full}
If $\mathcal P=H$, choose $v\in H$ and $\delta=\kappa(v)$.
Lemma~\ref{lem:exp-projection} now applies with a supported $v$, giving
\begin{equation}\label{eq:exp-rank6-hyperplane-reduction}
\binom{\kappa(\beta)}{\beta}\longmapsto
\binom{\kappa(\beta)+\beta_m\kappa(v)}{\pi(\beta)},\qquad\beta\in H.
\end{equation}
The column at $v$ becomes $(\kappa(v)+\kappa(v);0)=0$. Since $|H|=2^{s-1}$ and $\pi|_H$ is injective, $\pi(H)=\F_2^{s-1}$.
Puncturing the zero column therefore leaves the syndrome set
$\pi(H\setminus\{v\})=\F_2^{s-1}\setminus\{0\}$, with $n'=n-1$ and $s'=s-1$.
The zero column contributes nothing to the row ranks, and $\mathcal C_2'$ is projective.
\end{proof}

\subsection{Finite case analysis for the length lower bound}\label{app:exp-decide-one}

This subsection assembles the results of
Appendices~\ref{app:exp-search}--\ref{app:exp-rank-reductions} into the proof of
Theorem~\ref{thm:exp-floor38}: no CSS code with $k=3$, $d_Z\ge3$, and the nine CH conditions has
length $n\le38$. The proof first lowers the stabilizer rank to $s\in\{6,7,8\}$. At each remaining
rank, it then excludes candidate supports by their indicator polynomial $c$ alone, in ascending
order of $\deg c$ and one affine-equivalence class of $c$ at a time. Every support that these tests do not exclude is passed to
Algorithm~\ref{alg:exp-search}, which solves for the logical rows $\kappa$ and finds none.

Appendix~\ref{app:exp-lengths-flat} characterizes the admissible support $\mathcal P$
in $\F_2^s$. Appendix~\ref{app:exp-rank-reductions} examines when the stabilizer rank can be
reduced, and Appendix~\ref{app:exp-search} decides whether a given $\mathcal P$ admits logical
labels $\kappa$ satisfying the CH conditions. See Figure~\ref{fig:exp-floor-flow} for a visual
guide to the steps of the following proof.

\begin{proof}[Proof of Theorem~\ref{thm:exp-floor38}]

\medskip\noindent\emph{Step 1. Restrict the stabilizer ranks.}\quad
Suppose a generator matrix $G=[K;S]$ of length at most $38$ satisfies the nine CH
conditions with $k=3$ and $d_Z\ge3$.
Lemma~\ref{lem:exp-projred} makes $\mathcal C_2$ projective. We then apply
Lemma~\ref{lem:exp-levelred} repeatedly to lower the stabilizer rank of $S$.
For $n\le38$ and $s\ge9$, the counting inequality of Lemma~\ref{lem:exp-levelred} holds:
\begin{equation}\label{eq:exp-floor-ranks}
2^s-1>n+\Bigl\lfloor\tbinom n2\big/2\Bigr\rfloor.
\end{equation}

Each reduction strictly decreases rank, so we obtain $S$ with $s\le8$.
Since Proposition~\ref{prop:exp-lb} excludes $s\le5$, only $s\in\{6,7,8\}$ remain.

For these ranks, Theorem~\ref{thm:exp-spectrum} expresses each candidate support as
$\mathcal P=\supp(c)\setminus\{0\}$ for $c\in\RM(s-4,s)$,
with length $n=|c|-c(0)$.

By Appendix~\ref{app:exp-affine-equivalence}, it suffices to test one polynomial $c$
from each affine-equivalence class in $\RM(s-4,s)$, together with its translated supports
$\mathcal P_a$ of~\eqref{eq:exp-floor-origin}. Each step below checks the rank of $\mathcal P_a$.

\medskip\noindent\emph{Step 2. Exclude ranks $6$ and $7$.}\quad
For each rank, we first determine the possible weights of $c$ and, at each weight, the
affine-equivalence classes of $c$. The flat and hyperplane cases are then excluded, and
Algorithm~\ref{alg:exp-search} is applied to the remaining candidate supports $\mathcal P_a$.

At $s=6$, the Dickson forms~\eqref{eq:exp-dickson-general} in
Appendix~\ref{app:exp-quadratic-normal-forms} classify $c\in\RM(2,6)$.
Table~\ref{tab:exp-spectrum} lists the possible weights $16,24,28,32,36$ when $n\le38$.
Weight $16$ gives $n\in\{15,16\}$ and an affine $4$-flat, which
Lemma~\ref{lem:exp-flatobstruction} excludes. The other weights give six affine-equivalence
classes of $c$, counted in the $s=6$ rows of Table~\ref{tab:exp-floorcells}: the degree-one class
$c=y_6$ and five quadratic classes.

The degree-one representative $c(y)=y_6$ has support $H=\{y:y_6=1\}$ and weight $32$.
If $a\notin H$, then $\mathcal P_a=H+a$ is a full hyperplane not containing zero, so
Corollary~\ref{cor:exp-hyperplane-reduction}, through the
\hyperref[case:exp-hyperplane-full]{reduction}~\eqref{eq:exp-rank6-hyperplane-reduction},
lowers rank $6$ to $5$.
If $a\in H$, then $\mathcal P_a=(H+a)\setminus\{0\}$ lies in a hyperplane through zero and
has actual rank $5$ by the \hyperref[case:exp-hyperplane-origin]{span bound} of
Appendix~\ref{app:exp-flat-supports}.
Proposition~\ref{prop:exp-lb} excludes the resulting generator matrices of rank $5$.

Algorithm~\ref{alg:exp-search} finds no CH triple on the supports $\mathcal P_a$ of
the five quadratic classes, as recorded in the $s=6$ rows of Table~\ref{tab:exp-floorcells}.
Hence rank $6$ is excluded for $n\le38$.

At $s=7$, the candidates are $c\in\RM(3,7)$ with $|c|\le n+1<40$, and
Table~\ref{tab:exp-spectrum} lists their possible weights. The code $\RM(3,7)$ has minimum weight
$16$, which gives an affine $4$-flat that Lemma~\ref{lem:exp-flatobstruction}
excludes. The weights $24$, $28$, and $36$ are associated with the candidates counted in the $s=7$ rows of
Table~\ref{tab:exp-floorcells}.

Weight $32$ needs one more distinction, because the classification of
Kasami--Tokura--Azumi at this weight covers only the polynomials that are not divisible by an
affine function. A polynomial $c$ is \emph{divisible by an
affine function} $\ell(\beta)=a\cdot\beta+b$, with $a\ne0$, when $c=\ell c'$ for some
polynomial $c'$. Then $c(\beta)=1$ forces $\ell(\beta)=1$, so $\supp(c)$ lies in the affine
hyperplane $H=\{\beta:\ell(\beta)=1\}$. Table~\ref{tab:exp-s7-w32} sorts the weight-$32$
polynomials $c$ by degree and by this divisibility. No weight-$32$ polynomial has degree one,
since a nonconstant affine function on $\F_2^7$ has weight $64$.

\begin{table}[htbp]
\centering
\small
\begin{tabular}{@{}c p{0.2\linewidth} l p{0.36\linewidth}@{}}
\toprule
$\deg c$ & divisible by an affine function & classes & verdict\\
\midrule
$2$ & yes & $y_1y_2$ & $\supp(c)$ in a hyperplane, rank at most $6$ after reduction\\
$3$ & yes & every such $c$ & $\supp(c)$ in a hyperplane, rank at most $6$ after reduction\\
$3$ & no & \begin{tabular}[t]{@{}l@{}}$y_1y_2y_3+y_4y_5+y_1y_4y_5$,\\
$y_1y_2y_3+y_2y_4y_5+y_3y_4y_6$,\\ $y_1y_2y_3+y_3y_4y_5+y_4y_6y_7$\\
\cite[Cor.~1(1)]{kasamiTokuraAzumi1976}\end{tabular} & Algorithm~\ref{alg:exp-search}, no $\kappa_3$\\
\bottomrule
\end{tabular}
\caption{The weight-$32$ polynomials $c\in\RM(3,7)$ in Step~2, sorted by degree and by
divisibility by an affine function.}
\label{tab:exp-s7-w32}
\end{table}

Every polynomial $c$ divisible by an affine function is excluded without
Algorithm~\ref{alg:exp-search}. By~\eqref{eq:exp-floor-origin}, $\mathcal P_a\subseteq H+a$ for
every $a\in\F_2^7$. If $0\in H+a$, the \hyperref[case:exp-hyperplane-origin]{span bound} in
Appendix~\ref{app:exp-flat-supports} gives actual rank at most $6$. Otherwise
$n\le32<|H+a|=64$, so $\mathcal P_a\subsetneq H+a$, and
Corollary~\ref{cor:exp-hyperplane-reduction} lowers the rank to $6$. Ranks at most $6$ are
already excluded. The quadratic class $y_1y_2$ is divisible by $y_1$, so it is one of these
polynomials. For the weight-$32$ polynomials not divisible by an affine function,
Kasami--Tokura--Azumi~\cite[Cor.~1(1)]{kasamiTokuraAzumi1976} give three cubic classes.

Algorithm~\ref{alg:exp-search} finds no CH triple on the remaining supports at rank
$7$, as recorded in the $s=7$ rows of Table~\ref{tab:exp-floorcells}.
Consequently rank $7$ is excluded at any $n\le38$, including the weights $32$ and $36$
needed in the later rank reductions.

\medskip\noindent\emph{Step 3. Reduce candidates at rank $8$.}\quad
We first exclude candidates at rank $8$ that reduce to a rank already excluded
in Steps~1--2. If $|c|\le30$, then $n\le30$ by~\eqref{eq:exp-puncture}, and the counting
hypothesis of Lemma~\ref{lem:exp-levelred} holds:
\begin{equation}\label{eq:exp-rank8-small}
n+\Bigl\lfloor\tbinom n2\big/2\Bigr\rfloor\le247<2^8-1.
\end{equation}
Hence every candidate with $|c|\le30$ reduces to a rank excluded in Steps~1--2.

By~\eqref{eq:exp-affine-constant}, $|c|$ is even. Together with $|c|\le n+1\le39$, only
weights in the interval $32\le|c|<40$ remain.

To exclude lower-degree polynomials, suppose $\deg c<4$, so $c\in\RM(3,8)$.
The code $\RM(3,8)$ has minimum weight $32$, and
Kasami--Tokura~\cite[Thm.~1]{kasamiTokura1970} exclude weights strictly between $32$ and $40$.
Thus $32\le|c|<40$ forces $|c|=32$, whose support is an affine $5$-flat.

Every translated support $\mathcal P_a$ remains inside an affine $5$-flat.
The \hyperref[case:exp-hyperplane-origin]{span bound} of
Appendix~\ref{app:exp-flat-supports} therefore gives actual rank at most $6$, excluded in
Steps~1--2. Only degree $4$ and weights $32,34,36,38$ remain at rank $8$.

\medskip\noindent\emph{Step 4. Exclude degree-$4$ cases at rank $8$.}\quad
Steps~4--7 exclude the remaining degree-$4$ polynomials by expressing them in
six-variable polynomials and enumerating the possibilities.

Theorem~1 of Kasami--Tokura--Azumi~\cite{kasamiTokuraAzumi1976} gives an affine
$6$-flat in $\F_2^8$ on which $c$ vanishes. Choose affine coordinates $(u,y_7,y_8)$ with
$u=(y_1,\ldots,y_6)\in\F_2^6$ so that this flat is given by $y_7=y_8=0$.
Since $c$ vanishes on the flat, $c(u,0,0)=0$ for every $u$, so no monomial of $c$ uses only the
variables $u$. Grouping the remaining monomials by whether they contain $y_7$, $y_8$, or both
gives
\begin{equation}\label{eq:exp-rank8-decomposition}
c(u,y_7,y_8)=y_7f(u)+y_8g(u)+y_7y_8h(u).
\end{equation}
Since $\deg c=4$, the polynomials $f$ and $g$ have degree at most $3$, and $h$ has
degree at most $2$. Setting $(y_7,y_8)$ to each of its four values in~\eqref{eq:exp-rank8-decomposition}
gives
\begin{equation*}
c(u,0,0)=0,\quad c(u,1,0)=f(u),\quad c(u,0,1)=g(u),\quad
c(u,1,1)=f(u)+g(u)+h(u).
\end{equation*}
The weight $|c|$ counts the points $(u,y_7,y_8)$ of $\F_2^8$ at which $c=1$.
Grouping these $256$ points by the value of $(y_7,y_8)$ gives
\begin{equation}\label{eq:exp-rank8-weight-sum}
|c|=|f|+|g|+|f+g+h|.
\end{equation}
The candidate supports still require the origin recovery~\eqref{eq:exp-floor-origin}
and rank checks of Appendix~\ref{app:exp-affine-equivalence}.

\medskip\noindent\emph{Step 5. Exclude vanishing components and bound their weights.}\quad
To exclude the cases $f=0$, $g=0$, or $f+g+h=0$, we examine the hyperplane that
contains $\supp(c)$ in each case:
\begin{equation}\label{eq:exp-rank8-vanishing-components}
\begin{aligned}
f=0:\quad &c=y_8(g+y_7h),
&&\supp(c)\subseteq\{y_8=1\},\\
g=0:\quad &c=y_7(f+y_8h),
&&\supp(c)\subseteq\{y_7=1\},\\
f+g+h=0:\quad &c=y_7(1+y_8)f+y_8(1+y_7)g,
&&\supp(c)\subseteq\{y_7+y_8=1\}.
\end{aligned}
\end{equation}

Denote any of the hyperplanes in~\eqref{eq:exp-rank8-vanishing-components} by $H$.
For each $a\in\F_2^8$, Equation~\eqref{eq:exp-floor-origin} gives
$\mathcal P_a\subseteq H+a$. If $0\in H+a$, the
\hyperref[case:exp-hyperplane-origin]{span bound in Appendix~\ref{app:exp-flat-supports}}
gives actual rank at most $7$.
Otherwise $n\le38<|H+a|=128$ makes $\mathcal P_a$ a proper subset, and
Corollary~\ref{cor:exp-hyperplane-reduction} lowers rank $8$ to $7$.
The exclusion of rank $7$ in Step~2 therefore forces all three polynomials
$f$, $g$, and $f+g+h$ to be nonzero for every remaining candidate at rank $8$.

To avoid enumerating equivalent permutations of the nonzero components, order their
weights as $0<|f|\le|g|\le|f+g+h|$.
An invertible linear change of $(y_7,y_8)$ permits this ordering:
$\operatorname{GL}(2,2)$ permutes the three nonzero values of $(y_7,y_8)$ and hence the
three restrictions $f$, $g$, and $f+g+h$. With this ordering,
\eqref{eq:exp-rank8-weight-sum} gives $|c|\ge|f|+2|g|$, so $|c|<40$ bounds the two smaller
weights.

The bound $|f|+2|g|<40$ leaves only a few weight pairs, which we now list. Both $f$
and $g$ are nonzero polynomials of degree at most $3$ in $6$ variables, that is, nonzero codewords
of $\RM(3,6)$. The code $\RM(3,6)$ has minimum weight $8$, and its only weights below $16$
are $8$, $12$, and $14$~\cite[Thm.~1]{kasamiTokura1970}. If $|g|\ge16$, then
$|f|+2|g|\ge8+32=40$. If $|g|\le14$, the condition $|f|+2|g|<40$ with $|f|\le|g|$ leaves exactly
\begin{equation}\label{eq:exp-quadpairs}
(|f|,|g|)\in\{(8,8),(8,12),(12,12),(8,14)\},
\end{equation}
because the pairs $(12,14)$ and $(14,14)$ give $40$ and $42$. The four pairs are the rows
of Table~\ref{tab:exp-orbitcover}.
\par

\medskip\noindent\emph{Step 6. Enumerate the remaining polynomials.}\quad
Fix $f$ to be the Kasami--Tokura normal form~\cite[Thm.~1]{kasamiTokura1970}
listed in Table~\ref{tab:exp-orbitcover}. Subsequent changes of the six coordinates must
preserve this chosen $f$. For an invertible affine map $L$ of $\F_2^6$, that is, $L\in\operatorname{AGL}(6,2)$, with
$f\circ L^{-1}=f$, the substitution in~\eqref{eq:exp-rank8-decomposition} acts on the
remaining polynomials by
\begin{equation}\label{eq:exp-rank8-polynomial-transport}
(g,h)\longmapsto(g\circ L^{-1},h\circ L^{-1}).
\end{equation}
These maps $L$ form the subgroup of $\operatorname{AGL}(6,2)$ that fixes $f$, and the orbit of $g$ under them is
\begin{equation*}
\operatorname{Orb}_f(g)=\{\,g\circ L^{-1}:L\in\operatorname{AGL}(6,2),\ f\circ L^{-1}=f\,\}.
\end{equation*}
Choose one representative $g$ from each orbit.

For each representative pair $(f,g)$ and each candidate weight $|c|$,
enumerate every $h$ of degree at most $2$ satisfying the weight constraint
from~\eqref{eq:exp-rank8-weight-sum}:
\begin{equation}\label{eq:exp-rank8-h-weight}
|f+g+h|=|c|-|f|-|g|.
\end{equation}
Lemma~\ref{lem:exp-affine-invariance} shows that the substitution
\eqref{eq:exp-rank8-polynomial-transport} preserves the degree of $h$ and the weights of
$f$, $g$, and $f+g+h$. The substitution therefore preserves the weight
constraint~\eqref{eq:exp-rank8-h-weight} and the total weight $|c|$.

A computation with \texttt{nauty}~\cite{mckayPracticalGraphIsomorphism2014} lists one
representative $g$ of each orbit for every row of Table~\ref{tab:exp-orbitcover}. Three checks
show that no polynomial $g$ is missed. First, the listed orbits are pairwise disjoint. Second,
every $g$ in a listed orbit has the weight of its row. Third, for each row the orbit sizes sum to
the number of all $g$ of that weight, given in the third column. Hence every remaining candidate
$c$ at rank $8$ is affine-equivalent to $y_7f+y_8g+y_7y_8h$ with $f$ from
Table~\ref{tab:exp-orbitcover}, $g$ a listed representative, and $h$ one of the enumerated
polynomials.

\begin{table}[!hbp]
\centering
\begin{tabular}{@{}c c r r@{}}
\toprule
$(\lvert f\rvert,\lvert g\rvert)$ & $f$ normal form & number of $g$ & orbits of $g$\\
\midrule
$(8,8)$   & $y_1y_2y_3$            & $11{,}160$       & $20$\\
$(8,12)$  & $y_1y_2y_3$            & $1{,}749{,}888$  & $88$\\
$(12,12)$ & $y_1(y_2y_3{+}y_4y_5)$ & $1{,}749{,}888$  & $2{,}167$\\
$(8,14)$  & $y_1y_2y_3$            & $22{,}855{,}680$ & $124$\\
\bottomrule
\end{tabular}
\caption{Orbit partitions at $s=8$, with $f$ fixed to normal form and $g$ reduced by
invertible affine coordinate changes that preserve $f$. The third column counts all
$g$ with the weight in the first column.}
\label{tab:exp-orbitcover}
\end{table}

\medskip\noindent\emph{Step 7. Exclude the remaining candidates at rank $8$.}\quad
It remains to exclude logical labels for the polynomials enumerated in Step~6.
For each such $c$, choose $a_0\in\F_2^8$ with $c(a_0)=0$ and form the full translated
support $\mathcal P_{a_0}$ of~\eqref{eq:exp-floor-origin}.
When the matrix $S$ whose columns are the points of $\mathcal P_{a_0}$ has actual rank below $8$, Steps~1--2
already exclude a generator matrix satisfying the CH conditions on that support.

Algorithm~\ref{alg:exp-search} finds no three logical rows satisfying the CH conditions
on any remaining $\mathcal P_{a_0}$. Corollary~\ref{cor:exp-one-support} then excludes logical
labels satisfying the CH conditions on $\mathcal P_a$ for every $a\in\F_2^8$, including the punctured supports.

All cases with $n\le38$ have now been excluded, proving the theorem.
Table~\ref{tab:exp-floorcells} summarizes the cases at ranks $6$, $7$, and $8$.
\end{proof}

\begin{table}[!htb]
\centering
\small
\begin{tabular}{@{}c c r r l@{}}
\toprule
$s$ & $\lvert c\rvert$ & \shortstack{Number\\enumerated} & \shortstack{Supports\\tested} & result\\
\midrule
$6$ & $16$ & -- & -- & Excluded by Lemma~\ref{lem:exp-flatobstruction}\\
$6$ & $24$ & $1$ & $1$ & no $\kappa_3$\\
$6$ & $28$ & $1$ & $1$ & no $\kappa_3$\\
$6$ & $32$ & $3$ & $2$ & no $\kappa_3$\\
$6$ & $36$ & $1$ & $1$ & no $\kappa_3$\\
$7$ & $16$ & -- & -- & Excluded by Lemma~\ref{lem:exp-flatobstruction}\\
$7$ & $24$ & $1$ & $1$ & no $\kappa_3$\\
$7$ & $28$ & $2$ & $2$ & no $\kappa_3$\\
$7$ & $32$ & $4$ & $3$ & no $\kappa_3$\\
$7$ & $36$ & $9{,}792$ & $7{,}038$ & no $\kappa_3$\\
$8$ & $32$ & $108$ & $13{,}632$ & no $\kappa_3$\\
$8$ & $34$ & $108$ & $3{,}712$ & no $\kappa_3$\\
$8$ & $36$ & $2{,}399$ & $292{,}128$ & no $\kappa_3$\\
$8$ & $38$ & $2{,}399$ & $160{,}256$ & no $\kappa_3$\\
\bottomrule
\end{tabular}
\par\smallskip
\begin{tabular}{@{}p{0.31\linewidth} p{0.65\linewidth}@{}}
\toprule
Rows & What the third column counts\\
\midrule
$s=6$; $s=7$, $|c|=24,28$ & Affine-equivalence classes of $c$\\
$s=7$, $|c|=32$ & The quadratic class $y_1y_2$ and the three cubic classes not divisible by an affine function (Table~\ref{tab:exp-s7-w32})\\
$s=7$, $|c|=36$ & Weight-$36$ functions obtained by fixing a cubic representative and enumerating quadratic, linear, and constant terms~\cite{cohenHelminck1988,braekenBooleanClassification2004}\\
$s=8$ & Orbits of $g$ under affine maps preserving the selected $f$ (Table~\ref{tab:exp-orbitcover}). Step~6 then enumerates the permitted $h$\\
\bottomrule
\end{tabular}
\caption{Cases in the proof of Theorem~\ref{thm:exp-floor38}.
The fourth column counts full translated supports tested by Algorithm~\ref{alg:exp-search},
possibly including equivalent supports.
At $s=8$, each count includes the permitted $h$ from Step~6.
Dashes mark cases excluded analytically. An entry ``no $\kappa_3$'' means that no tested
pair extends to a CH triple on any support counted in that row.
At $s=6$ and $|c|=32$, Step~2 excludes the degree-one class through the hyperplane cases.}
\label{tab:exp-floorcells}
\end{table}

\FloatBarrier

\begin{figure}[!htbp]
\centering
\begingroup
\tikzset{every picture/.style={}}
\newcommand{\flowbadge}[3]{\begingroup\setlength{\fboxsep}{1.2pt}%
  \raisebox{.15ex}{\fcolorbox{#1}{#2}{\textcolor{#1}{\bfseries\fontsize{7.2}{8}\selectfont #3}}}%
  \endgroup}
\newcommand{\badgeR}{\flowbadge{blue!80!black}{blue!7}{R}}
\newcommand{\badgeH}{\flowbadge{cyan!55!black}{cyan!7}{H}}
\newcommand{\badgeKT}{\flowbadge{teal!75!black}{teal!7}{KT}}
\newcommand{\badgeKTA}{\flowbadge{violet!80!black}{violet!7}{KTA}}
\newcommand{\badgeF}{\flowbadge{red!70!black}{red!7}{F}}
\newcommand{\badgeA}{\flowbadge{brown!80!black}{brown!7}{A}}
\newcommand{\badgeL}{\flowbadge{orange!85!black}{orange!7}{L}}
\begin{tikzpicture}[
  x=1cm,y=1cm,
  proofbox/.style={draw,rounded corners=2pt,minimum width=7.6cm,
    inner sep=6pt,outer sep=2pt,font={\fontsize{9}{11}\selectfont},align=center,line width=.6pt},
  source/.style={draw,rounded corners=2pt,minimum width=4.4cm,
    inner sep=7pt,outer sep=2pt,font={\fontsize{8.5}{10.5}\selectfont},align=center,line width=.6pt},
  flow/.style={-{Stealth[length=2.2mm]},line width=.8pt},
  condition/.style={font=\footnotesize,anchor=west,inner sep=3pt}
]
\draw[black!30,rounded corners=4pt,line width=.5pt]
  (-4.45,1.55) rectangle (4.25,-7.8);
\node[fill=white,inner sep=3pt,rotate=90,font={\fontsize{9}{11}\selectfont}]
  at (-4.45,-3.1) {Restrict rank, weight, and degree};
\draw[black!30,rounded corners=4pt,line width=.5pt]
  (-4.45,-8.0) rectangle (4.25,-14.05);
\node[fill=white,inner sep=3pt,rotate=90,font={\fontsize{9}{11}\selectfont}]
  at (-4.45,-11.0) {Constrain the decomposition};
\draw[black!30,rounded corners=4pt,line width=.5pt]
  (-4.45,-14.3) rectangle (4.25,-20.3);
\node[fill=white,inner sep=3pt,rotate=90,font={\fontsize{9}{11}\selectfont}]
  at (-4.85,-17.3) {Enumerate and exclude};
\node[proofbox] (s1) at (0,0) {\begin{tabular}{@{}c@{}}
\textbf{1. Restrict the stabilizer ranks}\\
$n\le38$: projective reduction (Lemma~\ref{lem:exp-projred})\\
\badgeR\ Repeat rank reduction while $s\ge9$\\
$s\le5$ excluded (Proposition~\ref{prop:exp-lb}); retain $s\in\{6,7,8\}$\\
Theorem~\ref{thm:exp-spectrum}: $c\in\RM(s-4,s)$\\
\badgeA\ $n=|c|-c(0)$; form $\mathcal P_a$ and check rank
\end{tabular}};
\node[proofbox] (s2) at (0,-3.15) {\begin{tabular}{@{}c@{}}
\textbf{2. Exclude ranks $6$ and $7$}\\
Classify $s=6$: Dickson~\eqref{eq:exp-dickson-general}~\cite[Ch.~15, Thms.~4--5]{macwilliams1977theory}\\
$s=7$: \badgeKT\ \badgeKTA\ Ax: $4\mid |c|$~\cite[Sec.~4.1.5]{carletBooleanFunctions2021}\\
\badgeA\ Recover supports\\
\badgeF\ Weight $16$: Lemma~\ref{lem:exp-flatobstruction}\\
\badgeH\ Exclude hyperplane families; \badgeL\ test the remainder\\
(Table~\ref{tab:exp-floorcells})
\end{tabular}};
\node[proofbox] (s3) at (0,-6.3) {\begin{tabular}{@{}c@{}}
\textbf{3. Reduce candidates at rank $8$}\\
Exclude by smaller rank (Steps~1--2):\\
$|c|\le30$: \badgeR\ reduce to $s\le7$ using~\eqref{eq:exp-rank8-small}\\
\badgeF\ Coefficient parity gives even $|c|$\\
\badgeF\ \badgeKT\ $\deg c<4$: weight-$32$ flat, rank $\le6$\\
Retain $\deg c=4$, $|c|\in\{32,34,36,38\}$
\end{tabular}};
\node[proofbox] (s4) at (0,-9.45) {\begin{tabular}{@{}c@{}}
\textbf{4. Exclude degree-$4$ cases at rank $8$}\\
$\deg c=4$: \badgeKTA\ gives a zero $6$-dimensional flat\\
$c=y_7f+y_8g+y_7y_8h$, $\deg f,\deg g\le3$, $\deg h\le2$\\
$|c|=|f|+|g|+|f+g+h|$; \badgeA\ form $\mathcal P_a$
\end{tabular}};
\node[proofbox] (s5) at (0,-12.6) {\begin{tabular}{@{}c@{}}
\textbf{5. Exclude vanishing components}\\
\textbf{and bound their weights}\\
$f=0$, $g=0$, or $f+g+h=0$ gives\\
\badgeA\ $\mathcal P_a\subseteq H+a$\\
\badgeH\ $0\in H+a$: rank $\le7$; otherwise rank $8\to7$\\
($n\le38<128$: proper support); apply Step~2\\
\badgeKT\ \badgeF\ $|c|<40$: four weight pairs (Table~\ref{tab:exp-orbitcover})
\end{tabular}};
\node[proofbox] (s6) at (0,-15.75) {\begin{tabular}{@{}c@{}}
\textbf{6. Enumerate the remaining polynomials}\\
\badgeKT\ Fix $f$; enumerate $g$ under maps preserving $f$\\
Enumerate $h$: $\deg h\le2$, weight condition~\eqref{eq:exp-rank8-h-weight}\\
\texttt{nauty}~\cite{mckayPracticalGraphIsomorphism2014}: disjoint orbits and class membership\\
with size totals establish coverage (Table~\ref{tab:exp-orbitcover})
\end{tabular}};
\node[proofbox,minimum width=7.3cm,fill=black!4] (s7) at (0,-18.9) {\begin{tabular}{@{}c@{}}
\textbf{7. Exclude the remaining candidates at rank $8$}\\
\badgeA\ Step~6: choose $a_0$ with $c(a_0)=0$\\
\badgeL\ No CH triple on $\mathcal P_{a_0}$ (Table~\ref{tab:exp-floorcells})\\
Corollary~\ref{cor:exp-one-support}: no CH triple on any $\mathcal P_a$\\
All cases excluded; no code has $n\le38$
\end{tabular}};
\node[source,draw=blue!80!black] (rank) at (7.05,0) {\begin{tabular}{@{}c@{}}
\badgeR\ \textbf{Counting reduction}\\
Lemma~\ref{lem:exp-levelred}, Appendix~\ref{app:exp-levelred-proof}\\
Lower rank, no greater length\\
Recheck rank after reduction\\
Failed bound: no exclusion
\end{tabular}};
\node[source,draw=cyan!55!black] (hyper) at (7.05,-3.15) {\begin{tabular}{@{}c@{}}
\badgeH\ \textbf{Hyperplane cases}\\
\hyperref[case:exp-hyperplane-origin]{Span bound: Appendix~\ref{app:exp-flat-supports}}\\
Reduction: Corollary~\ref{cor:exp-hyperplane-reduction}, Table~\ref{tab:exp-hyperplane-reductions}\\
Exclude only at a rank already settled
\end{tabular}};
\node[source,draw=teal!75!black] (kt) at (7.05,-6.3) {\begin{tabular}{@{}c@{}}
\badgeKT\ \textbf{Kasami--Tokura}\\
Low-weight classes and gap\\
\cite[Thm.~1]{kasamiTokura1970}
\end{tabular}};
\node[source,draw=violet!80!black] (kta) at (7.05,-9.45) {\begin{tabular}{@{}c@{}}
\badgeKTA\ \textbf{Kasami--Tokura--Azumi}\\
Classes and flats where $c=0$\\
\cite[Thm.~1, Cor.~1(1)]{kasamiTokuraAzumi1976}
\end{tabular}};
\node[source,draw=red!70!black] (facts) at (7.05,-12.6) {\begin{tabular}{@{}c@{}}
\badgeF\ \textbf{Shared facts (Appendix~\ref{app:exp-facts})}\\
\hyperref[fact:exp-rm-minweight]{Minimum weight and flats}\\
Parity~\eqref{eq:exp-topcoeff-parity}: even weight
\end{tabular}};
\node[source,draw=brown!80!black] (affine) at (7.05,-15.75) {\begin{tabular}{@{}c@{}}
\badgeA\ \textbf{Affine representatives}\\
\textbf{and translated supports}\\
Appendix~\ref{app:exp-affine-equivalence}, \eqref{eq:exp-floor-origin}\\
Labels stay on their columns\\
(Lemma~\ref{lem:exp-affine-support})\\
If $c(a_0)=0$ and $\mathcal P_{a_0}$ has no CH triple,\\
no $\mathcal P_a$ has one (Cor.~\ref{cor:exp-one-support})
\end{tabular}};
\node[source,draw=orange!85!black] (test) at (7.05,-18.9) {\begin{tabular}{@{}c@{}}
\badgeL\ \textbf{Fixed-support label tests}\\
Appendix~\ref{app:exp-search}, Algorithm~\ref{alg:exp-search}\\
Pair tests~\eqref{eq:exp-admissible-pair}\\
Third-label system~\eqref{eq:exp-third-label-coefficients}
\end{tabular}};
\foreach \a/\b in {1/2,2/3,3/4,4/5,5/6,6/7}
  \draw[flow] (s\a.south) -- (s\b.north);
\end{tikzpicture}
\endgroup
\caption{Proof steps for Theorem~\ref{thm:exp-floor38}.
Solid arrows follow the numbered steps. Matching lettered badges identify the results
from earlier sections and the literature used by each step.}
\label{fig:exp-floor-flow}
\end{figure}

\bibliography{references/Gates}

@article{haahCodesProtocolsDistilling2017,
  title = {Codes and Protocols for Distilling {$T$}, controlled-{$S$}, and {Toffoli} Gates},
  author = {Haah, Jeongwan and Hastings, Matthew B.},
  journal = {Quantum},
  volume = {2},
  pages = {71},
  year = {2018},
  doi = {10.22331/q-2018-06-07-71},
  eprint = {1709.02832},
  primaryclass = {quant-ph},
  archiveprefix = {arXiv}
}

@article{bargGeometricStructureTransversal2024,
  title = {Geometric Structure and Transversal Logic of Quantum {{Reed-Muller}} Codes},
  author = {Barg, Alexander and Coble, Nolan J. and Hangleiter, Dominik and Kang, Christopher},
  year = {2026},
  month = jan,
  journal = {IEEE Transactions on Information Theory},
  volume = {72},
  number = {1},
  pages = {415--436},
  eprint = {2410.07595},
  primaryclass = {quant-ph},
  doi = {10.1109/TIT.2025.3592631},
  archiveprefix = {arXiv}
}

@article{bravyiMagicStateDistillation2012,
  title = {Magic State Distillation with Low Overhead},
  author = {Bravyi, Sergey and Haah, Jeongwan},
  year = 2012,
  month = nov,
  journal = {Physical Review A},
  volume = {86},
  number = {5},
  eprint = {1209.2426},
  pages = {052329},
  issn = {1050-2947, 1094-1622},
  doi = {10.1103/PhysRevA.86.052329},
  archiveprefix = {arXiv}
}

@article{brownColorCodeLogical2025,
  title = {Color Code with a Logical Control-{$S$} Gate Using Transversal {$T$} Rotations},
  author = {Brown, Benjamin J.},
  year = 2025,
  month = aug,
  journal = {Physical Review Letters},
  volume = {135},
  number = {7},
  pages = {070602},
  issn = {0031-9007, 1079-7114},
  doi = {10.1103/lwxd-fdlb}
}

@article{campbellUnifiedFrameworkMagic2017,
  title = {Unified framework for magic state distillation and multiqubit gate synthesis with reduced resource cost},
  author = {Campbell, Earl T. and Howard, Mark},
  year = 2017,
  month = feb,
  journal = {Physical Review A},
  volume = {95},
  number = {2},
  eprint = {1606.01904},
  primaryclass = {quant-ph},
  pages = {022316},
  issn = {2469-9926, 2469-9934},
  doi = {10.1103/PhysRevA.95.022316},
  archiveprefix = {arXiv}
}

@article{huClimbingDiagonalClifford2021,
  title = {Climbing the {{Diagonal Clifford Hierarchy}}},
  author = {Hu, Jingzhen and Liang, Qingzhong and Calderbank, Robert},
  year = {2025},
  month = oct,
  journal = {Academia Quantum},
  volume = {2},
  number = {4},
  eprint = {2110.11923},
  primaryclass = {quant-ph},
  doi = {10.20935/AcadQuant7957},
  archiveprefix = {arXiv}
}

@article{jainTransversalCliffordTgate2025,
  title = {Transversal {{Clifford}} and {{T-gate}} Codes of Short Length and High Distance},
  author = {Jain, Shubham P. and Albert, Victor V.},
  year = 2025,
  journal = {IEEE Journal on Selected Areas in Information Theory},
  volume = {6},
  eprint = {2408.12752},
  primaryclass = {quant-ph},
  pages = {127--137},
  issn = {2641-8770},
  doi = {10.1109/JSAIT.2025.3570832},
  archiveprefix = {arXiv}
}

@article{nezamiClassificationSmallTriorthogonal2022,
  title = {Classification of Small Triorthogonal Codes},
  author = {Nezami, Sepehr and Haah, Jeongwan},
  year = 2022,
  month = jul,
  journal = {Physical Review A},
  volume = {106},
  number = {1},
  pages = {012437},
  publisher = {American Physical Society},
  doi = {10.1103/PhysRevA.106.012437}
}

@article{bombinFaultTolerantPostselectionLowOverhead2024,
  title = {Fault-{{Tolerant Postselection}} for {{Low-Overhead Magic State Preparation}}},
  author = {Bomb{\'i}n, H{\'e}ctor and Pant, Mihir and Roberts, Sam and Seetharam, Karthik I.},
  year = {2024},
  month = jan,
  journal = {PRX Quantum},
  volume = {5},
  number = {1},
  pages = {010302},
  issn = {2691-3399},
  doi = {10.1103/PRXQuantum.5.010302}
}

@misc{gidneyMagicStateCultivation2024,
  title = {Magic State Cultivation: Growing {{T}} States as Cheap as {{CNOT}} Gates},
  author = {Gidney, Craig and Shutty, Noah and Jones, Cody},
  year = {2024},
  month = sep,
  number = {arXiv:2409.17595},
  eprint = {2409.17595},
  primaryclass = {quant-ph},
  publisher = {arXiv},
  archiveprefix = {arXiv},
  note = {\href{https://arxiv.org/abs/2409.17595}{arXiv:2409.17595}}
}

@misc{kohPhantomCodes2025,
  title = {Entangling logical qubits without physical operations},
  author = {Koh, Jin Ming and Gong, Anqi and Diaconu, Andrei C. and Tan, Daniel Bochen and Geim, Alexandra A. and Gullans, Michael J. and Yao, Norman Y. and Lukin, Mikhail D. and Majidy, Shayan},
  year = {2026},
  eprint = {2601.20927},
  note = {\href{https://arxiv.org/abs/2601.20927}{arXiv:2601.20927}},
  archiveprefix = {arXiv},
  primaryclass = {quant-ph}
}

@article{rengaswamyOnOptimalityCSSCodes2020,
  title = {On {{Optimality}} of {{CSS Codes}} for {{Transversal}} {{$T$}}},
  author = {Rengaswamy, Narayanan and Calderbank, Robert and Newman, Michael and Pfister, Henry D.},
  year = {2020},
  month = aug,
  journal = {IEEE Journal on Selected Areas in Information Theory},
  volume = {1},
  number = {2},
  pages = {499--514},
  eprint = {1910.09333},
  primaryclass = {quant-ph},
  doi = {10.1109/JSAIT.2020.3012914},
  archiveprefix = {arXiv}
}

@article{huDesigningQuantumChannels2023,
  title = {Designing the {{Quantum Channels Induced}} by {{Diagonal Gates}}},
  author = {Hu, Jingzhen and Liang, Qingzhong and Calderbank, Robert},
  year = {2022},
  month = sep,
  journal = {Quantum},
  volume = {6},
  pages = {802},
  eprint = {2109.13481},
  primaryclass = {quant-ph},
  archiveprefix = {arXiv},
  doi = {10.22331/q-2022-09-08-802}
}

@misc{huDivisibleCodesQuantum2022,
  title = {Divisible {{Codes}} for {{Quantum Computation}}},
  author = {Hu, Jingzhen and Liang, Qingzhong and Calderbank, Robert},
  year = 2022,
  month = apr,
  number = {arXiv:2204.13176},
  eprint = {2204.13176},
  note = {\href{https://arxiv.org/abs/2204.13176}{arXiv:2204.13176}},
  primaryclass = {quant-ph},
  publisher = {arXiv},
  archiveprefix = {arXiv}
}

@misc{koutsioumpasSmallestCodeTransversal2022,
  author = {Koutsioumpas, Stergios and Banfield, Darren and Kay, Alastair},
  title = {The Smallest Code with Transversal {T}},
  year = {2022},
  eprint = {2210.14066},
  archiveprefix = {arXiv},
  primaryclass = {quant-ph},
  note = {\href{https://arxiv.org/abs/2210.14066}{arXiv:2210.14066}}
}

@article{bombinExactTopologicalQuantum2015,
  author = {Bomb{\'\i}n, H{\'e}ctor},
  title = {Gauge Color Codes: Optimal Transversal Gates and Gauge Fixing in Topological Stabilizer Codes},
  journal = {New Journal of Physics},
  volume = {17},
  number = {8},
  pages = {083002},
  year = {2015},
  doi = {10.1088/1367-2630/17/8/083002}
}

@article{mckayPracticalGraphIsomorphism2014,
  author = {McKay, Brendan D. and Piperno, Adolfo},
  title = {Practical graph isomorphism, {II}},
  journal = {Journal of Symbolic Computation},
  volume = {60},
  pages = {94--112},
  year = {2014},
  doi = {10.1016/j.jsc.2013.09.003}
}

@book{macwilliams1977theory,
  author = {MacWilliams, Florence J. and Sloane, Neil J. A.},
  title = {The Theory of Error-Correcting Codes},
  publisher = {North-Holland},
  year = {1977},
  doi = {10.1016/S0924-6509(08)X7030-8}
}

@article{campsCSST2024,
  author = {Camps-Moreno, Eduardo and L{\'o}pez, Hiram H. and Matthews, Gretchen L.
               and Ruano, Diego and San-Jos{\'e}, Rodrigo and Soprunov, Ivan},
  title = {An algebraic characterization of binary {CSS-T} codes and cyclic {CSS-T}
               codes for quantum fault tolerance},
  journal = {Quantum Information Processing},
  volume = {23},
  number = {6},
  pages = {230},
  doi = {10.1007/s11128-024-04427-5},
  eprint = {2312.17518},
  year = {2024}
}

@article{betsumiyaMunemasa2012,
  author = {Betsumiya, Koichi and Munemasa, Akihiro},
  title = {On triply even binary codes},
  journal = {Journal of the London Mathematical Society},
  volume = {86},
  number = {1},
  pages = {1--16},
  year = {2012},
  doi = {10.1112/jlms/jdr054}
}

@inproceedings{heinleinProjectiveDivisible2017,
  author = {Heinlein, Daniel and Honold, Thomas and Kiermaier, Michael and Kurz, Sascha and Wassermann, Alfred},
  title = {Projective divisible binary codes},
  booktitle = {The Tenth International Workshop on Coding and Cryptography (WCC 2017)},
  note = {\href{https://arxiv.org/abs/1703.08291}{arXiv:1703.08291}},
  eprint = {1703.08291},
  year = {2017}
}

@misc{kurzDivisible2021,
  author = {Kurz, Sascha},
  title = {Divisible codes},
  note = {Revised edition, December 2025},
  doi = {10.15495/EPub_UBT_00008789},
  howpublished = {\doi{10.15495/EPub_UBT_00008789}},
  eprint = {2112.11763},
  year = {2025}
}

@article{calderbankKantor1986,
  author = {Calderbank, A. R. and Kantor, W. M.},
  title = {The geometry of two-weight codes},
  journal = {Bulletin of the London Mathematical Society},
  volume = {18},
  number = {2},
  pages = {97--122},
  year = {1986},
  doi = {10.1112/blms/18.2.97}
}

@article{landjev2000,
  author = {Landjev, Ivan N.},
  title = {Linear codes over finite fields and finite projective geometries},
  journal = {Discrete Mathematics},
  volume = {213},
  number = {1--3},
  pages = {211--244},
  year = {2000},
  doi = {10.1016/S0012-365X(99)00183-1}
}

@incollection{assmusKeyPolynomial1998,
  author = {Assmus, Jr., E. F. and Key, J. D.},
  title = {Polynomial codes and finite geometries},
  booktitle = {Handbook of Coding Theory},
  editor = {Pless, V. S. and Huffman, W. C.},
  publisher = {Elsevier},
  volume = {2},
  chapter = {16},
  pages = {1269--1343},
  year = {1998}
}

@incollection{brouwerTwoWeight2021,
  author = {Brouwer, A. E.},
  title = {Two-weight codes},
  booktitle = {Concise Encyclopedia of Coding Theory},
  editor = {Huffman, W. C. and Kim, J.-L. and Sol{\'e}, P.},
  publisher = {CRC Press},
  chapter = {19},
  pages = {449--462},
  year = {2021}
}

@inproceedings{alyQLDPCFiniteGeometries2008,
  author = {Aly, Salah A.},
  title = {A class of quantum {LDPC} codes constructed from finite geometries},
  booktitle = {IEEE Global Telecommunications Conference (GLOBECOM)},
  year = {2008},
  pages = {1--5},
  doi = {10.1109/GLOCOM.2008.ECP.217},
  note = {\href{https://arxiv.org/abs/0712.4115}{arXiv:0712.4115}}
}

@misc{willsClassificationGeneralised2026,
  title = {Classification of Generalised Triorthogonal Codes through Length 54},
  author = {Wills, Adam and Jain, Shubham P. and Singh, Shraddha},
  year = {2026},
  eprint = {2609.30860},
  archiveprefix = {arXiv},
  primaryclass = {quant-ph},
  note = {\href{https://arxiv.org/abs/2609.30860}{arXiv:2609.30860}}
}

@misc{dalfaveroSatisfyingQuantumCodes2026,
  title = {Satisfying Quantum Codes: Physics-Informed and Hardware-Aware Code Design with {SAT} Solvers},
  author = {DalFavero, Ben and Watkins, William M. and LaBorde, Margarite L. and Russo, Vincent and Egger, Ethan and Quiroz, Gregory and LaRose, Ryan},
  year = {2026},
  eprint = {2608.23460},
  archiveprefix = {arXiv},
  primaryclass = {quant-ph},
  note = {\href{https://arxiv.org/abs/2608.23460}{arXiv:2608.23460}}
}

@misc{jacintoCompactFactories2026,
  author = {Jacinto, Hugo and Valcarce, Xavier and Barizien, Victor and Gouzien, {\'E}lie and Sangouard, Nicolas},
  title = {Exploring the landscape of compact magic-state distillation factories},
  year = {2026},
  note = {\href{https://arxiv.org/abs/2606.07734}{arXiv:2606.07734}}
}

@article{kasamiTokura1970,
  author = {Kasami, Tadao and Tokura, Nobuki},
  title = {On the weight structure of {Reed--Muller} codes},
  journal = {IEEE Transactions on Information Theory},
  volume = {16},
  number = {6},
  pages = {752--759},
  year = {1970},
  doi = {10.1109/TIT.1970.1054545}
}

@article{kasamiTokuraAzumi1976,
  author = {Kasami, Tadao and Tokura, Nobuki and Azumi, Saburo},
  title = {On the weight enumeration of weights less than $2.5d$ of {Reed--Muller} codes},
  journal = {Information and Control},
  volume = {30},
  number = {4},
  pages = {380--395},
  year = {1976},
  doi = {10.1016/S0019-9958(76)90355-7}
}

@article{sloaneBerlekamp1970,
  author = {Sloane, N. J. A. and Berlekamp, E. R.},
  title = {Weight enumerator for second-order {Reed--Muller} codes},
  journal = {IEEE Transactions on Information Theory},
  volume = {16},
  number = {6},
  pages = {745--751},
  year = {1970},
  doi = {10.1109/TIT.1970.1054553}
}

@book{assmusKeyDesigns1992,
  author = {Assmus, Jr., E. F. and Key, J. D.},
  title = {Designs and Their Codes},
  series = {Cambridge Tracts in Mathematics},
  number = {103},
  publisher = {Cambridge University Press},
  year = {1992},
  doi = {10.1017/CBO9781316529836}
}

@book{carletBooleanFunctions2021,
  author = {Carlet, Claude},
  title = {Boolean Functions for Cryptography and Coding Theory},
  publisher = {Cambridge University Press},
  year = {2021},
  doi = {10.1017/9781108606806}
}

@misc{gongMagicStateBinaryExtension2026,
  author = {Gong, Anqi and Pattison, Christopher A. and Rall, Patrick and Wills, Adam},
  title = {Magic State Distillation via Codes over Binary Extension Fields},
  year = {2026},
  eprint = {2608.09727},
  archiveprefix = {arXiv},
  note = {\href{https://arxiv.org/abs/2608.09727}{arXiv:2608.09727}},
  doi = {10.48550/arXiv.2608.09727}
}

@book{huffmanFundamentals2003,
  author = {Huffman, W. Cary and Pless, Vera},
  title = {Fundamentals of Error-Correcting Codes},
  publisher = {Cambridge University Press},
  year = {2003},
  doi = {10.1017/CBO9780511807077}
}

@article{cohenHelminck1988,
  author = {Cohen, Arjeh M. and Helminck, Aloysius G.},
  title = {Trilinear Alternating Forms on a Vector Space of Dimension 7},
  journal = {Communications in Algebra},
  volume = {16},
  number = {1},
  pages = {1--25},
  year = {1988},
  doi = {10.1080/00927878808823558}
}

@misc{braekenBooleanClassification2004,
  author = {Braeken, An and Borissov, Yuri and Nikova, Svetla and Preneel, Bart},
  title = {Classification of {Boolean} Functions of 6 Variables or Less with Respect to Cryptographic Properties},
  year = {2004},
  howpublished = {Cryptology ePrint Archive, Report 2004/248},
  url = {https://eprint.iacr.org/2004/248},
  note = {Conference version: \emph{Automata, Languages and Programming}, LNCS 3580, pp.~324--334 (2005), \doi{10.1007/11523468_27}}
}

\end{document}